\documentclass[12pt,reqno]{amsart}

\usepackage{graphicx}
\usepackage{amsmath}
\usepackage{float}

\usepackage{amssymb}

\usepackage{amsthm}

\usepackage[foot]{amsaddr}

\usepackage{color}

\usepackage{url}
\usepackage{hyperref}
\hypersetup{
colorlinks=true, %colorise les liens
breaklinks=true, %permet le retour ??? la ligne dans les liens trop longs
urlcolor= blue, %couleur des hyperliens
linkcolor= red, %couleur des liens internes
citecolor=blue, %couleur des r???f???rences
pdftitle={}, %informations apparaissant dans
pdfauthor={}, %les informations du document
pdfsubject={} %sous Acrobat.
}

\numberwithin{equation}{section}

\newtheorem{thm}{Theorem}[section]

\newtheorem{prop}{Proposition}[section]

\newtheorem{lem}{Lemma}[section]

\newtheorem{coro}{Corollary}[section]

\newtheorem{rem}{Remark}[section]

\newtheorem{definition}{Definition}[section]

\newcommand{\bR}{\mathbb{R}}

\newcommand{\bE}{\mathbb{E}}

\newcommand{\bP}{\mathbb{P}}

\def\<{\left<}\def\>{\right>}
\def\({\left(}\def\){\right)}

\newcommand{\dd}{\ensuremath{\operatorname{d}\! }}

\newcommand{\ds}{\ensuremath{\operatorname{d}\! s}}

\begin{document}

\title[De Finetti's problem with fixed transaction costs and regime switching]{De Finetti's problem with fixed transaction costs\\ and regime switching}
 
\author{Wenyuan Wang$^1$}
\address{$^1$School of Mathematics and Statistics, Fujian Normal University, Fuzhou, Fujian, 350117, People's Republic of China; School of Mathematical Sciences, Xiamen University, Fujian 361005, People's Republic of China. Email: \url{wwywang@xmu.edu.cn}}
\author{Zuo Quan Xu$^2$}
\address{$^2$Department of Applied Mathematics, The Hong Kong Polytechnic University, Hong Kong, China. Email: \url{maxu@polyu.edu.hk}}
\author{Kazutoshi Yamazaki$^3$}
\address{$^3$The University of Queensland, School of Mathematics and Physics, St Lucia, Brisbane, 4072, QLD, Australia. Email: \url{k.yamazaki@uq.edu.au}}
\author{Kaixin Yan$^4$}
\address{$^4$School of Mathematical Sciences, Xiamen University, Xiamen, Fujian, 361005, P.R. China. Email: \url{kaixinyan@stu.xmu.edu.cn}}
\author{Xiaowen Zhou$^5$}
\address{$^5$Department of Mathematics and Statistics, Concordia University, Montreal, Quebec, Canada H3G 1M8. Email: \url{xiaowen.zhou@concordia.ca}} 


\begin{abstract} 
In this paper, we examine a modified version of de Finetti's optimal dividend problem, incorporating fixed transaction costs and altering the surplus process by introducing two-valued drift and two-valued volatility coefficients. This modification aims to capture the transitions or adjustments in the company's financial status. We identify the optimal dividend strategy, which maximizes the expected total net dividend payments (after accounting for transaction costs) until ruin, as a two-barrier impulsive dividend strategy. Notably, the optimal strategy can be explicitly determined for almost all scenarios involving different drifts and volatility coefficients. Our primary focus is on exploring how changes in drift and volatility coefficients influence the optimal dividend strategy.
\end{abstract}

\maketitle

\section{Introduction }

De Finetti's optimal dividend problem is a classic stochastic control problem that seeks to determine the optimal timing for paying dividends to maximize the total expected dividends until the point of ruin. Due to discounting, dividends should be paid as soon as possible. However, these decisions must be made carefully to avoid increasing the risk of ruin. Despite its long history since de Finetti's original work \cite{deFinetti}, research in this area remains active at the intersection of control theory and financial/actuarial mathematics. Recent advances in the theory of stochastic processes and stochastic control have enabled the development of more realistic models and their solutions. Various stochastic processes, both Gaussian and non-Gaussian, have been employed as alternatives to the classical Brownian motion and Cramér-Lundberg processes. For a comprehensive review, we refer to \cite{AT09} and the references therein.

In this paper, we examine a diffusion model described by the stochastic differential equation (SDE) \eqref{def.X}, where the drift and volatility each take on two different values depending on the state. In the stochastic control literature, SDE \eqref{def.X} is associated with what is known as bang-bang control of diffusion. In \cite{Benes1980}, the optimally controlled process for the bounded velocity follower problem is the solution to \eqref{def.X} with $\mu_+=\mu_-$. In \cite{McNamara}, the optimal state equation is given by \eqref{def.X}. The solution to \eqref{def.X} with $\sigma_+=\sigma_-$ is proposed in \cite{Gerber} as a refracted risk model. Additionally, SDE \eqref{def.X} is used in local volatility models in mathematical finance; see, for example, \cite{Gairat}.

As a prototype for SDEs with discontinuous coefficients, it 
is also interesting to investigate various properties of the solution to (\ref{def.X}).
The transition density of solution $X$ to (\ref{def.X}) with $\sigma_+=\sigma_-$ is found in \cite{Karatzas1984} and applied to compute the optimal expected costs in the classical control problem treated in \cite{Benes1980}. The above transition density is also used in \cite{Chen22} to identify limiting distribution arising from the central limit theorem, under nonlinear expectation, for random variables with the same conditional variance but ambiguous means. Similarly, the transition density of solution $X$ to (\ref{def.X}) with $\mu_+=0=\mu_-$, called oscillating Brownian motion, is found in \cite{Keilson1978}, and used in \cite{Chen23} for computations concerning the central limit theorem, again under nonlinear expectation, for random variables with mean $0$ and varying conditional variance. More recently, an explicit expression for transition density of process $X$ in (\ref{def.X}) is obtained in \cite{Chen} using solution to the exit problem together with a perturbation approach, which can be applied to express the value function for the control problem in \cite{McNamara}.

We focus on a Brownian motion-driven surplus process, as defined in (\ref{def.X}), where the drift and diffusion coefficients vary depending on whether the process is above or below a fixed threshold. In most of the literature, spatially homogeneous processes, such as Brownian motion and Lévy processes, are used, often resulting in the optimality of a simple barrier-type strategy. In contrast, models involving processes with dynamics dependent on their current values are limited and typically yield non-analytical results. Despite the straightforward dynamics of our surplus process, it has practical applications: it offers a way to model ``regime-switching", where the regime changes depending on whether the surplus process is above or below the threshold. This approach is particularly suitable for modeling non-stationary premium rates and volatility that can vary based on the company's financial status. Our regime-switching model differs from the classical regime-switching models in the literature, as it is caused endogenously, while the latter is driven by exogenous factors due to transitions or adjustments in the economic system; see \cite{AngBeK02b}, \cite{AngTim}, \cite{Pelle} and \cite{So}.
Works studying exogenous regime-switching involved optimal dividend problems can be found in \cite{Azcue15}, \cite{JiaPis2012} and \cite{WeiWY}. A work that studies an endogenous regime-switching involved optimal dividend problem appears in \cite{WY21}.

Another extension we consider is the inclusion of fixed transaction costs, which transforms de Finetti's optimal dividend problem (from a regular/singular control problem) into an impulsive control problem. This extension makes the problem more practical but at the same time significantly more challenging. Typically, unlike the barrier strategy used in the absence of fixed costs, the objective becomes demonstrating the optimality of two-barrier strategies,
which we call $(z_1, z_2)$-strategy (also known as $(s, S)$-policy in the inventory control literature). According to such a strategy in the insurance/financial context, a dividend is paid immediately after the surplus reaches above the upper level $z_2$.
In the insurance/financial context, this strategy involves paying a dividend immediately after the surplus exceeds the upper level $z_2$, reducing the surplus to $z_1$. Identifying these two barriers and proving that the value function satisfies the associated quasi-variational inequality (QVI) is mathematically challenging. Although impulsive control is popular in inventory control problems for infinite-time horizon scenarios, the inclusion of fixed costs is relatively rare in de Finetti's problem, which is terminated at the time of ruin. 

Among these works, several different uncontrolled state processes have been considered: in \cite{Cade06} the surplus is governed by a Brownian motion with drift; in \cite{Bai10} and \cite{Paulsen} the income process follows the dynamics of a general diffusion process;
in \cite{BaiGuo} 
the Crem\'er-Lunderberg risk process is considered; in \cite{Hunting} the surplus process is a jump diffusion; 
and, in \cite{BK14} the surplus process is a spectrally positive L\'evy process. 

It is important to note that, to the best of the authors' knowledge, all existing contributions on de Finetti’s optimal dividend problem utilize spatially homogeneous processes or diffusion processes with regular drift and volatility coefficients as their uncontrolled reserve processes. The combination of non-regular drift and volatility coefficients with fixed transaction costs makes the problem particularly complex and demanding.
In such a setting, the optimal selections of $z_1$ and $z_2$ are interdependent and, in this case, also influenced by the drift/diffusion change-trigger barrier $a$. There are scenarios where both $z_1$ and $z_2$ are either above or below $a$, as well as cases where $z_2$ is above $a$ while $z_1$ is below $a$. Each scenario requires a different analysis and approach. To tackle this problem, we first solve the exit problem for \eqref{def.X} using a martingale approach, which allows us to derive explicit expressions for the expected dividend function under each $(z_1, z_2)$-strategy. Subsequently, we establish sufficient conditions for optimality. This is followed by a case-by-case analysis, demonstrating that by appropriately selecting the barriers, the candidate value function solves the quasi-variational inequality (QVI). As a result, we obtain explicit optimal strategies for all different parameter choices of $\mu_{\pm}$ and $\sigma_{\pm}$ in the model, which aids in better understanding and analyzing the connections between the optimal strategies and these parameters.

 The contributions of this paper can be summarized from two perspectives: one pertains to the problem formulation, and the other to the solution methodology.
We subsequently provide a detailed elaboration of these contributions.

From a problem formulation perspective, our main contribution is to solve, for the first time, an optimal dividend problem by simultaneously incorporating two critical features: (1) an endogenous state-triggered
switching mechanism, and (2) fixed transaction costs, thereby framing the problem as a challenging impulse control problem. While each feature has been studied in isolation, their combination has not been previously addressed. On the one hand, most existing studies model regime switching exogenously using a Markov chain (e.g., \cite{Azcue15}, \cite{JiaPis2012}, \cite{Pelle}, \cite{WeiWY}). A recent work by \cite{WY21} considers endogenous regime switching. However, none of these works incorporate fixed transaction costs. On the other hand, many works include fixed costs (e.g., \cite{BaiGuo}, \cite{BK14}, \cite{Cade06}), but their models are spatially homogeneous and exclude regime switching. Others (\cite{Bai10}, \cite{Paulsen}, \cite{Hunting}) allow state-dependent coefficients but impose strong regularity assumptions (e.g., Lipschitz continuity, differentiability). 
In contrast, our model allows for discontinuous drift and volatility coefficients, which change abruptly at a threshold and take distinct values on either side.

From a methodological perspective, existing methods for dividend problems with regime switching (e.g., \cite{JiaPis2012}, \cite{WY21}, \cite{WeiWY}) rely on regular or singular control and are not applicable to our impulse control setting with fixed transaction costs. Similarly, methods for fixed-cost problems (e.g., \cite{BaiGuo}, \cite{Bai10}, \cite{BK14}, \cite{Cade06}, \cite{Hunting}, \cite{Paulsen}) generally assume either spatial homogeneity, constant coefficients, or regular (smooth) state-dependence coefficients, none of which hold in our framework of discontinuous coefficients. 
To address these challenges, we develop a novel analytical methodology for solving the associated QVI, which includes the following four technical contributions:

First, through extensive technical analysis, we provide a complete and unified characterization of the piecewise convexity/concavity of the function $g$ (Propositions \ref{P.0.3}–\ref{prop.4.5}), identifying 11 mutually exclusive and collectively exhaustive cases (with 21 sub-cases) that fully describe the geometry of $g$ and facilitate the explicit construction of candidate optimal strategies.

Second, we introduce a novel construction-based method to characterize the set of candidate optimal impulse dividend strategies $\mathcal{M}_{\zeta}$ (Theorems \ref{thm2.1}–\ref{thm4.3}), yielding closed-form expressions for all relevant quantities and enabling rigorous sensitivity analysis and numerical implementation.

Third, in the proof of Theorem \ref{thm3.3.vr.}, we overcome the challenges posed by discontinuous coefficients by constructing new inequalities (see \eqref{30.add.new.x} and \eqref{mu-+}, etc.) and proposing three calibrated sufficient (of which the union of the first two is also necessary) conditions to verify optimality. 

Finally, to establish the monotonicity result in Proposition \ref{prop.2.5.add}, we employ a novel re-parameterization technique, transforming the analysis from the ``cost (i.e., $\beta$)" space to a more tractable ``slope of the value function" space. This idea may offer useful insights for other control problems.

The rest of the paper is organized as follows. Section \ref{sec2} formulates the problem, provides preliminary results, and introduces the two-barrier $(z_{1},z_{2})$-impulsive dividend strategy. In Section \ref{sec:negative}, we present a complete and explicit characterization of the optimal strategy among the class of $(z_{1},z_{2})$-impulsive dividend strategies. Section \ref{sec:main} is devoted to characterizing the optimal strategy to the targeted dividend control problem.
Some lengthy and technical proofs are provided in Appendix \ref{sec.4}.

\setcounter{section}{1}

\medskip
\section{Problem formulation and preliminary results} \label{sec2}
\setcounter{section}{2}

We fix a complete filtered probability space $(\Omega,\mathcal{F},\mathbb{F},\mathbb{P})$ throughout the paper, where $\mathbb{F}=(\mathcal{F}_t)_{t\geq0}$ is the filtration generated by a standard one-dimensional Brownian motion $B=({B}_t)_{t\geq0}$ defined in the space and satisfies the usual conditions. 
Fix constants $a\in(0,\infty)$, $\mu_{\pm}\in\bR$, and assume without loss of generality that $\sigma_{\pm}\in(0,\infty)$. 

Consider a stochastic differential equation (SDE):
\begin{align}
\label{def.X}
\dd X_t=(\sigma_+\mathbf{1}_{\{X_t>a\}}+\sigma_-\mathbf{1}_{\{X_t\leq a\}})\dd B_t+(\mu_+ \mathbf{1}_{\{X_t>a\}}+\mu_- \mathbf{1}_{\{X_t\leq a\}})\dd t,\quad t\geq0.
\end{align}
The existence and uniqueness of a strong solution to SDE \eqref{def.X} are guaranteed by Theorem 1.3 in Page 55 of \cite{LeGall84}. We use it to describe the surplus process of a company before paying dividend. 
{Since we are interested in the process until it ruins, we assume $a>0$, for otherwise, the process before its ruin time follows the classical drifted Brownian motion model.}
Clearly, one can deal with the case $a\leq0$ by shifting the process $X$ properly, so our model covers all values of $a\in\bR$. Of course, the corresponding ruin time shall be redefined. 

We will investigate an optimal impulsive dividend payout problem. To this end, we first introduce impulsive dividend payout strategies. An impulsive dividend strategy $\pi=(L^\pi_t)_{t\geq 0}$ is an $\mathbb{F}$-adapted non-decreasing, right continuous pure jump process such that $L^\pi_t=\sum_{0\leq s\leq t} \Delta L^\pi_s$ where $\Delta L_s^{\pi}=L_s^{\pi}-L_{s-}^{\pi}\geq 0$ with $L_{0-}^{\pi}=0$.

Applying an impulsive dividend payout strategy $\pi$ to the process \eqref{def.X}, we see the surplus process $U^{\pi}$ after paying dividend becomes
\begin{align}
\label{def.contr.proc.}
\dd U^\pi_t=(\sigma_+\mathbf{1}_{\{U^\pi_t>a\}}+\sigma_-\mathbf{1}_{\{U^\pi_t\leq a\}})\dd B_t+(\mu_+ \mathbf{1}_{\{U^\pi_t>a\}}+\mu_- \mathbf{1}_{\{U^\pi_t\leq a\}})\dd t-\dd L^\pi_t.
\end{align}
Define the ruin time of $U^\pi$ as $$T^{\pi}:=\inf\{t\geq0\colon U_t^{\pi}<0\},$$
where $\inf\emptyset=\infty$.
We fix a constant $q\in(0,\infty)$ to represent the discount rate. 
\begin{definition}
An impulsive dividend payout strategy $\pi=(L^\pi_t)_{t\geq 0}$ is called admissible if $$0\leq \Delta L_{t}^{\pi}=L_{t}^{\pi}-L_{t-}^{\pi} \leq U_{t-}^{\pi}\vee0$$ for any $t\geq0$ (i.e., the amount of a lump sum of dividend payout is not allowed to make the company ruin),
$$\mathbb{E}_x\Bigg[\sum_{0\leq s\leq T^{\pi}}e^{-qs}\Delta L^{\pi}_s\mathbf{1}_{\{\Delta L^{\pi}_s>0\}}\Bigg]<\infty,$$
and SDE \eqref{def.contr.proc.} with any initial value $U^{\pi}_{0-}=x$ admits a unique strong solution. We use $\Pi$ to denote the set of all admissible impulsive dividend payout strategies.
\end{definition}

Let $\beta$ be a positive constant, which can be interpreted as a fixed transaction costs or penalty parameter. The reward function for an admissible impulsive dividend payout strategy $\pi\in\Pi$ is defined as
\begin{align}
\label{value_func}
V_\pi(x)&:=\mathbb{E}_x\Big[\sum_{0\leq s\leq T^{\pi}}e^{-qs}(\Delta L_s^{\pi}-\beta)\mathbf{1}_{\{\Delta L^\pi_s>0\}}\Big], \quad x\geq 0.
\end{align}
Our aim is to determine the optimal value function associated with the impulsive dividend control problem \eqref{value_func}: 
\[\sup_{\pi\in\Pi}V_\pi(x).\] 
An impulsive strategy $\pi^{*}\in\Pi$ is called an optimal impulsive strategy to the problem \eqref{value_func} if it satisfies
\begin{align}\label{6}
V_{\pi^{*}}(x)=\sup_{\pi\in\Pi}V_\pi(x)<\infty.
\end{align}

Clearly, when $x\leq 0$, $T^{\pi}=0$ for any admissible strategy $\pi$, so $\sup_{\pi\in\Pi}V_\pi(x)=0$.
From now on, we only need to study the case $x>0$.

\begin{rem}
By the definition of $V_\pi(x)$, one can easily verify that the value function $\sup_{\pi\in\Pi}V_{\pi}(x)$ is non-increasing and convex with respect to $\beta\in(0,\infty)$ for any fixed $x\in(0,\infty)$. 
\end{rem}

We now provide a lower and an upper bound for $\sup_{\pi\in\Pi}V_{\pi}(x)$ in the following Lemma \ref{lem2.1.add.new.x}, whose proof is presented in Appendix \ref{sec5.1}.

\begin{lem}
\label{lem2.1.add.new.x}
{We have 
\begin{align}
\label{5.add.new.y}
0\leq \sup_{\pi\in\Pi}V_\pi(x)\leq x+\frac{\sqrt{\mu_{+}^{2}
+2q\sigma_{+}^{2}}+\mu_{+}}{2q}
+\frac{\sqrt{\mu_{-}^{2}
+2q\sigma_{-}^{2}}+\mu_{-}}{2q}, \quad x\geq 0.
\end{align}
}
\end{lem}

\subsection{Verification lemma}
We now attempt to characterize the optimal impulsive strategies
to the control problem \eqref{6}. Such a characterization will be given in Lemma \ref{lem2.1}.

The following Lemma \ref{lem2.1} gives a sufficient condition for an admissible impulsive strategy $\hat{\pi}\in\Pi$ to be the optimal dividend strategy for the control problem \eqref{6}. 
Indeed, it is shown that the optimal strategy must belong to a subset of $\Pi$: 
\begin{align}
\label{def.Pi0}
\Pi_0:=\big\{\pi=(L^\pi_t)_{t\geq 0}\in\Pi; \text{ for any } t\geq 0, \triangle L_{t}^{\pi}\geq \beta \text{ if and only if } \triangle L_{t}^{\pi}>0\big\},
\end{align} 
which consists of those admissible strategies in $\Pi$ that only has jump size no less than $\beta$. Intuitively speaking, one shall not pay dividend less than $\beta$, since it will not only give a negative impact on the reward functional but also lead to an earlier ruin time.

Let $\mathcal{A}$ be the infinitesimal generator associated with the process $X$ defined as
$$\mathcal{A}f(x) :=\frac{1}{2}(\sigma^2_+ \mathbf{1}_{\{x>a\}} +\sigma^2_- \mathbf{1}_{\{x\leq a\}})f^{\prime\prime}(x)+(\mu_+ \mathbf{1}_{\{x>a\}} +\mu_- \mathbf{1}_{\{x\leq a\}})f^{\prime}(x),$$
for any function $f$ that is piecewise $C^2$ on $\bR_{+}$.
\begin{lem}[Verification lemma]
\label{lem2.1}
Suppose, for some $\hat{\pi}\in\Pi$, one has 
$V_{\hat{\pi}}$ is piecewise $C^2$ and satisfies $(\mathcal{A}-q)V_{\hat{\pi}}\leq 0$ on
$\bR_{+}$ except finite many points. Assume further that $V_{\hat{\pi}}(x)\geq 0$ and $V_{\hat{\pi}}(x)-V_{\hat{\pi}}(y)\geq x-y-\beta$ for all $x>y\geq 0$. Then 
$\hat{\pi}$ is an optimal impulsive strategy to the problem \eqref{value_func}. Moreover, $\hat{\pi}\in\Pi_{0}$.
\end{lem}

\begin{proof}
For any admissible impulsive strategy $\pi=(L_t^{\pi})_{t\geq0}\in\Pi\setminus\Pi_0$, we define a new admissible impulsive strategy $\pi_0=(L^{\pi_0}_t)_{t\geq 0}\in\Pi_{0}$ where 
$$L_t^{\pi_0}:=\sum_{0\leq s\leq t} \Delta L^\pi_s\mathbf{1}_{\{\Delta L^\pi_s\geq \beta\}}.$$
By definition, it holds that $\triangle L_t^{\pi_0}\leq \triangle L_t^{\pi}$ for all $t\geq0$, so $U_t^{\pi_0}\geq U_t^{\pi}$ for all $t\geq0$ and consequently, $T^{\pi_0}\geq T^{\pi}$.
Therefore, for any $x\in(0,\infty)$,
\begin{align}
V_{\pi_0}(x)
&=\bE_x\bigg[\sum_{0\leq s\leq T^{\pi_0}} e^{-qs}( \Delta L^{\pi}_{s} -\beta)\mathbf{1}_{\{\Delta L^{\pi}_{s}\geq\beta\}} \bigg]
\nonumber\\
&\geq\bE_x\bigg[\sum_{0\leq s\leq T^{\pi}} e^{-qs}(\Delta L^{\pi}_{s} -\beta)\mathbf{1}_{\{\Delta L^{\pi}_{s}\geq\beta\}} \bigg]
\nonumber\\
&>
\bE_x\bigg[\sum_{0\leq s\leq T^{\pi}} e^{-qs}(\Delta L^{\pi}_{s} -\beta)\left(\mathbf{1}_{\{\Delta L^{\pi}_{s}\geq\beta\}}
+\mathbf{1}_{\{0<\Delta L^{\pi}_{s}<\beta\}} \right)\bigg]
\nonumber\\
&
=V_\pi(x).\nonumber
\end{align}
Hence, we only need to prove that $V_{\hat{\pi}}(x)\geq V_\pi(x)$ for any $x> 0$ and admissible strategy $\pi\in\Pi_0$. 

Fix any $x> 0$ and $\pi\in\Pi_0$. Let $\theta_n:=\inf\{t\geq 0: U^\pi_{t}>n\text{ or }U^\pi_{t}<0\}$.	By a version of Ito's formula (see Theorem 4.57 in page 57 of \cite{Jacod03}, or Theorem 70 in Chapter IV of \cite{PL}) we have, for any constant $t>0$,
\begin{equation*}
	\begin{split}
	&\;\quad e^{-q(t\wedge \theta_n\wedge T^\pi)}V_{\hat{\pi}}(U^\pi_{t\wedge \theta_n\wedge T^\pi})\\
	&=V_{\hat{\pi}}(x)
	+\int_0^{t\wedge \theta_n\wedge T^\pi}(\mathcal{A}-q)V_{\hat{\pi}}(U^\pi_{s})\dd s +\sum_{0\leq s\leq t\wedge \theta_n\wedge T^\pi} e^{-qs}\Delta V_{\hat{\pi}}(U^\pi_{s})+M_{t\wedge \theta_n\wedge T^\pi}\\
	&\leq V_{\hat{\pi}}(x)-\sum_{0\leq s\leq t\wedge \theta_n\wedge T^\pi} e^{-qs}( U^\pi_{s-}-U^\pi_{s}-\beta )\mathbf{1}_{\{\Delta U^\pi_{s}\neq 0\}}+M_{t\wedge \theta_n\wedge T^\pi}\\
	&=V_{\hat{\pi}}(x)-\sum_{0\leq s\leq t\wedge \theta_n\wedge T^\pi} e^{-qs}( \Delta L^\pi_{s} -\beta)\mathbf{1}_{\{\Delta L^\pi_{s}>0\}}+M_{t\wedge \theta_n\wedge T^\pi},
	\end{split}
\end{equation*}
where $(M_{t})_{t\geq 0}$ is a continuous local martingale. Since $V_{\hat{\pi}}(y)\geq0$ for any $y\geq 0$, it follows
$$V_{\hat{\pi}}(x)\geq \sum_{0\leq s\leq t\wedge \theta_n\wedge T^{\pi}}e^{-qs}(\Delta L^\pi_{s}-\beta)\mathbf{1}_{\{\Delta L_s^{\pi}>0\}}-M_{t\wedge \theta_n\wedge T^{\pi}}.$$
Let $\tau_m$ be an increasing localizing sequence of stopping times of $M$ such that $\lim_{m\rightarrow\infty}\tau_m=\infty.$ Since $\pi\in\Pi_0$, we have $(\Delta L^\pi_{s}-\beta)\mathbf{1}_{\{\Delta L_s^{\pi}>0\}}\geq0$ for any $s\geq0$. Using Fatou's lemma, we have
\begin{equation*}
\begin{split}
V_{\hat{\pi}}(x) 
&\geq \liminf_{n\to\infty}\bE_x\Big[\sum_{0\leq s\leq n\wedge \tau_n\wedge \theta_n\wedge T^\pi} e^{-qs}( \Delta L^\pi_{s} -\beta)\mathbf{1}_{\{\Delta L^\pi_{s}>0\}} \Big]
\\
&= \bE_x\Big[\sum_{0\leq s\leq T^\pi} e^{-qs}(\Delta L^\pi_{s} -\beta)\mathbf{1}_{\{\Delta L^\pi_{s}>0\}} \Big]
\\
&= V_\pi(x),
	\end{split}
\end{equation*}
which is the desired result. As a byproduct, $\hat{\pi}\in\Pi_{0}$.
\end{proof}

By this result, our problem now reduces to finding an impulsive strategy $\hat{\pi}\in\Pi_0$ (see \eqref{def.Pi0}) that fulfills the requirement of Lemma \ref{lem2.1}.
We conjecture that the optimal strategy solving the optimal control problem \eqref{6} shall be some two-barrier impulsive strategy. To show this, we introduce this kind of dividend payout strategy in the subsequent section.

\subsection{Two-barrier strategies and preliminary results}

For each pair $0<z_1<z_2$, the corresponding two-barrier impulsive dividend strategy, denoted by $(L_t^{z_1,z_2})_{t\geq0}$, is the strategy under which a lump sum of dividends is paid out to bring the surplus process down to the level $z_1$ once the surplus process is greater than or attempts to up-cross the level $z_2$, and no dividend payout happens if the surplus process is below $z_2$. For convenience, we also call the two-barrier impulsive strategy $(L_t^{z_1,z_2})_{t\geq0}$ a $(z_1,z_2)$-strategy.

Mathematically, the $(z_1,z_2)$-strategy and its corresponding surplus process $(U_t^{z_1,z_2})_{t\geq0}$ can be jointly determined by
\begin{align}
\label{def.zz2.Uzz2}
\begin{cases}
\quad L_{t}^{z_1,z_2}=\sum_{0\leq s\leq t}(U_{s-}^{z_1,z_2}-z_1)\mathbf{1}_{\{U_{s-}^{z_1,z_2}\geq z_2\}},\quad t\geq 0,\\
\; \dd U^{z_1,z_2}_t
=
(\mu_+ \mathbf{1}_{\{U^{z_1,z_2}_t>a\}}+\mu_- \mathbf{1}_{\{U^{z_1,z_2}_t\leq a\}})\dd t-(U_{t-}^{z_1,z_2}-z_1)\mathbf{1}_{\{U_{t-}^{z_1,z_2}\geq z_2\}}
\\
\qquad\qquad \quad+(\sigma_+\mathbf{1}_{\{U^{z_1,z_2}_t>a\}}+\sigma_-\mathbf{1}_{\{U^{z_1,z_2}_t\leq a\}})\dd B_t,\quad t\geq 0,
\\
\quad\! U^{z_1,z_2}_{0-}=x.
\end{cases}
\end{align}
Write the ruin time of $U^{z_1,z_2}$ as $$T^{z_1,z_2}:=\inf\{t\geq0:U^{z_1,z_2}_t<0\},$$
and denote the value function of the two-barrier impulsive strategy $(L_t^{z_1,z_2})_{t\geq0}$ by
\[V^{z_2}_{z_1}(x):=
\mathbb{E}_x\Big[\sum_{0\leq s\leq T^{z_1,z_2}}e^{-qs}(\Delta L_s^{z_1,z_2}-\beta)\mathbf{1}_{\{\Delta L^{z_1,z_2}_s>0\}}\Big], ~~ x\geq 0.\]
Thanks to Lemma \ref{lem2.1.add.new.x}, $V^{z_2}_{z_1}(x)$ is finite for any $x\geq0$. 

In the following, we aim to find an explicit expression for $V^{z_2}_{z_1}(x)$ so that we can apply Lemma \ref{lem2.1} to derive an optimal strategy to \eqref{6}. To this end, we first define two functions $g^{\pm}\in C^1(\bR)\cap C^2(\bR\backslash \{a\}) $ that satisfy the following ordinary differential equation (ODE) except at $a$: 
\begin{align}
\label{dis.generator.}
\frac{1}{2}(\sigma^2_+ \mathbf{1}_{\{x>a\}} +\sigma^2_- \mathbf{1}_{\{x\leq a\}})g^{\prime\prime}(x)+(\mu_+ \mathbf{1}_{\{x>a\}} +\mu_- \mathbf{1}_{\{x\leq a\}})g^{\prime}(x)=qg(x).
\end{align}
Put
\begin{align}
\label{theta12}
\theta^{\pm}_1:=\frac{\sqrt{\mu^2_{\pm}+2q\sigma^2_{\pm}}+\mu_{\pm}}{\sigma^2_{\pm}}>0,\quad \theta_2^{\pm}:=\frac{\sqrt{\mu^2_{\pm}+2q\sigma^2_{\pm}}-\mu_{\pm}}{\sigma^2_{\pm}}>0,
\end{align}
\begin{align}
\label{sign.c-}
c_-
=
\frac{\theta_1^--\theta_1^+}{\theta_2^-+\theta_1^-},
\quad
1-c_-=\frac{\theta_2^-+\theta_1^+}{\theta_2^-+\theta_1^-}>0,
\end{align}
and
\begin{align}
\label{sign.c+}
c_+=\frac{\theta_2^+-\theta_2^-}{\theta_2^++\theta_1^+},
\quad
1-c_+=\frac{\theta_1^++\theta_2^-}{\theta_2^++\theta_1^+}>0.
\end{align}
Then, the following two functions $g^{\pm}\in C^1(\bR)\cap C^2(\bR\backslash \{a\})$ satisfy ODE \eqref{dis.generator.} (except at a): 
\begin{align}
\begin{cases}
g^-(x)
=e^{-\theta_1^+{(x-a)}} \mathbf{1}_{\{x>a\}}+\left(c_-e^{\theta_2^-{(x-a)}}+(1-c_-) e^{-\theta_1^-{(x-a)}} \right)\mathbf{1}_{\{x\leq a\}},\\
g^+(x)
=\left((1-c_{+})e^{\theta_2^+{(x-a)}}+c_+ e^{-\theta_1^+{(x-a)}} \right)\mathbf{1}_{\{x>a\}}+e^{\theta_2^-{(x-a)}}\mathbf{1}_{\{x\leq a\}}.
\end{cases}
\label{scale.fun.}
\end{align}
It is easy to check $g^{\pm}(a)=1$ and $g^{\pm\prime}(a-)=g^{\pm\prime}(a+)$.

Define 
\begin{align}
\label{10}
g(x)&:=g^+(x)g^-(0)-g^-(x)g^+(0)\\
&= \left[(1-c_{+})g^-(0)e^{\theta_2^+{(x-a)}}-(g^+(0)-c_+g^-(0))e^{-\theta_1^+{(x-a)}} \right]1_{\{x>a\}}\nonumber\\
&\quad\;+\left[(g^-(0)-c_{-}g^+(0))e^{\theta_2^-{(x-a)}}-(1-c_-)g^+(0)e^{-\theta_1^-{(x-a)}}\right]1_{\{x\leq a\}}.\nonumber
\end{align}
Then, one has $g\in C^1(\bR)\cap C^2(\bR\backslash \{a\})$ and it satisfies ODE \eqref{dis.generator.} (except at a). 
\begin{lem}
\label{prop2.2.v2}
The function $g$ defined by \eqref{10} satisfies $g'>0$ and $g(0)=0$.
\end{lem}
Its proof is provided in Appendix \ref{sec5.2}.

We next introduce the first hitting time of level $y\in\bR$ for the process $(X_t)_{t\geq 0}$ given by \eqref{def.X} as
$$T_y:=\inf\{t\geq0:\,X_t=y\}.$$
For $y\leq x\leq z$ with $y\neq z$, applying the generalized Ito's formula (see Theorem 70 of Chapter IV in \cite{PL} for more details) we know that the two processes
$(e^{-qt}g^{\pm}(X_t))_{t\geq 0}$ are local martingales. Then, it follows from Doob's optional stopping theorem that
\[g^{\pm}(x)=\bE_x [e^{-q(T_y\wedge T_z)}g^{\pm}(X_{T_y\wedge T_z})]
=g^{\pm}(y)\bE_x [e^{-qT_y}\mathbf{1}_{\{T_y<T_z\}}]+g^{\pm}(z)\bE_x [e^{-qT_z}\mathbf{1}_{\{T_z<T_y\}}].\] 
Solving the above two equations, we get the following solutions of two-sided exit problem.
\begin{lem} 
For any $y\leq x\leq z$ with $y\neq z$, we have 
\begin{align}
\bE_{x} [e^{-qT_y}\mathbf{1}_{\{T_y<T_z\}}]=\frac{g^+(z)g^-(x)-g^-(z)g^+(x)}{g^-(y)g^+(z)-g^-(z)g^+(y)},
\end{align}
and
\begin{align}
\bE_{x} [e^{-qT_z}\mathbf{1}_{\{T_z<T_y\}}]=\frac{g^+(y)g^-(x)-g^-(y)g^+(x)}{g^-(z)g^+(y)-g^-(y)g^+(z)}.
\end{align}
\end{lem}
Now we are ready to give the explicit expression for $V^{z_2}_{z_1}$.
\begin{prop}\label{prop.2.1}
Given $\beta\leq z_1+\beta\leq z_2$, we have 
\begin{align}\label{Vx.3}
V^{z_2}_{z_1}(x)=
\begin{cases}
\frac{g(z_1)(z_2-z_1-\beta)}{g(z_2)-g(z_1)}+x-z_1-\beta, & x\geq z_2,\bigskip
\\
\frac{g(x)(z_2-z_1-\beta)}{g(z_2)-g(z_1)}, & 0\leq x< z_2.
\end{cases}
\end{align}
Thanks to Lemma \ref{prop2.2.v2}, $V^{z_2}_{z_1}(\cdot)$ is continuous and strictly increasing on $\bR_{+}$. 
\end{prop}
\begin{proof}
{Recall that $V^{z_2}_{z_1}(x)$ is finite for any $x\geq0$.}
Since both $V^{z_2}_{z_1}(0)$ and $g(0)$ are zero, the claim holds when $x=0$.
When $x\in(0,z_2)$, 
\begin{align*} 
V^{z_2}_{z_1}(x)
&=\bE_x\left[e^{-qT_{z_2}}\mathbf{1}_{\{T_{z_2}<T_0\}}\right]V^{z_2}_{z_1}(z_2)
\nonumber\\
&=\frac{g^+(0)g^-(x)-g^-(0)g^+(x)}{g^+(0)g^-(z_2)-g^-(0)g^+(z_2)}(V^{z_2}_{z_1}(z_1)+z_2-z_1-\beta)\nonumber\\
&=\frac{g(x)}{g(z_2)}(V^{z_2}_{z_1}(z_1)+z_2-z_1-\beta).
\end{align*}
Setting $x=z_1$ in the above equation and using the finiteness of $V_{z_1}^{z_2}$, we have 
\begin{align*} 
V^{z_2}_{z_1}(z_1)&= 
\frac{g(z_1)}{g(z_2)-g(z_1)}(z_2-z_1-\beta).
\end{align*}
Combining above two proves the claim when $x\in(0,z_{2})$. When $x\geq z_2$, by the strong Markov property of the process $(U^{z_1,z_2}_{t})_{t\geq 0}$, we have
$$V^{z_2}_{z_1}(x)=V^{z_2}_{z_1}(z_1)+x-z_1-\beta,$$ and 
the claim follows by combining the above two equations. 
\end{proof}

To address our targeting impulsive dividend control problem \eqref{6}, we conjecture that the optimal strategy is a $(z_{1},z_{2})$-strategy for some $(z_1,z_2)$ satisfying $\beta\leq z_{1}+\beta\leq z_{2}<\infty$. To verify our conjecture, we shall first find the optimal strategy among the class of $(z_{1},z_{2})$-strategies. From the above result, we see that we need to maximize $\frac{z_2-z_1-\beta}{g(z_2)-g(z_1)}$ 
if we want to maximize $V^{z_2}_{z_1}(\cdot)$. This motivates us to define the following. 

Let 
\begin{align*}
\mathcal{D}_{\zeta}&:=\{(x,y)\in[0,\infty)^2: x+\beta\leq y \},\\
\zeta(z_{1},z_{2}) &:=\frac{z_2-z_1-\beta}{g(z_2)-g(z_1)}\geq 0,~~ (z_{1},z_{2})\in \mathcal{D}_{\zeta},
\end{align*}
and
\begin{align}
\label{M}
\mathcal{M}_{\zeta}:=\big\{(z_{1},z_{2})\in\mathcal{D}_{\zeta}: \zeta(z_{1},z_{2})\geq \zeta(x,y) \text{ for all } (x,y)\in \mathcal{D}_{\zeta}\big\}.
\end{align} 
Hence, $\mathcal{M}_{\zeta}$ denotes the set of global maximizers of the function $\zeta(z_{1},z_{2})$ defined on its domain $\mathcal{D}_{\zeta}$.
The following result states that $\mathcal{M}_{\zeta}$ is a non-empty and bounded set, and its proof is provided in Appendix \ref{sec5.11}.

\begin{prop}
\label{P.0.4}
The set $\mathcal{M}_{\zeta}$ is not empty. 
In addition, there exists a finite $z_{0}\in(0,\infty)$ such that
$\mathcal{M}_{\zeta}\subseteq \{(x,y)\in[0,\infty)^2: x+\beta<y< z_{0}\}.$
\end{prop}

\begin{rem}
\label{rem2.1}
For any $(z_{1},z_{2})\in\mathcal{M}_{\zeta}$, by Proposition \ref{P.0.4}, we have $\frac{\partial}{\partial z_{2}}\zeta(z_{1},z_{2})=0$, which is equivalent to
\begin{align}
\label{z2.int.}
g(z_2)-g(z_1)=(z_2-z_1-\beta)g^{\prime}(z_2).
\end{align}
Substituting \eqref{z2.int.} into \eqref{Vx.3} yields
\begin{align}\label{Vx.4}
V^{z_2}_{z_1}(x)=
\begin{cases}
\frac{g(z_2)}{g^{\prime}(z_2)}+x-z_2, & x\geq z_2,\bigskip\\
\frac{g(x)}{g^{\prime}(z_2)}, & 0\leq x< z_2.
\end{cases}
\end{align}
This indicates $V^{z_2}_{z_1}(\cdot) \in C^{1}(\bR_{+})$.
In addition, if $(z_{1},z_{2})\in\mathcal{M}_{\zeta}$ is such that $z_{1}>0$, then $\frac{\partial}{\partial z_{1}}\zeta(z_{1},z_{2})=0$, that is
\begin{align}
\label{z1.int.}
g(z_2)-g(z_1)=(z_2-z_1-\beta)g^{\prime}(z_1).
\end{align}
Thanks to $z_2>z_1$, $g'>0$ by Lemma \ref{prop2.2.v2}, it follows from \eqref{z2.int.} and \eqref{z1.int.} that $g^{\prime}(z_1)=g^{\prime}(z_2)$ for any $(z_{1},z_{2})\in\mathcal{M}_{\zeta}$ with $z_{1}>0$. 
Put
\begin{align}
\label{psi}
\psi(x,y):=\int_x^y\left(1-\frac{g^{\prime}(s)}{g^{\prime}(y)}\right)\ds,\quad x,y\in(0,\infty).
\end{align}
Then, {equation} \eqref{z2.int.} is equivalent to
\begin{align}
\label{psi=beta}
\psi(z_1,z_2)=\beta.
\end{align}
Therefore, we have $\mathcal{M}_{\zeta}\subseteq \mathcal{N}:=\mathcal{N}^{+}\cup \mathcal{N}^{-}$, where 
\begin{align}
\mathcal{N}^{+} &:= \{(z_{1},z_{2}): 0< z_{1}<z_{2}<\infty, \,\psi(z_1,z_2)=\beta, g^{\prime}(z_1)=g^{\prime}(z_2)\},\nonumber\\
\mathcal{N}^{-} &:= \{(0,z_{2}): 0< z_{2}<\infty, \,\psi(0,z_2)=\beta\}.
\end{align}
\end{rem}

\section{Explicit characterization of $\mathcal{M}_{\zeta}$} \label{sec:negative}
\setcounter{section}{3}
This section is devoted to the characterization of the explicit form of the set $\mathcal{M}_{\zeta}$ corresponding to four mutually exclusive and collectively exhaustive cases: (1) $\mu_{\pm}>0$; (2) $\mu_{\pm}\leq 0$; (3) $\mu_{+}\leq 0$ and $\mu_{-}>0$; and (4) $\mu_{+}>0$ and $\mu_{-}\leq 0$. 
We also provide, in Subsection \ref{sec3.5}, several general properties of $\mathcal{M}_{\zeta}$ that can help to better understand and analyze the connections between the optimal dividend strategy and the model parameters.

To proceed, define five constants $x_0$, $\Theta$, $a_{1}$, $a_2$ and $a_{3}$ as
\begin{align}\label{def.x0}
x_0:=\frac{\ln\frac{(g^+(0)-c_+g^-(0))(\theta_1^+)^2}{(1-c_+)g^-(0)(\theta^+_2)^2}}{\theta^+_2+\theta_1^+}+a, \quad \Theta:=c_+(\theta_1^+)^2+(1-c_+)(\theta_2^+)^2,
\end{align}
and
\begin{align}\label{def.a}
a_1:=\frac{2\ln\frac{\theta_1^-}{\theta_2^-}}{\theta_2^-+\theta_1^-},\quad a_2:=\frac{\ln\frac{\theta^+_2+\theta_1^-}{\theta_2^+-\theta_2^-}}{\theta_1^-+\theta_2^-},\quad 
a_3:=\frac{\ln\frac{(1-c_-c_+)(\theta_1^+)^2-(1-c_+)c_-(\theta_2^+)^2}{(1-c_-)\Theta}}{\theta_1^-+\theta_2^-},
\end{align}
whenever they are well-defined (note that $\ln x$ is not well-defined for $x\leq 0$). The constants defined in \eqref{def.x0} and \eqref{def.a} are instrumental in establishing the piecewise convexity (concavity) of $g$ (see, Propositions \ref{P.0.3}-\ref{prop.4.5}), a property that is fundamental to the analysis in this paper. 

\subsection{Explicit characterization of $\mathcal{M}_{\zeta}$ in the case $\mu_{\pm}>0$.}\label{sec3.1}
When $\mu_{\pm}>0$, by considering all scenarios of settings of the parameters, we distinguish the following mutually exclusive and collectively
exhaustive Cases (i)-(iv).

\begin{description}
\item[Case (i)] one of the following conditions holds
\begin{itemize}
\item $0<a_2\leq a\leq a_1$ and $c_+>0$,
\item $0<a_3\leq a\leq a_1\wedge a_2$ and $c_+>0$,
\item $0<a_3\leq a\leq a_1$, $c_+\leq0$ and $\Theta>0$.
\end{itemize}

\item[Case (ii)] one of the following conditions holds
\begin{itemize}
\item $0<a\leq a_1\wedge a_2\wedge a_3$ and $c_+>0$,
\item $0<a\leq a_1$, $c_+\leq0$ and $\Theta\leq0$,
\item $0<a\leq(a_1\wedge a_3)$, $c_+\leq0$ and $\Theta>0$.
\end{itemize}

\item[Case (iii)] one of the following conditions holds
\begin{itemize}
\item $(a_1\vee a_2)\leq a$ and $c_+>0$,
\item $(a_1\vee a_3)\leq a<a_2$ and $c_+>0$,
\item $(a_1\vee a_3)\leq a$, $c_+\leq0$ and $\Theta>0$.
\end{itemize}

\item[Case (iv)] one of the following conditions holds
\begin{itemize}
\item $a_1<a< (a_2\wedge a_3)$ and $c_+>0$,
\item $a_1<a$, $c_+\leq0$ and $\Theta\leq0$,
\item $a_1<a<a_3$, $c_+\leq0$ and $\Theta>0$.
\end{itemize}
\end{description}

With the Cases (i)-(iv) described above and notations given by \eqref{def.x0} and \eqref{def.a}, the following Proposition \ref{P.0.3} gives a complete characterization of the piece-wise concavity or convexity of the function $g(x)$ on $(0,\infty)$. Its proof is provided in Appendix \ref{sec5.3}. 

\begin{prop}
\label{P.0.3} Suppose that $\mu_{\pm}>0$. Under Case (i), $g(x)$ is concave on $(0,a)$ and convex on $(a,\infty)$. Under Case (ii), $g(x)$ is concave on $(0,x_0)$ and convex on $(x_0,\infty)$. Under Case (iii), $g(x)$ is concave on $(0,a_1)$ and convex on $(a_1,\infty)$. Under Case (iv), $g(x)$ is concave on $(0,a_1)$, convex on $(a_1,a)$, concave on $(a,x_0)$ and convex on $(x_0,\infty)$.
\end{prop}

Recall that, under Case (i), $g^{\prime}(x)$ is continuous, strictly decreasing and continuously differentiable on $(0,a)$ and strictly increasing and continuously differentiable on $(a,\infty)$ {by Proposition \ref{P.0.3}}. {Let $(g^{\prime})_{-}^{-1}(x):=\inf\{z\in[0,a];g'(z)\leq x\}$ for $x\in[g'(a),\infty)$ and $(g^{\prime})_{+}^{-1}(x)$ be the inverse function of $[a,\infty)\ni x\mapsto g^{\prime}(x)\in[g^{\prime}(a),\infty)$.} Define $a_4:=(g')^{-1}_+(g'(0))$. Further, define the unary function $\phi$ as
\begin{align}
\label{def.phi.}
\phi(x) &:= \psi((g^{\prime})_-^{-1}(g^{\prime}(x)),x)= 
\int_{(g^{\prime})_-^{-1}(g^{\prime}(x))}^{x}\left(1-\frac{g^{\prime}(s)}{g^{\prime}(x)}\right)\ds, \quad x\in[a,\infty).
\end{align}

Under Case (ii), $g^{\prime}(x)$ is continuous, strictly decreasing and continuously differentiable on $(0,x_0)$ and strictly increasing and continuously differentiable on $(x_0,\infty)$ {by Proposition \ref{P.0.3}}. Define inverse function $(\Bar{g}^{\prime})_{-}^{-1}$ (resp., $(\Bar{g}^{\prime})_+^{-1}$) the same as $(g^{\prime})_{-}^{-1}$ (resp., $(g^{\prime})_{+}^{-1}$) but with $x_0$ in place of $a$.
We further define a unary function $\Bar{\phi}$ the same as \eqref{def.phi.} but with $(g^{\prime})_{-}^{-1}$ replaced by $(\Bar{g}^{\prime})_{-}^{-1}$ and $x_0$ in place of $a$.
Denote by $\phi^{-1}$ and $\Bar{\phi}^{-1}$ the inverse functions of $\phi$ and $\Bar{\phi}$, respectively. 
The well-definedness of these three inverse functions will be confirmed in the proof of the upcoming Theorem \ref{thm2.1}, which 
explicitly characterizes the set $\mathcal{M}_{\zeta}$.
 
Under Case (iii), $g^{\prime}(x)$ is continuous, strictly decreasing and continuously differentiable on $(0,a_1)$ and strictly increasing and continuously differentiable on $(a_1,\infty)$. Let $(\Tilde{g}^{\prime})_{-}^{-1}$ (resp., $(\Tilde{g}^{\prime})_{+}^{-1}$) be defined the same as $(g^{\prime})_{-}^{-1}$ (resp., $(g^{\prime})_{+}^{-1}$) but with $a_1$ in place of $a$. Define further a unary function $\Tilde{\phi}$ in the same manner as \eqref{def.phi.} but with $(g^{\prime})^{-1}$ replaced by $(\Tilde{g}^{\prime})^{-1}$ and $a_1$ in place of $a$. The inverse functions of $\Tilde{\phi}$ is denoted as $\Tilde{\phi}^{-1}$.

Under Case (iv), $g^{\prime}(x)$ is continuous, strictly decreasing (resp., increasing) and continuously differentiable on $(0,a_1)$ and $(a,x_0)$ (resp., $(a_1,a)$ and $(x_0,\infty)$). Let $(g^{\prime})_1^{-1}$
be defined the same as $(g^{\prime})_{-}^{-1}$ but with $a_1$ in place of $a$; 
$(g^{\prime})_2^{-1}$ be the inverse function of $[a_1,a]\ni x\mapsto g^{\prime}(x)\in [g^{\prime}(a_1),g^{\prime}(a)]$; $(g^{\prime})_3^{-1}$ be the inverse function of $[a,x_0]\ni x\mapsto g^{\prime}(x)\in[g^{\prime}(x_0),g^{\prime}(a)]$; and, $(g^{\prime})_4^{-1}$ be the inverse function of $[x_0,\infty)\ni x\mapsto g^{\prime}(x)\in[g^{\prime}(x_0),\infty)$. Denote $a_5:=\inf\{x\geq x_0; g^{\prime}(x)\geq g^{\prime}(a_1)\}$ and $a_6:=(g^{\prime})^{-1}_4(g^{\prime}(a))$. 
Furthermore, put
\begin{align}
\label{x1.def}
x_1&:=\inf\bigg\{x\in[a_5,a_{6}]: \int_{(g^{\prime})_1^{-1}(g^{\prime}(x))}^{(g^{\prime})_3^{-1}(g^{\prime}(x))}\left(1-\frac{g^{\prime}(s)}{g^{\prime}(x)}\right)\ds\geq 0\bigg\},
\\
\label{x2.def}
x_2&:=\inf\bigg\{x\in[a_5,a_{6}]: \int_{(g^{\prime})_2^{-1}(g^{\prime}(x))}^x\left(1-\frac{g^{\prime}(s)}{g^{\prime}(x)}\right)\ds\geq 0\bigg\},
\\
\label{addnewomega1}
\omega_1(x)&:=
\int_{(g^{\prime})_1^{-1}(g^{\prime}(x))}^x\left(1-\frac{g^{\prime}(s)}{g^{\prime}(x)}\right)\ds,\quad x\in[a_1,(g^{\prime})_{2}^{-1}( g^{\prime}(x_2))]\cup [x_2,\infty),
\\
\label{addnewomega2}
\omega_{2}(x)&:=
\begin{cases}
\int_{(g^{\prime})_3^{-1}(g^{\prime}(x))}^x\left(1-\frac{g^{\prime}(s)}{g^{\prime}(x)}\right)\ds,& x\in[x_0,x_1),
\\
\int_{(g^{\prime})_1^{-1}(g^{\prime}(x))}^x\left(1-\frac{g^{\prime}(s)}{g^{\prime}(x)}\right)\ds, & x\in [x_1,\infty).
\end{cases}
\end{align}
With the above notations, we are now ready to provide an explicit characterization of $\mathcal{M}_{\zeta}$ in the following Theorem \ref{thm2.1}, whose proof will be presented in Appendix \ref{sec5.4}.
\begin{thm} \label{thm2.1}
Suppose that $\mu_{\pm}>0$.
\begin{itemize}
\item
Under Case (i), we have
$\mathcal{M}_{\zeta}=\{((g^{\prime})_-^{-1}(g^{\prime}(\phi^{-1}(\beta))),\phi^{-1}(\beta))\}.$
\item Under Case (ii), we have
$\mathcal{M}_{\zeta}=\{((\Bar{g}^{\prime})_-^{-1}(g^{\prime}(\Bar{\phi}^{-1}(\beta))),\Bar{\phi}^{-1}(\beta))\}.$ 
\item Under Case (iii), we have
$\mathcal{M}_{\zeta}=\{((\Tilde{g}^{\prime})_-^{-1}(g^{\prime}(\Tilde{\phi}^{-1}(\beta))),\Tilde{\phi}^{-1}(\beta))\}.$ 
\item Under Case (iv), we have
\begin{equation*}
\mathcal{M}_{\zeta}=
\begin{cases}
\{
(\tilde{z}_1,\tilde{z}_2)\}, & \text{if $\beta\in A_{1}\cup A_{3}$,
}\\
\{(\Bar{z}_1,\Bar{z}_2)
\}, & \text{if $\beta\in A_{2}\cap \overline{A}_3$,}
\\
\{(\tilde{z}_1,\tilde{z}_2)\}\cup\{(\bar{z}_1,\bar{z}_2)
\}, &\text{if $\beta\in \overline{A_{1}\cup A_{2}\cup A_{3}}$,}
\end{cases}
\end{equation*}
where $(\tilde{z}_1,\tilde{z}_2):=((g^{\prime})_{1}^{-1}(g^{\prime}(\omega_{1}^{-1}(\beta))),\omega_{1}^{-1}(\beta))$, $(\Bar{z}_1,\Bar{z}_2):=((g^{\prime})_{3}^{-1}(g^{\prime}(\omega_{2}^{-1}(\beta))),\omega_{2}^{-1}(\beta))$, $A_{1}:=\{\beta>0: g^{\prime}(\omega_1^{-1}(\beta))<g^{\prime}(\omega_{2}^{-1}(\beta))\}$, $A_{2}:=\{\beta>0: g^{\prime}(\omega_1^{-1}(\beta))>g^{\prime}(\omega_{2}^{-1}(\beta))\}$, $A_{3}:=(\omega_{2}(x_1),\infty)$, 
and $\overline{A}:=(0,\infty)\setminus A$. It holds that $ A_{2}\cap A_3=\emptyset$. We mention that, when $\beta=\omega_{1}(x_2)$, one has $\omega_{1}^{-1}(\beta)=\{(g^{\prime})_{2}^{-1}( g^{\prime}(x_2)),x_2\}$; in this case, $\{(\tilde{z}_1,\tilde{z}_2)\}$ is understood as $\{(\tilde{z}_1,(g^{\prime})_{2}^{-1}( g^{\prime}(x_2))),(\tilde{z}_1,x_2)\}$.
\end{itemize}
\end{thm}

\subsection{Explicit characterization of $\mathcal{M}_{\zeta}$ in the case $\mu_{\pm}\leq 0$.}

In this subsection, we discuss the case of $\mu_{\pm}\leq0$. We intend to offer merely the main results while most of their proofs are omitted because they require no new techniques in comparison to that of Subsection \ref{sec3.1}.

When $\mu_{\pm}\leq0$, one can conclude that 
{$\theta^{-}_1\leq \theta_2^{-}$ and $\theta^{+}_1\leq \theta_2^{+}$} using \eqref{theta12}. Recall that the function $g(x)$ defined by \eqref{10} is strictly increasing with $g(0)=0$.
The following Proposition \ref{prop.4.1} gives the convexity of the function $g(x)$ on $(0,\infty)$. The proof is deferred to Appendix \ref{sec5.5}.

\begin{prop}\label{prop.4.1}
Suppose that $\mu_{\pm}\leq0$. The function $g(x)$ is convex on $(0,\infty)$.
\end{prop}

The following Theorem \ref{thm4.1} explicitly characterizes the set $\mathcal{M}_{\zeta}$ defined by \eqref{M} as a singleton. 
The proof of Theorem \ref{thm4.1} is similar to that of Theorem \ref{thm2.1} and hence omitted. 

\begin{thm}
\label{thm4.1}
Suppose that $\mu_{\pm}\leq0$. The set $\mathcal{M}_{\zeta}$ is a singleton set as
$$\mathcal{M}_{\zeta}=\{(0,\phi_0^{-1}(\beta))\},$$ where the function $\phi^{-1}_0$ is the well-defined inverse function of $\phi_0$ given by 
\begin{align}\label{def.phi0}
\phi_0(x):=\psi(0,x)
,\quad x\in[0,\infty).
\end{align}
\end{thm}

\subsection{Explicit characterization of $\mathcal{M}_{\zeta}$ in the case $\mu_{+}\leq 0$ and $\mu_{-}>0$.}

We next consider the case $\mu_+\leq0$ and $\mu_->0$, under which one can conclude $\theta_1^+\leq\theta_2^+$ and $\theta_1^->\theta_2^-$ using \eqref{theta12}. The following Proposition \ref{prop.4.3}, corresponding to all scenarios of settings of the parameters, gives a complete characterization of the piece-wise concavity or convexity of the function $g(x)$ on $(0,\infty)$. The proof is deferred to Appendix \ref{sec5.6}.
\begin{prop}\label{prop.4.3}
Suppose that $\mu_+\leq 0$ and $\mu_->0$. Then the function $g(x)$ is concave on $(0,a_{1}\wedge a)$ and convex on $(a_{1}\wedge a,\infty)$.
\end{prop}

The following Theorem \ref{thm4.2} explicitly characterizes the set $\mathcal{M}_{\zeta}$ as a singleton in the case $\mu_{+}\leq0$ and $\mu_->0$. We do not provide a proof for Theorem \ref{thm4.2} due to its similarity to that of Theorem \ref{thm2.1}.

\begin{thm}
\label{thm4.2}
Suppose that $\mu_{+}\leq0$ and $\mu_->0$. For the following mutually exclusive and collectively exhaustive cases, the set $\mathcal{M}_{\zeta}$ is a singleton set and can be characterized explicitly as follows.
\begin{itemize}
\item If $0<a\leq a_1$, then $\mathcal{M}_{\zeta}=\{((g^{\prime})_-^{-1}(g^{\prime}(\phi^{-1}(\beta))),\phi^{-1}(\beta))\}$.
\item If $a>a_1$, then $\mathcal{M}_{\zeta}=\{((\Tilde{g}^{\prime})_-^{-1}(g^{\prime}(\Tilde{\phi}^{-1}(\beta))),\Tilde{\phi}^{-1}(\beta))\}.$
\end{itemize}
Here, the functions 
$\phi^{-1}$, $\Tilde{\phi}^{-1}$, $(g^{\prime})^{-1}_{\pm}$ and $(\Tilde{g}^{\prime})^{-1}_{\pm}$ are defined just before Theorem \ref{thm2.1}.
\end{thm}

\subsection{Explicit characterization of $\mathcal{M}_{\zeta}$ in the case $\mu_{+}>0$ and $\mu_{-}\leq 0$.}
\label{sec3.4}

We finally consider the case $\mu_{+}>0$ and $\mu_-\leq 0$, under which one can conclude $\theta_1^+>\theta^+_2$ and $\theta_1^-\leq \theta_2^-$ using \eqref{theta12}. 
By considering all scenarios of settings of the parameters, we distinguish the following mutually exclusive and collectively exhaustive Cases (i) and (ii).
\begin{itemize}
\item[]Case (i): one of the following conditions holds
\begin{itemize}
\item[$\bullet$] $c_+>0$ and $a\geq a_2$,
\item[$\bullet$] $c_+>0$ and $0<a_3\leq a\leq a_2$,
\item[$\bullet$] $c_+\leq0$, $a\geq a_3$ and $\Theta>0$.
\end{itemize}
\item[]Case (ii): one of the following conditions holds
\begin{itemize}
\item[$\bullet$] $c_+>0$ and $0<a<a_2\wedge a_3$,
\item[$\bullet$] $c_+\leq 0$ and $\Theta\leq0$, 
\item[$\bullet$] $c_+\leq 0$, $0<a< a_3$ and $\Theta>0$. 
\end{itemize}
\end{itemize}

The following Proposition \ref{prop.4.5} gives a complete characterization of the piece-wise concavity or convexity of the function $g(x)$ on $(0,\infty).$ The proof is deferred to Appendix \ref{sec5.7}.

\begin{prop}\label{prop.4.5}
Suppose that $\mu_{+}>0$ and $\mu_-\leq 0$. 
Under Case (i), $g(x)$ is convex on $(0,\infty)$.
Under Case (ii), $g(x)$ is convex on $(0,a)$, concave on $(a,x_0)$ and convex on $(x_0,\infty)$.
\end{prop}

For later use, we introduce the following functions. Under Case (ii), 
$[x_0,a_7]\ni x\mapsto (\hat{g}^{\prime})_-^{-1}(x):=\inf\{y\in[0,a]; g^{\prime}(y)\geq g^{\prime}(x)\}$ and define the inverse function of $\left.g^{\prime}\right|_{[a,x_0]}$ as $(\hat{g}^{\prime})_+^{-1}:[g^{\prime}(x_0),g^{\prime}(a_7)]\rightarrow [a,x_0]$ with $a_7:=\sup\{x>x_0:g^{\prime}(x)\leq g^{\prime}(a)\}$. 
Recall that $\phi_0$ is given by \eqref{def.phi0}.
Furthermore, put
\begin{align}
x_3&:=\inf\bigg\{x\in [x_0,a_7]:\int_{0}^{(\hat{g}')^{-1}_+(g'(x))}\left(1-\frac{g'(s)}{g'(x)}\right)\ds\geq0\bigg\},\\
x_4&:=\inf\bigg\{x\in[x_0,a_7]:\int_{(\hat{g}')^{-1}_-(g'(x))}^{x}\left(1-\frac{g'(s)}{g'(x)}\right)\ds\geq0\bigg\},\\
\omega_3(x)&:=
\int_{0}^{x}\left(1-\frac{g'(s)}{g'(x)}\right)\ds, \quad x\in[0,(\hat{g}')^{-1}_-(g'(x_4))]\cup [x_4,\infty),
\\ 
\omega_4(x)&:=
\begin{cases}
\int_{(\hat{g}')^{-1}_+(g'(x))}^{x}\left(1-\frac{g'(s)}{g'(x)}\right)\ds, & x\in[x_0,x_3) \\
\int_{0}^{x}\left(1-\frac{g'(s)}{g'(x)}\right)\ds,& x\in[x_3,\infty).
\end{cases}
\end{align}

The following Theorem \ref{thm4.3} explicitly characterizes the set $\mathcal{M}_{\zeta}$
in the case $\mu_{+}>0$ and $\mu_-\leq 0$.
We omit the proof due to its similarity to that of Theorem \ref{thm2.1}. 

\begin{thm}
\label{thm4.3}
Suppose that $\mu_+>0$ and $\mu_-\leq0$. 
\begin{itemize}
\item Under Case (i), we have $\mathcal{M}_{\zeta}=\{(0,\phi_0^{-1}(\beta))\}$.
\item Under Case (ii), we have
\begin{equation*}
\mathcal{M}_{\zeta}=
\begin{cases}
\{(\hat{w}_1,\hat{w}_2)\},
&\text{if }\beta\in B_1\cup B_3,
\\
\{(\tilde{w}_1,\tilde{w}_2)\},
&\text{if }\beta\in B_2\cap \overline{B}_3,
\\
\{(\hat{w}_1,\hat{w}_2)\}\cup\{(\tilde{w}_1,\tilde{w}_2)\},
&\text{if }\beta\in\overline{B_1\cup B_2\cup B_3},
\end{cases}
\end{equation*}
where, $(\tilde{w}_1,\tilde{w}_2):=((\hat{g}')^{-1}_+(g'(\omega^{-1}_4(\beta))),\omega^{-1}_4(\beta))$,\,$(\hat{w}_1,\hat{w}_2):=(0,\omega_3^{-1}(\beta))$, $B_1:=\{\beta>0:g'(\omega^{-1}_3(\beta))<g'(\omega^{-1}_4(\beta))\}$, $B_2:=\{\beta>0:g'(\omega^{-1}_3(\beta))>g'(\omega^{-1}_4(\beta))\}$, $B_3:=(\omega_4(x_3),\infty)$,
\,with $\omega^{-1}_3$ and $\omega^{-1}_4$ being the inverse functions of $\omega_3$ and $\omega_4$, respectively. It holds that $B_2\cap B_3=\emptyset$. We mention that, when $\beta=\omega_3(x_4)$, one has $\omega^{-1}_3(\beta)=\{(\hat{g}')^{-1}_-(g'(x_4)),x_4\}$; in this case, $\{(\hat{w}_1,\hat{w}_2)\}$ is understood as $\{(0,(\hat{g}')^{-1}_-(g'(\omega^{-1}_3(\beta))),(0,x_4)\}$.
\end{itemize}
\end{thm}

\subsection{General properties of $\mathcal{M}_{\zeta}$.}
\label{sec3.5}
We showed the characterization of the explicit form of the set $\mathcal{M}_{\zeta}$ in the previous section. 
The properties of the elements of $\mathcal{M}_{\zeta}$ are stated in the following Propositions \ref{prop.2.4.add}, \ref{prop.2.5.add} and \ref{prop.3.4}, whose proofs are lengthy and nontrivial and hence are deferred to Appendixes \ref{sec5.8}, \ref{sec5.9} and \ref{sec5.10}.

\begin{prop}
\label{prop.2.4.add}
We have $\lim_{\beta\rightarrow0+}\max_{(z_1,z_2)\in\mathcal{M}_{\zeta}}(z_2-z_1)=0.$
\end{prop}

\begin{prop}
\label{prop.2.5.add}
Both $\max_{(z_1,z_2)\in\mathcal{M}_{\zeta}}(z_2-z_1)$ and $\min_{(z_1,z_2)\in\mathcal{M}_{\zeta}}(z_2-z_1)$ are increasing in $\beta$.
\end{prop}

\begin{prop}
\label{prop.3.4}
Let $(z_1,z_2)\in\mathcal{M}_{\zeta}$. For any fixed $\beta\in(0,a)$, there exists a sufficiently large constant $K>0$ such that
$$\beta\leq z_1+\beta< z_2\leq a,\quad \text{if }\,\mu_->K.$$
\end{prop}

\begin{rem}
We treat $(z_{1},z_{2})$-strategy (with $(z_1,z_2)\in\mathcal{M}_{\zeta}$) as a candidate optimal impulsive dividend strategy (whose optimality will be demonstrated in the upcoming Theorem \ref{thm3.3.vr.}) of the control problem \eqref{6}. As $\beta$ increases, the issuing of a new lump-sum of dividends becomes more costly. Hence, when $\beta$ is larger, it would be sensible to adjust the dividend barriers so that the size of each lump-sum of dividends becomes larger. While, in the extreme case of $\beta=0$ (paying dividends incurs no costs), it seems reasonable to pay dividends as much and as frequently as possible, implying that $z_{2}=z_{1}$ in the limiting sense. These intuitive perceptions are confirmed in the above Propositions \ref{prop.2.4.add}-\ref{prop.2.5.add}.
\end{rem}

\begin{rem}
By Proposition \ref{prop.3.4}, if the expected rate of return $\mu_{-}$ (i.e., the drift coefficient when the wealth level is below $a$) is sufficiently large, then
the manager prefers the wealth process to stay below $a$ (rather than above $a$) to quickly accumulate wealth.
In turn, the upper barrier $z_2$ of the carefully calibrated optimal impulse dividend strategy $(z_1,z_2)\in\mathcal{M}_{\zeta}$ is lower than $a$.
\end{rem}

\section{Characterization of the optimal impulsive strategy} \label{sec:main}
\setcounter{section}{4}

The main result of this paper is contained in the following theorem, which, under certain sufficient conditions, characterizes an optimal impulsive strategy to the control problem \eqref{6}.

\begin{thm}\label{thm3.3.vr.}
{Let $(z_1,z_2)$ be the element of $\mathcal{M}_{\zeta}$.}
Then, the $(z_1,z_2)$-strategy is an optimal impulsive strategy to the control problem \eqref{6}, if one of the following conditions holds true:
\begin{itemize}
\item[(a)] $z_{2}>a$;
\item[(b)] $z_{2}\leq a$ and $\mu_+-q\left(a-z_2+ \frac{g(z_{2})}{g^{\prime}(z_{2})}\right)\leq 0$; 
\item[(c)] $g''(a+)\geq0.$ 
\end{itemize}
\end{thm}
\begin{proof}
Let $(z_1,z_2)\in\mathcal{M}_{\zeta}$.
We first prove that
\begin{align}
\label{addnew50.v3}
V^{z_2}_{z_1}(x)-V^{z_2}_{z_1}(y)\geq x-y-\beta, \quad x\geq y\geq0.
\end{align}
By \eqref{M} and \eqref{z2.int.}, we have
\begin{align}\label{x-y-beta}
\frac{x-y-\beta}{g(x)-g(y)}\leq \frac{z_2-z_1-\beta}{g(z_2)-g(z_1)}=\frac{1}{g^{\prime}(z_2)},\quad \beta \leq y+\beta\leq x<\infty.
\end{align} 
We distinguish the following mutually exclusive and collectively exhaustive cases.
\begin{itemize}

\item If $y+\beta>x\geq y\geq0$, it holds that
\begin{align}
\label{addnew54}
V^{z_2}_{z_1}(x)-V^{z_2}_{z_1}(y)\geq 0>x-y-\beta.
\end{align}

\item If $x\geq y\geq z_2$ and $x\geq y+\beta$, it holds that
$$V^{z_2}_{z_1}(x)-V^{z_2}_{z_1}(y)=x-y>x-y-\beta.$$ 

\item If $x\geq z_{2}\geq y\geq 0$ 
and $x\geq y+\beta$, by \eqref{x-y-beta} and \eqref{addnew54}, one can get
$$V^{z_2}_{z_1}(x)-V^{z_2}_{z_1}(y)=x-z_2+\frac{g(z_2)-g(y)}{g^{\prime}(z_2)}\geq x-y-\beta.$$

\item If $z_{2}\geq x\geq y+\beta\geq \beta$, by \eqref{x-y-beta}, one can get
$$V^{z_2}_{z_1}(x)-V^{z_2}_{z_1}(y)=\frac{g(x)-g(y)}{g^{\prime}(z_2)}\geq x-y-\beta.$$
\end{itemize}
Combining above yields \eqref{addnew50.v3}.

We next prove that
\begin{align}
\label{30.add.new.x}
\frac{g(a)}{g^{\prime}(a)}\leq \frac{g(z_{2})}{g^{\prime}(z_{2})}+a-z_{2},\quad \text{if}\quad z_{2}\leq a.
\end{align}
Actually, by \eqref{scale.fun.} and \eqref{10}, it can be verified that
\begin{align*}
1-\frac{g(x)g^{\prime\prime}(x)}{(g^{\prime}(x))^{2}}=\left[\frac{g(x)}{g^{\prime}(x)}\right]^{\prime}&=\frac{e^{-(\theta_{1}^{-}-\theta_{2}^{-})x}(\theta_1^-+\theta_2^-)^{2}}{\big(\theta_{1}^{-}e^{-\theta_{1}^{-}x}+\theta_{2}^{-}e^{\theta_{2}^{-}x}\big)^{2}}
>0, \quad x\in[0,a].
\end{align*}
Suppose $ z_{2}\leq a.$
Using Propositions \ref{P.0.3}-\ref{prop.4.5}, and Theorems \ref{thm2.1}-\ref{thm4.3}, one sees that
$g^{\prime\prime}(x)> 0$ for all $x\in[z_{2},a)$.
\begin{align*}
0<\left[\frac{g(x)}{g^{\prime}(x)}\right]^{\prime}<1\,\,\,\text{for all}\,\,\, x\in[z_{2},a),\quad \text{if}\quad z_{2}\leq a,
\end{align*}
which implies \eqref{30.add.new.x}.

By \eqref{Vx.4}, $V^{z_2}_{z_1}(x)\in C^1(\mathbb{R}_{+})\cap C^2(\mathbb{R}_{+}\backslash\{a,z_{2}\})$. We next verify $(\mathcal{A}-q)V^{z_2}_{z_1}(x)\leq0$ for $x\in[0,\infty)\backslash\{a,z_{2}\}$.
Using \eqref{10}, \eqref{Vx.4}, and the fact that $(\mathcal{A}-q)g^{\pm}(x)=0$ for all $x\neq a$, we have
\begin{align}
\label{x<z2}
(\mathcal{A}-q)V^{z_2}_{z_1}(x)
&=[g^{\prime}(z_2)]^{-1}\left[g^-(0)(\mathcal{A}-q)g^+(x)-g^+(0)(\mathcal{A}-q)g^-(x)\right]
\\
&=0,\quad x\in(0,z_{2})\backslash\{a\}.\nonumber
\end{align}
In addition, 
\begin{align*}
\lim_{x\rightarrow z_2+}(\mathcal{A}-q)V^{z_2}_{z_1}(x)
&=
\lim_{x\rightarrow z_2+}(\mu_+ \mathbf{1}_{\{x>a\}} +\mu_- \mathbf{1}_{\{x\leq a\}})-qV^{z_2}_{z_1}(z_2)
\\ \nonumber
&=\mu_+ \mathbf{1}_{\{z_{2}\geq a\}} +\mu_- \mathbf{1}_{\{z_{2}<a\}}-qV^{z_2}_{z_1}(z_2),
\\
\lim_{x\rightarrow z_2-}(\mathcal{A}-q)V^{z_2}_{z_1}(x)
&=\frac{1}{2}(\sigma^2_+ \mathbf{1}_{\{z_2>a\}} +\sigma^2_- \mathbf{1}_{\{z_2\leq a\}}){V^{z_2}_{z_1}}^{\prime\prime}(z_2-)
\nonumber\\
&\quad +(\mu_+ \mathbf{1}_{\{z_2>a\}} +\mu_- \mathbf{1}_{\{z_2\leq a\}})-qV^{z_2}_{z_1}(z_2).
\end{align*}
Combining above yields 
\begin{align*}
0=\lim_{x\rightarrow z_2-}(\mathcal{A}-q)V^{z_2}_{z_1}(x) 
&=\frac{1}{2}(\sigma^2_+ \mathbf{1}_{\{z_2>a\}} +\sigma^2_- \mathbf{1}_{\{z_2\leq a\}}){V^{z_2}_{z_1}}^{\prime\prime}(z_2-)\\
&\quad\;+\left(\mu_{-}-\mu_{+}\right)\mathbf{1}_{\{z_2= a\}}
+\lim_{x\rightarrow z_2+}(\mathcal{A}-q)V^{z_2}_{z_1}(x).
\end{align*}
Using this and the fact that ${V^{z_2}_{z_1}}^{\prime\prime}(z_2-)\geq 0$ (Actually, from the explicit characterizations of $\mathcal{M}_{\zeta}$ provided in Theorems \ref{thm2.1}-\ref{thm4.3}, one knows that $g^{\prime\prime}(z_2-)\geq 0$, which, by \eqref{Vx.4}, is equivalent to ${V^{z_2}_{z_1}}^{\prime\prime}(z_2-)\geq 0$), one has
\begin{align}
\label{mu-+}
\left(\mu_{-}-\mu_{+}\right)\mathbf{1}_{\{z_2= a\}}
+
\lim_{x\rightarrow z_2+}(\mathcal{A}-q)V^{z_2}_{z_1}(x)\leq 0.
\end{align}

We next prove $(\mathcal{A}-q)V^{z_2}_{z_1}(x)\leq 0$ on $(z_2,\infty)\backslash\{a\}$. 
\begin{itemize}
\item When Condition (a) holds true, it follows from \eqref{mu-+} that
\begin{align*}
(\mathcal{A}-q)V^{z_2}_{z_1}(x)&=\mu_+-qV^{z_2}_{z_1}(x)\leq
\mu_+-qV^{z_2}_{z_1}(z_2)\\
&=
\lim_{x\rightarrow z_2+}(\mathcal{A}-q)V^{z_2}_{z_1}(x)\leq -\left(\mu_{-}-\mu_{+}\right)\mathbf{1}_{\{z_2= a\}}=0,~~x>z_2.\nonumber
\end{align*}
\item When Condition (b) holds true, we have
\begin{align*}
(\mathcal{A}-q)V^{z_2}_{z_1}(x)
&=
\mu_--qV^{z_2}_{z_1}(x)
\leq
\mu_--qV^{z_2}_{z_1}(z_2)
\\
&=\left(\mu_--qV^{z_2}_{z_1}(z_2)\right)\mathbf{1}_{\{z_{2}<a\}}
+\left(\mu_+-qV^{z_2}_{z_1}(z_2)\right)\mathbf{1}_{\{z_{2}=a\}}\\
&\quad\;+\left(\mu_--\mu_+\right)\mathbf{1}_{\{z_{2}=a\}}
\nonumber\\
&=
\lim_{x\rightarrow z_2+}(\mathcal{A}-q)V^{z_2}_{z_1}(x)
+\left(\mu_--\mu_+\right)\mathbf{1}_{\{z_{2}=a\}}\leq 0, ~~~x\in (z_2, a], 
\end{align*}
and 
\begin{align*}
(\mathcal{A}-q)V^{z_2}_{z_1}(x)&=\mu_+-qV^{z_2}_{z_1}(x)
=\mu_+-q(x-z_{2})-qV^{z_2}_{z_1}(z_2)
\\
&\leq \mu_+-q(a-z_{2})-qg(z_2)/g^{\prime}(z_2)
\leq 0, \quad x>a. \nonumber
\end{align*}

\item When Condition (c) holds true, the claim follows if $z_{2}> a$ by Condition (a); otherwise $z_2\leq a$, and 
\begin{align*}
(\mathcal{A}-q)V^{z_2}_{z_1}(x)&=\mu_--qV^{z_2}_{z_1}(x)=\mu_--q\big((x-z_{2})+g(z_2)/g^{\prime}(z_2)\big)\\
&\leq \mu_--qg(z_2)/g^{\prime}(z_2)
=-\frac{1}{2}\sigma_{-}^{2}g^{\prime\prime}(z_2-)/g^{\prime}(z_2)
\leq 0, \quad x\in (z_2, a],\\
(\mathcal{A}-q)V^{z_2}_{z_1}(x)
&\leq \mu_+-q\big((a-z_{2})+g(z_2)/g^{\prime}(z_2)\big)
\\
&\leq \mu_+-qg(a)/g^{\prime}(a) =-\frac{1}{2}\sigma_{+}^{2}g^{\prime\prime}(a+)/g^{\prime}(a)
\leq 0, \quad x>a,\nonumber
\end{align*}
where the equality is due to \eqref{dis.generator.}, and the second inequality due to \eqref{30.add.new.x}.
\end{itemize}
Together with \eqref{x<z2}, we proved $(\mathcal{A}-q)V^{z_2}_{z_1}(x)\leq 0$ on $(0,\infty)\backslash\{a,z_{2}\}$ when one of Conditions (a), (b), or (c) holds true. Combining with \eqref{addnew50.v3} and Lemma \ref{lem2.1}, we prove Theorem \ref{thm3.3.vr.}.
\end{proof}

\begin{coro}
There exists a $(z_1,z_2)\in\mathcal{M}_{\zeta}$ such that the $(z_1,z_2)$-strategy is an optimal impulsive strategy to the control problem \eqref{6}, 
except for the following two minor cases:
\begin{itemize}
\item $\mu_{\pm}>0$, Case (iv), 
 $\beta<\omega_{1}(x_2)$, $g^{\prime}(\omega^{-1}_1(\beta))< g^{\prime}(\omega^{-1}_2(\beta))$;
\item $\mu_{+}>0$, $\mu_{-}<0$, Case (ii), $\beta<\omega_{3}(x_4)$, $g^{\prime}(\omega^{-1}_3(\beta))< g^{\prime}(\omega^{-1}_4(\beta))$.
\end{itemize}
In the last two cases, the $(z_1,z_2)$-strategy remains optimal if Condition (b) in Theorem \ref{thm3.3.vr.} is satisfied. 
\end{coro}
 
\begin{proof}
The proof is a straightforward application of Propositions \ref{P.0.3}-\ref{prop.4.5}, and Theorems \ref{thm2.1}-\ref{thm4.3} and \ref{thm3.3.vr.}, one just needs to check the following facts. 

(1) Assume $\mu_{\pm}>0$. Then 
\begin{itemize}
\item In Cases (i) and (ii), $\mathcal{M}_{\zeta}$ is a singleton, and we have $z_{2}>a$. 
\item In Case (iii), $\mathcal{M}_{\zeta}$ is a singleton, and we have $g^{\prime\prime}(a+)>0$. 
\item In Case (iv), if either one of the following conditions
\begin{itemize}
\item $\beta\geq \omega_{1}(x_2)$,
\item $\beta<\omega_{1}(x_2)$ \text{ and } $g^{\prime}(\omega^{-1}_1(\beta))\geq g^{\prime}(\omega^{-1}_2(\beta))$,
\end{itemize}
holds true, then $\mathcal{M}_{\zeta}$ is not necessarily a singleton, but there is at least one $(z_1,z_2)\in \mathcal{M}_{\zeta}$ with $z_{2}>a$.
\end{itemize}

(2) Assume $\mu_{\pm}\leq 0$. Then $\mathcal{M}_{\zeta}$ is a singleton, and we have $g^{\prime\prime}(a+)>0$. 

(3) Assume $\mu_{+}\leq 0$ and $\mu_{-}>0$. Then $\mathcal{M}_{\zeta}$ is a singleton.
\begin{itemize}
\item If $0<a<a_1$, then $z_{2}>a$.
\item If $a>a_{1}$, then $g^{\prime\prime}(a+)>0$. 
\end{itemize}

(4) Assume $\mu_{+}>0$ and $\mu_{-}<0$. Then
\begin{itemize}
\item In Case (i), $\mathcal{M}_{\zeta}$ is a singleton, and we have $g^{\prime\prime}(a+)>0$. 
\item In Case (ii), if either one of the following conditions
\begin{itemize}
\item $\beta\geq \omega_{3}(x_4)$,
\item $\beta<\omega_{3}(x_4)$\quad \text{and}\quad $g^{\prime}(\omega^{-1}_3(\beta))\geq g^{\prime}(\omega^{-1}_4(\beta))$,
\end{itemize}
holds true, then $\mathcal{M}_{\zeta}$ may not be a singleton, but there is at least one $(z_1,z_2)\in \mathcal{M}_{\zeta}$ with $z_{2}>a$.
\end{itemize}
The proof is simple, so we omit the details. 
 \end{proof}

\section{Numerical Analysis and Economic Interpretations}

To complement the theoretical results derived in the previous sections, we now present a series of numerical experiments. These analyses serve to visualize the optimal dividend strategy in action and to provide economic intuition for how the firm's decisions are shaped by the underlying model parameters.

\subsection{ A sample path under the optimal strategy}

We begin by simulating a sample path of the optimal controlled surplus process $U^{\pi^*}$. The left panel of Figure \ref{fig: path} illustrates this path over a time horizon of $T=10$. The baseline parameters for this simulation are set as follows: the regime-switching threshold is $a=1$; the drift and volatility coefficients are $(\mu_-,\sigma_-)=(0.5,0.5)$ for the lower regime $(U_t\leq a)$ and $(\mu_+,\sigma_+)=(0.1,0.1)$ for the upper regime $(U_t>a)$; and the fixed transaction cost $\beta=0.5$. For this parameter set, the explicit characterization of $\mathcal{M}_{\zeta}$ derived in Section \ref{sec:negative} identifies the optimal dividend barriers as $z^*_1=0.4277$ and $z^*_2=1.9059$. This strategy constitutes the optimal dividend strategy, since Condition $(a)$ of Theorem \ref{thm3.3.vr.} holds. This models a company with a high-growth, high-risk ``startup" phase (when surplus is below $a$) that transitions into a low-growth, low-risk ``maturity" phase upon expansion (surplus above $a$). As observed in the left panel of Figure \ref{fig: path}, the firm sets a relatively low upper barrier $z^*_2$. It allows the surplus to enter the more stable (but less profitable) mature phase to safely accumulate funds and pay a dividend of size $\Delta L_t = z^*_2-z^*_1$. The dividend resets the surplus to $z^*_1$, positioning the firm to re-leverage its high-growth startup phase. The strategy is thus a dynamic cycle of navigating between risk and stability.

The right panel of Figure \ref{fig: path} represents a completely different economic reality over
a time horizon of $T = 50$. Here, the parameters are reversed: $(\mu_+,\sigma_+)=(0.5,0.5)$ and $(\mu_-,\sigma_-)=(0.1,0.1)$ with $\beta = 0.5$ and $a=1$. This models a company that is stable but stagnant when small, and only enters a high-growth, high-risk expansion phase after its surplus exceeds the threshold $a$. 
The firm endures a long period of slow growth below $a$, as the low drift makes it difficult to cross the threshold. Once the threshold is crossed, the firm enters the highly profitable expansion phase. The optimal strategy is to establish a high upper barrier at $z_2^*=10.4512$, a value consistent with Condition $(a)$ of Theorem \ref{thm3.3.vr.}, to capitalize on this phase. Upon reaching $z^*_2$, the firm pays out a very large dividend of $\Delta L_t = z^*_2-z^*_1$. The lower barrier is set at $z^*_1 = 0$, the ruin level itself. It signifies that the optimal path is not to restart, but to perform a terminal payout. This aggressive strategy is rational for shareholders aiming to fully extract the firm's value after a successful high-growth period, rather than risk returning to the stagnant, low-growth phase.

\begin{figure}[H]
\centering
\includegraphics[width=3in]{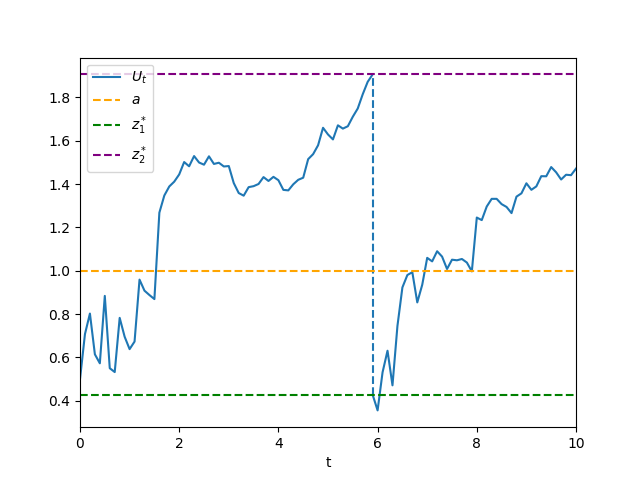}
\hspace{0.1in} 
\includegraphics[width=3in]{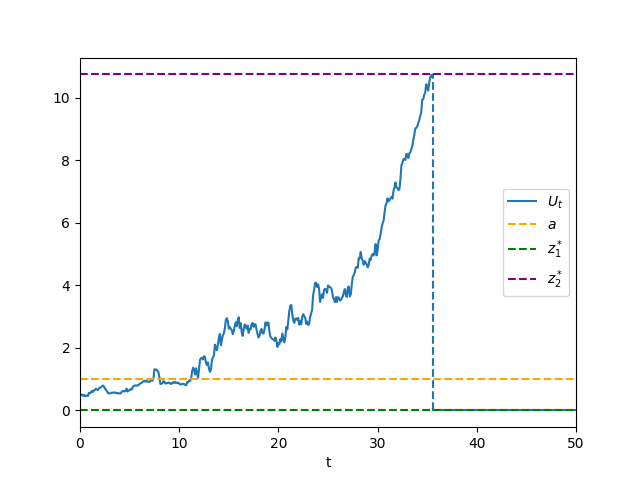}
\vspace{-0.2in}
\caption{{\footnotesize A sample path under the optimal strategy
}}
\label{fig: path}
\end{figure}

\subsection{Sensitivity analysis of model parameters}

We now investigate the sensitivity of the optimal dividend barriers $(z^*_1,z^*_2)$ with respect to the key model parameters: the transaction cost $\beta$, the regime-switching threshold $a$, the lower regime drift $\mu_-$, and the lower regime volatility $\sigma_-$.

The top-left panel of Figure \ref{fig:sensi} shows the effect of the transaction cost $\beta$ on the optimal barriers under $(\mu_+,\sigma_+)=(0.1,0.1)$, $(\mu_-,\sigma_-)=(0.5,0.5)$ and $a = 1$. The results are in complete agreement with Propositions \ref{prop.2.4.add} and \ref{prop.2.5.add}, as well as satisfying Condition $(a)$ of Theorem \ref{thm3.3.vr.}. As the cost $\beta$ of paying a dividend increases, the upper barrier $z^*_2$ increases, while the lower barrier $z^*_1$ decreases. Consequently, the dividend size $\Delta L_t = z^*_2-z^*_1$ becomes a strictly increasing function of the cost. The intuition is straightforward: a higher fixed cost per transaction incentivizes the firm to pay dividends less frequently but in larger amounts. 

We now examine the sensitivity of the optimal strategy to the regime-switching threshold $a$ under $(\mu_+,\sigma_+)=(0.5,0.1)$, $(\mu_-,\sigma_-)=(1,0.5)$ and $\beta = 1$. Each $z^*_2$ value shown in the figure satisfies Condition $(a)$ of Theorem \ref{thm3.3.vr.}.
When $a$ is small, the high-growth zone is narrow, forcing the firm to set a high $z^*_2$ to patiently accumulate surplus in the safe upper regime. As $a$ increases, making the high-growth zone more accessible, the strategy becomes more aggressive by progressively lowering $z^*_2$ to realize profits sooner. The sharp kink at $a = 2.6$ marks a fundamental strategic shift. At and beyond this point, the firm prefers to pay a dividend precisely when it is about to be pushed into the low-growth zone, immediately resetting the process to stay within its preferred, more profitable environment.

The bottom-left panel illustrates the impact of the lower regime's profitability $\mu_-$ under $(\mu_+,\sigma_+)=(0.5,1)$, $\sigma_- = 0.5$, $a = 2$ and $\beta = 1$. Each $z^*_2$ value shown in the figure satisfies Condition (a) of Theorem \ref{thm3.3.vr.}. When $\mu_-$ is negative, the lower regime is a ``danger zone." The firm adopts a highly conservative strategy, maintaining high barriers $z_1^*=4,z_2^*=10$ to create a large safety buffer and minimize the risk of ruin after a dividend payout. As $\mu_-$ increases and surpasses the upper regime's drift $\mu_+$, the lower regime becomes the engine of growth. The firm's strategy flips to become aggressive. It sets both barriers significantly lower to ensure that after a payout, the process returns to the highly profitable lower regime, and it pays dividends more quickly to minimize time spent in the less profitable upper regime.

Finally, the bottom-right panel shows the effect of the lower regime's risk $\sigma_-$ under $(\mu_+,\sigma_+)=(0.5,0.5)$, $\mu_-=1$, $\beta = 1$ and $a=8$.
Both barriers $z^*_1$ and $z^*_2$ are monotonically increasing functions of $\sigma_-$. For $\sigma_->0.63$, $z_2^*$ satisfies condition $(a)$ of Theorem \ref{thm3.3.vr.}; for $\sigma_-<0.63$, both Conditions $(b)$ and $(c)$ hold true. This reflects the firm's response to an increase in operational risk. A higher volatility in the primary operating regime increases the probability of ruin. To mitigate this risk, the firm adopts a more conservative policy by holding more precautionary cash. It raises $z^*_1$ to leave a larger buffer and consequently must also raise $z^*_2$ to ensure that the payout remains large enough to justify the transaction cost. Higher risk, therefore, leads to delayed and larger dividend payments.

\begin{figure}[H] 
\centering 
\includegraphics[width=2.4in]{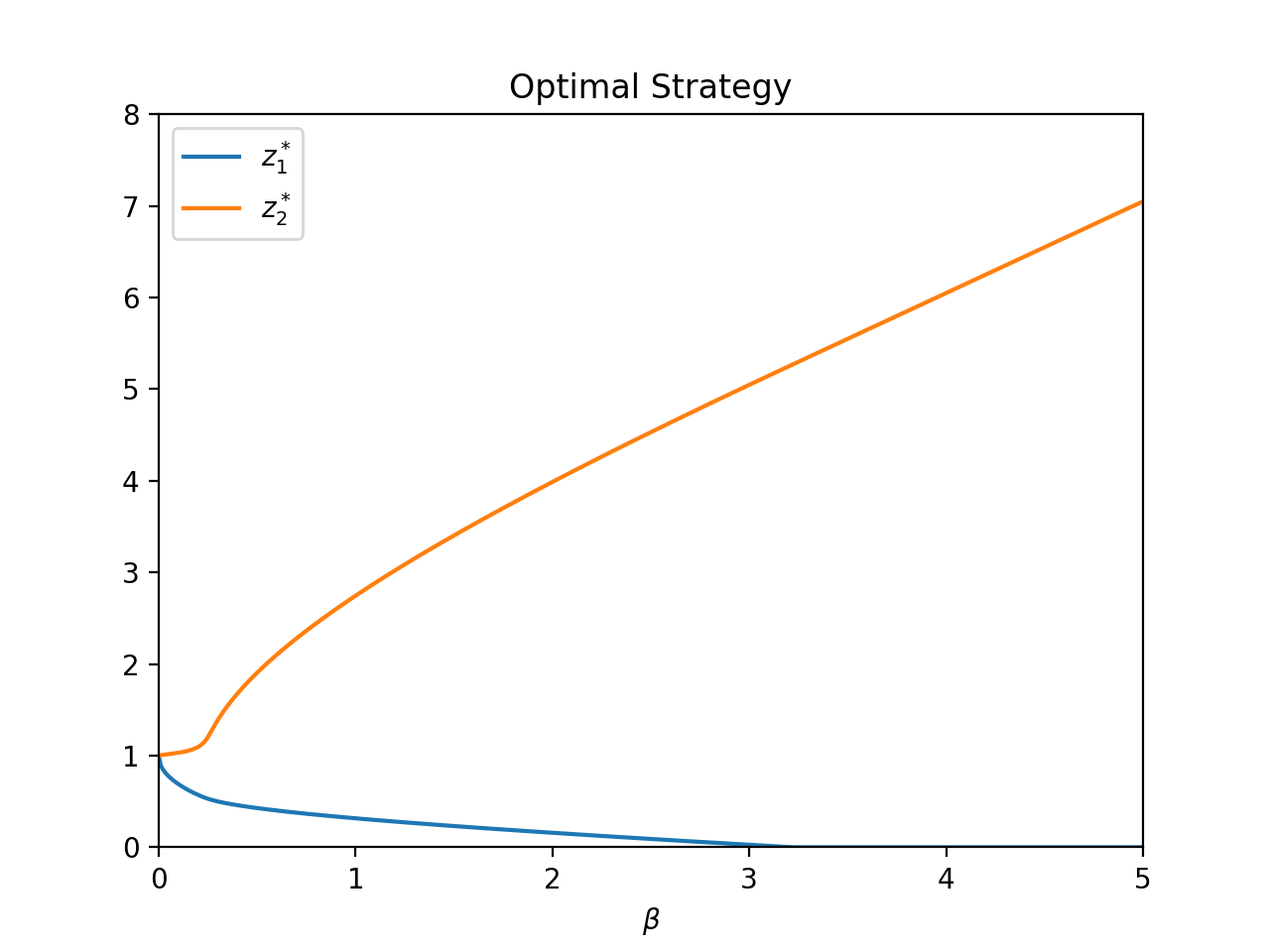} 
\hspace{0.1in} 
\includegraphics[width=2.4in]{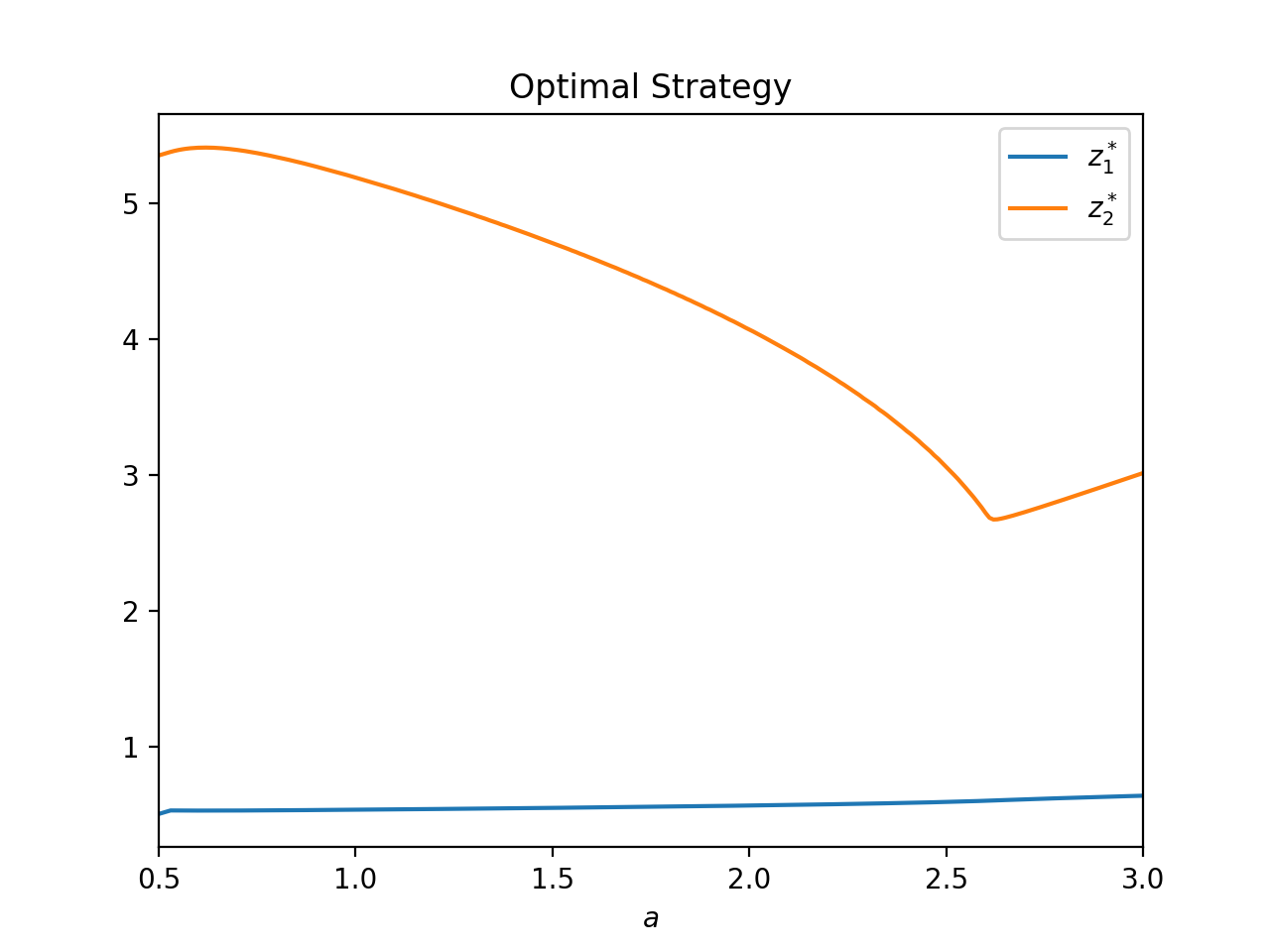} 
\vspace{0.1in} 
\includegraphics[width=2.4in]{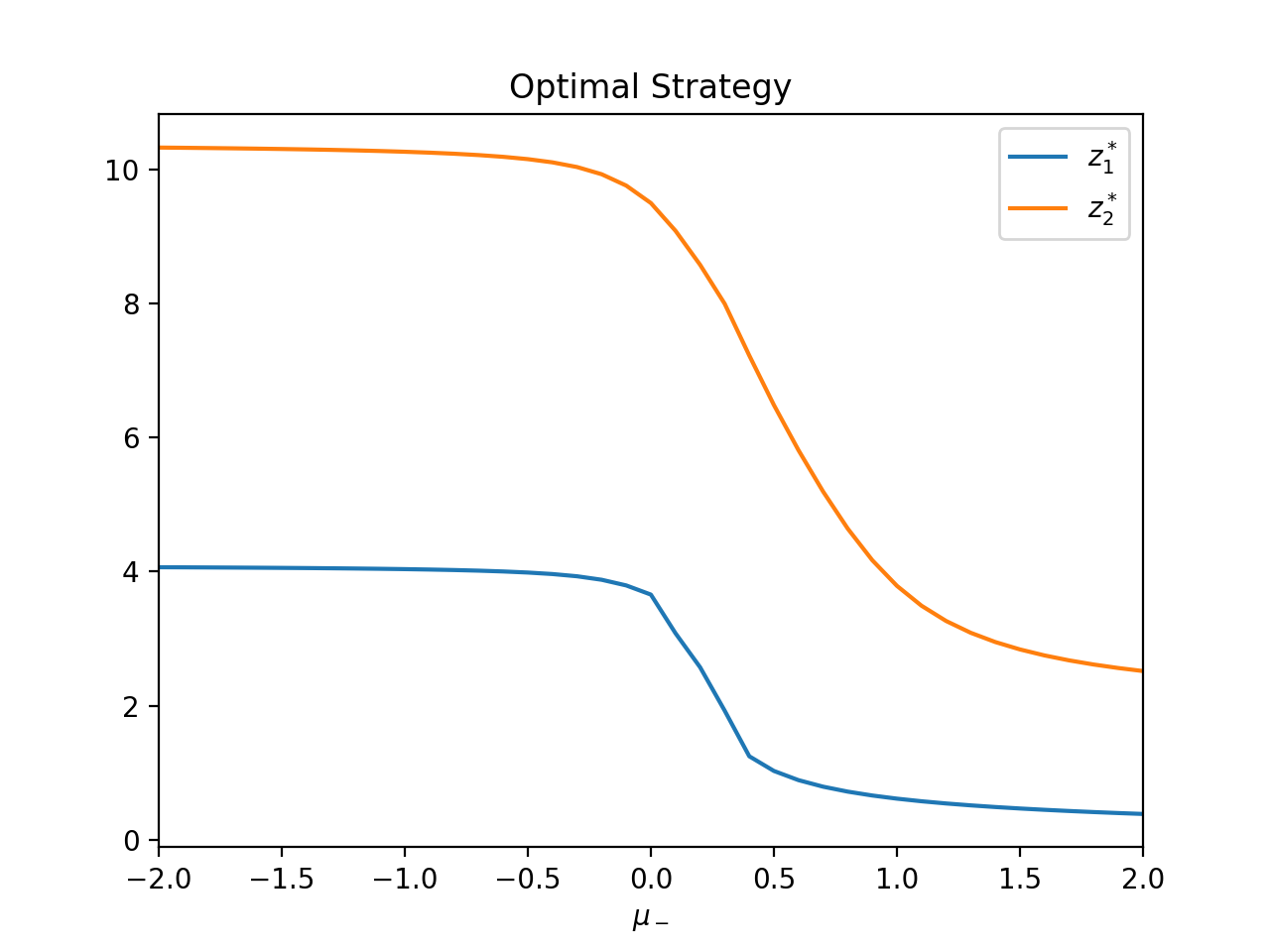} 
\hspace{0.1in}
\includegraphics[width=2.4in]{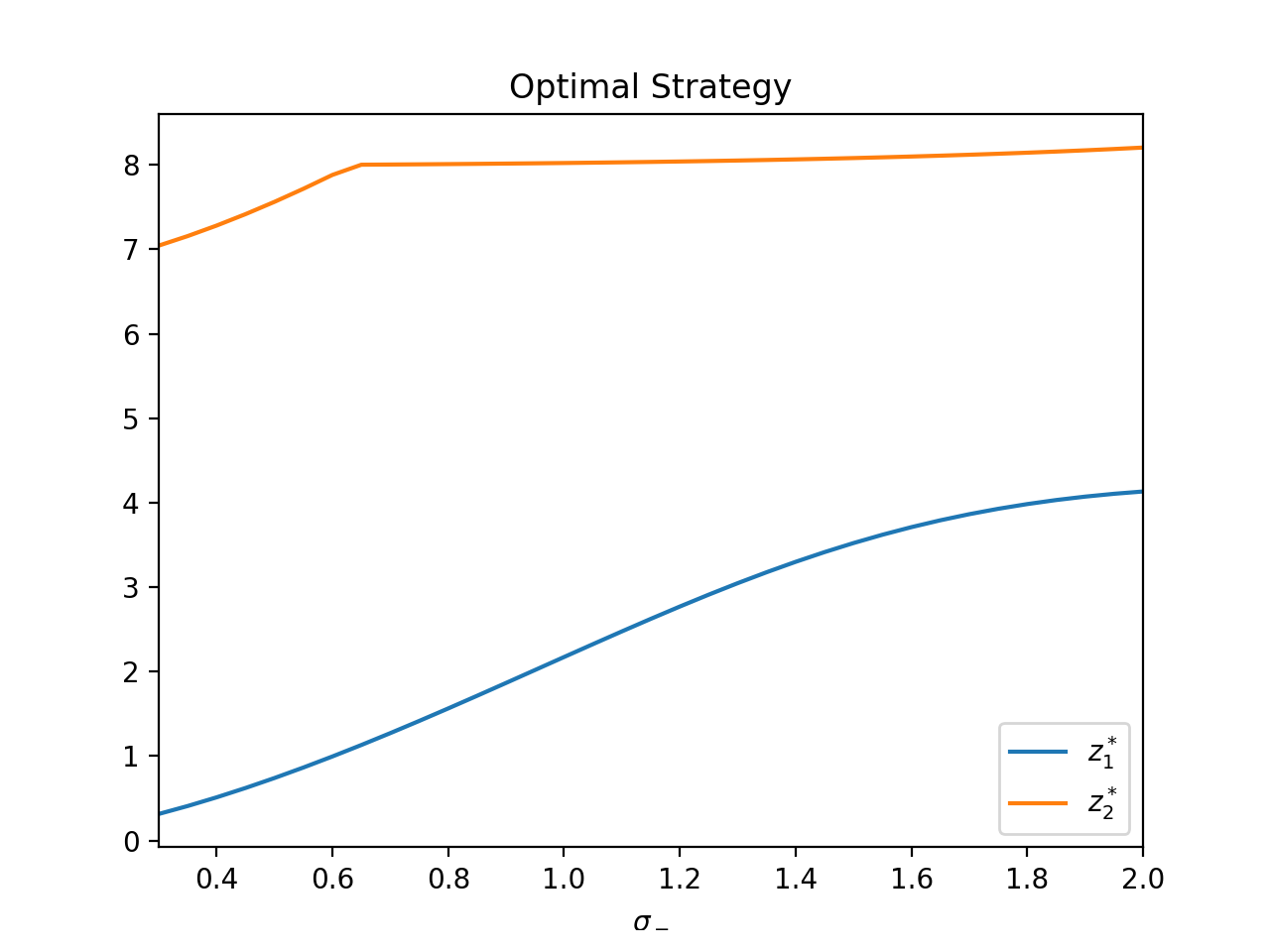} 
\vspace{-0.2in} 
\caption{\footnotesize Sensitive of parameters $\beta$, $a$, $\mu_-$ and $\sigma_-$}
\label{fig:sensi} 
\end{figure}

\appendix

\section{Proofs}
\label{sec.4}

In this appendix, we provide the proofs for some results given in the previous sections.

\subsection{Proof of Lemma \ref{lem2.1.add.new.x}}\label{sec5.1}

As the proof for the lower bound is trivial, we only need to prove the upper bound. 
For any admissible impulsive dividend payout strategy $\pi=(L^\pi_t)_{t\geq 0}$, we always have
\begin{align}
L^\pi_t\leq \sup_{s\in[0,t]}\left(X_{s}\right)_{+}\leq x+\sup_{s\in[0,t]}\left(\mu_{+}s+\sigma_{+}B_{s}\right)_{+}
+
\sup_{s\in[0,t]}\left(\mu_{-}s+\sigma_{-}B_{s}\right)_{+},\quad t\geq 0,\nonumber
\end{align}
where $x_{+}:=x\vee 0$ and $X$ is the unique solution of \eqref{def.X}. Hence, for any $x\geq 0$,
\begin{align}
\label{38.add.new.x}
V_\pi(x)&\leq x+
\bE\Big[\int_0^{\infty}e^{-qt} \dd \Big(\sup_{s\in[0,t]}\left(\mu_{+}s+\sigma_{+}B_{s}\right)_{+}
+
\sup_{s\in[0,t]}\left(\mu_{-}s+\sigma_{-}B_{s}\right)_{+}\Big)\Big]
\nonumber
\\
&=x+\frac{1}{q}\bE\Big[\sup_{s\in[0,\mathrm{e}_{q}]}\left(\mu_{+}s+\sigma_{+}B_{s}\right)_{+}
+
\sup_{s\in[0,\mathrm{e}_{q}]}\left(\mu_{-}s+\sigma_{-}B_{s}\right)_{+}\Big], 
\end{align}
where $\mathrm{e}_{q}$ denotes an exponential random variable (with mean $1/q$) independent of the Brownian motion $B$. By formula (1.1.2) in Page 250 of Part $\mathbf{II}$ of \cite{BANPS2015}, for constants $\mu\in\mathbb{R}$ and $\sigma>0$, we have
\begin{align}
\label{39.add.new.x}
\bP\bigg(\sup_{s\in[0,\mathrm{e}_{q}]}\left(\mu s+\sigma B_{s}\right)_{+}\geq y\bigg)=
e^{\big(\frac{\mu}{\sigma}-\sqrt{2q+\frac{\mu^{2}}{\sigma^{2}}}\big)\frac{y}{\sigma}},\quad y\geq 0.
\end{align}
Combining \eqref{38.add.new.x}-\eqref{39.add.new.x} and using the arbitrariness of $\pi$ yields the upper bound of \eqref{5.add.new.y}.

\subsection{Proof of Lemma \ref{prop2.2.v2}}
\label{sec5.2}

By definition one has $g(0)=0$.
Using \eqref{scale.fun.} and \eqref{sign.c+} one can verify straightforwardly that
\begin{align}
\label{15.v2}
g^-(0)>g^+(0)>0, \quad (1-c_+)\theta_2^+-c_+\theta_1^+=\theta_2^->0.
\end{align}
By \eqref{10} and the fact that $g\in C^1(\bR)\cap C^2(\bR\backslash \{a\})$, for any $x\in\bR$, it holds that
\begin{align}
\label{11}
g^{\prime}(x)
&= 
\left[(1-c_{+})g^-(0)\theta_2^+e^{\theta_2^+{(x-a)}}+(g^+(0)-c_+g^-(0))\theta_1^+e^{-\theta_1^+{(x-a)}} \right]1_{\{x>a\}}
\nonumber\\
&
+\left[(g^-(0)-c_{-}g^+(0))\theta_2^-e^{\theta_2^-{(x-a)}}+(1-c_-)g^+(0)\theta_1^-e^{-\theta_1^-{(x-a)}}\right]1_{\{x\leq a\}}.
\end{align}
Using \eqref{scale.fun.} and \eqref{15.v2} one can find that
\begin{align}
\label{12}
g^-(0)-c_{-}g^+(0)&=c_-e^{-\theta_2^-a}+(1-c_-) e^{\theta_1^-a}-c_-e^{-\theta_2^-a}= (1-c_-) e^{\theta_1^-a}>0.
\end{align}
By \eqref{theta12}, \eqref{sign.c-}, \eqref{15.v2}-\eqref{12} and the fact of $(1-c_{-})\theta_{1}^{-}>0$ (see \eqref{theta12} and \eqref{sign.c-}), one knows that $g^{\prime}(x)>0$ for any $x\leq a$. 
Suppose $g'$ has a real root, and let $x^{*}$ be its smallest root. Then $x^{*}\in(a,\infty)$. Since $g'(x)>0$ for $x\in(0,x^{*})$, we see $g''(x^{*})\leq 0$ and $g(x^{*})> g(0)= 0$. Taking $g''(x^{*})\leq 0$, $g'(x^{*})= 0$ and $g(x^{*})> 0$ into \eqref{dis.generator.} leads to a contradiction. 
Hence, we conclude $g'$ has no real root, thus $g'>0$ on $\bR$. The claim follows.

\subsection{Proof of Proposition \ref{P.0.4}}
\label{sec5.11}

Recall that $g^{\prime}(z)>0$ for all $z>0$ (see Lemma \ref{prop2.2.v2}).
It is easy to verify that
\begin{align}
\frac{\partial}{\partial z_1}\zeta(z_1,z_2)=&\frac{g^{\prime}(z_1)\left[\frac{g(z_1)-g(z_2)}{g^{\prime}(z_1)}+z_2-z_1-\beta\right]}{(g(z_2)-g(z_1))^2}
=\frac{g^{\prime}(z_1)\left[\int_{z_1}^{z_2}\left[1-\frac{g^{\prime}(z)}{g^{\prime}(z_1)}\right]\mathrm{d}z-\beta\right]}{(g(z_2)-g(z_1))^2},\nonumber\\
\frac{\partial}{\partial z_2}\zeta(z_1,z_2)=&\frac{g^{\prime}(z_2)\left[\frac{g(z_2)-g(z_1)}{g^{\prime}(z_2)}-z_2+z_1+\beta\right]}{(g(z_2)-g(z_1))^2}.\nonumber
\end{align}
By \eqref{10} and \eqref{15.v2}, one knows that there exists $z_{1}^{0}\in(0,\infty)$ such that $g^{\prime}(z)$ is increasing over $[z_{1}^{0},\infty)$, which implies $\frac{\partial}{\partial z_1}\zeta(z_1,z_2)<0$ for all $z_1\geq z_{1}^{0}$.
In addition, it holds that
\begin{align}
\lim_{z_2\rightarrow\infty}\left[\frac{g(z_2)-g(z_1)}{g^{\prime}(z_2)}-z_2+z_1+\beta\right]=-\infty,~~\text{uniformly for }0\leq z_1\leq z_{1}^{0},\nonumber
\end{align}
which implies that, there exists $z_{2}^{0}\in(0,\infty)$ such that $\frac{\partial}{\partial z_2}\zeta(z_1,z_2)<0$ for all $z_{2}\geq z_{2}^{0}$.
Put $z_0:=z_{1}^{0}\vee z_{2}^{0}+\beta\in(0,\infty)$. Then, we have
\begin{align}
\zeta(z_1,z_2)\leq \zeta({z}_1,{z}_0), &\quad (z_{1},z_{2})\in [0,z_0-\beta]\times[z_{0},\infty),
\nonumber\\
\zeta(z_1,z_2)\leq \zeta({z}_0-\beta,{z}_2)
\leq \zeta({z}_0-\beta,{z}_0), &\quad z_0-\beta\leq z_{1}\leq z_{2}-\beta<\infty,\nonumber
\end{align}
which implies
\begin{align*}
\max_{\beta\leq z_1+\beta\leq z_2<\infty}\zeta({z}_1,{z}_2)
=\max_{\beta\leq z_1+\beta\leq z_2\leq z_0}\zeta({z}_1,{z}_2).
\end{align*}
We next rule out the possibility that $\zeta(z_1,z_2)$ attains its maximum value in the boundary line $z_2=z_1+\beta$. Indeed, for any $\Tilde{z}_1,\Tilde{z_2}$ satisfying $\Tilde{z}_2=\Tilde{z}_1+\beta\geq\beta$, it holds that $\zeta({\Tilde{z}_1},{\Tilde{z}_2})\equiv0$. However, by the fact that $g(x)$ is strictly increasing, one gets $\zeta(z_1,z_2)>0=\zeta({\Tilde{z}_1},{\Tilde{z}_2})$ for all $(z_1,z_2)\in\mathcal{D}_{\zeta}$ satisfying $z_2>z_1+\beta\geq\beta$. Hence
\begin{align*}
\max_{\beta\leq z_1+\beta\leq z_2\leq z_0}\zeta({z}_1,{z}_2)
=\max_{\beta\leq z_1+\beta<z_2\leq z_0}\zeta({z}_1,{z}_2).
\end{align*}
Combining above yields the desired result.

\subsection{Proof of Proposition \ref{P.0.3}}
\label{sec5.3}

By \eqref{10}, for any $x\in(0,\infty)\backslash\{a\}$, it is easy to verify that
\begin{align}
\label{g''.v2}
g^{\prime\prime}(x)
&=
\left[(1-c_{+})g^-(0)(\theta_2^+)^2e^{\theta_2^+{(x-a)}}-(g^+(0)-c_+g^-(0))(\theta_1^+)^2e^{-\theta_1^+{(x-a)}} \right]1_{\{x>a\}}
\nonumber\\
&\quad\;+\left[(g^-(0)-c_{-}g^+(0))(\theta_2^-)^2e^{\theta_2^-{(x-a)}}-(1-c_-)g^+(0)(\theta_1^-)^2e^{-\theta_1^-{(x-a)}}\right]1_{\{x<a\}}.
\end{align}

We split the proofs into (1) and (2) as follows.
\begin{itemize}
\item[(1)]
We first discuss the sign of $g^{\prime\prime}(x)$ for $x\in(0,a)$.
When $x\in (0,a)$, using \eqref{scale.fun.} and \eqref{g''.v2} we have
\begin{align}
\label{g''x<a}
g^{\prime\prime}(x)
&=
(g^-(0)-c_{-}g^+(0))(\theta_2^-)^2e^{\theta_2^-{(x-a)}}-(1-c_-)g^+(0)(\theta_1^-)^2e^{-\theta_1^-{(x-a)}}
\nonumber\\
&=
(1-c_-)e^{(\theta_1^--\theta_2^-)a}\left[(\theta_2^-)^2e^{\theta_2^-x}-(\theta_1^-)^2e^{-\theta_1^-x}\right],
\end{align}
which is strictly increasing with its unique zero $a_1>0$ given by \eqref{def.a}. We hence conclude that
\begin{itemize}
\item[(1-1)] If $0<a\leq a_{1}$, one has $g^{\prime\prime}(x)<0$ on $(0,a)$.
\item[(1-2)] If $a>a_{1}$, one has $g^{\prime\prime}(x)<0$ on $(0,a_{1})$ and $g^{\prime\prime}(x)>0$ on $(a_{1},a)$.
\end{itemize}

\item[(2)] In the sequel, we discuss the sign of $g^{\prime\prime}(x)$ for $x>a$.
When $x>a$, by \eqref{g''.v2}, it holds that
\begin{align}
\label{g''}
g^{\prime\prime}(x)
&=
(1-c_{+})g^-(0)(\theta_2^+)^2e^{\theta_2^+{(x-a)}}-(g^+(0)-c_+g^-(0))(\theta_1^+)^2e^{-\theta_1^+{(x-a)}}.
\end{align}
We first discuss the sign of
\begin{align}
\label{h_{1}(a)}
h_{1}(a):=g^+(0)-c_+g^-(0)
=e^{-\theta_2^-a}(1-c_+c_-)-c_+(1-c_-)e^{\theta_1^-a},
\end{align}
where we used \eqref{scale.fun.} in the second equality. Due to the fact that $g^-(0)>g^+(0)>0$ (see \eqref{15.v2}), if $c_{+}\leq 0$, we have $h_{1}(a)>0$ for all $a>0$.
If $c_+>0$, it follows from
$1-c_{+}c_{-}>0$ 
that the function $\mathbb{R}_{+}\ni x\mapsto h_{1}(x)$ is strictly decreasing and has a unique zero $a_2>0$ given by \eqref{def.a}, 
where we have used the fact that $c_{+}>0$ implies $\theta_2^+-\theta_2^->0$ (see \eqref{sign.c+}). Therefore, if $c_+>0$ and $a\in(0,a_{2})$, we have $h_{1}(a)>0$; and, if $c_+>0$ and $a\geq a_{2}$, we have $h_{1}(a)\leq0$. In sum, we have
\begin{align}
h_1(a)
\begin{cases}
>0\,\,\text{ if }c_+\leq0,\\
>0\,\,\text{ if }c_+>0\text{ and }0<a<a_2,\\
\leq0\,\,\text{ if }c_+>0\text{ and }a\geq a_2.
\end{cases}\nonumber
\end{align}

\begin{itemize}
\item[(2-1)] Suppose $c_+>0$, in which case we have the following conclusions.
\begin{itemize}
\item[(2-1-1)] If $c_+>0$ and $a\geq a_{2}$, by \eqref{g''} we have $g^{\prime\prime}(x)>0$ for all $x>a$.

\item[(2-1-2)] If $c_+>0$ and $0<a<a_{3}\wedge a_{2}$, then $x_0> a$, $g^{\prime\prime}(x)<0$ on $(a,x_0)$, and, $g^{\prime\prime}(x)>0$ on $(x_0,\infty)$.
Actually, if $c_+>0$ and $a\in(0,a_{2})$ (hence, $h_{1}(a)>0$), the function $g^{\prime\prime}(x)$ is strictly increasing with its unique zero $x_0$ given by \eqref{def.x0}.
To check whether or not $x_0$ is greater than $a$, define
\begin{align}\label{def.h}
h_{2}(a)
&:= -g^{\prime\prime}(a+)=
(g^+(0)-c_+g^-(0))(\theta_1^+)^2-(1-c_+)g^-(0)(\theta^+_2)^2
\nonumber\\
&=
-(1-c_-)
\Theta
e^{\theta_1^-a}
+\left[(1-c_-c_+)(\theta_1^+)^2-(1-c_+)c_-(\theta_2^+)^2\right]e^{-\theta_2^-a}.
\end{align}
It follows from $c_+>0$, \eqref{sign.c-}, \eqref{sign.c+}, and, the definition of $\Theta$, that
$$-(1-c_-)\Theta=-(1-c_-)\left[c_+(\theta_1^+)^2+(1-c_+)(\theta_2^+)^2\right]<0,$$
which together with the fact of $h_{2}(0)=(1-c_+)[(\theta_1^+)^2-(\theta_2^+)^2]>0$ yields that
$$(1-c_-c_+)(\theta_1^+)^2-(1-c_+)c_-(\theta_2^+)^2
>(1-c_-)\Theta
>0.$$ 
Hence, the function $\mathbb{R}_{+}\ni x\mapsto h_{2}(x)$ is strictly decreasing and admits a unique zero $a_3>0$ given by \eqref{def.a}.
Hence, if $c_{+}>0$ and $0<a<a_{2}\leq a_{3}$, we have $-g^{\prime\prime}(a+)=h_{2}(a)>0$ on $a\in(0,a_{2})$, which implies $x_0> a$, and hence $g^{\prime\prime}(x)<0$ on $(a,x_0)$ and $g^{\prime\prime}(x)>0$ on $(x_0,\infty)$.
Similarly, if $c_+>0$ and $0<a<a_{3}<a_{2}$, one knows that $-g^{\prime\prime}(a+)=h_{2}(a)>0$ and $x_0> a$, and hence $g^{\prime\prime}(x)<0$ on $(a,x_0)$ and $g^{\prime\prime}(x)>0$ on $(x_0,\infty)$.

\item[(2-1-3)] If $c_+>0$ and $0<a_{3}\leq a\leq a_{2}$, then $g^{\prime\prime}(x)>0$ for all $x>a$.
Actually, in the case $c_+>0$ and $0<a_{3}\leq a\leq a_{2}$, we get $-g^{\prime\prime}(a+)=h_{2}(a)\leq0$, which means $x_0\leq a$.
\end{itemize}

\item[(2-2)] Suppose $c_+\leq0$, in which case the function $g^{\prime\prime}(x)$ is strictly increasing (see \eqref{g''}) with its unique zero $x_0$ given
by \eqref{def.x0}. Let $h_{2}(a)$ be defined by \eqref{def.h}. We have the following conclusions.
\begin{itemize}
\item[(2-2-1)] If $\Theta\leq 0$, $c_+\leq0$ and $a>0$, then $x_0>a$, $g^{\prime\prime}(x)<0$ on $(a,x_0)$, and, $g^{\prime\prime}(x)>0$ on $(x_0,\infty)$.
Indeed, if $\Theta\leq 0$ and $a>0$, it follows from $h_{2}(0)>0$ that $-g^{\prime\prime}(a+)=h_{2}(a)>0$, which means $x_0>a$.

\item[(2-2-2)] If $\Theta>0$, $c_+\leq0$ and $a\in(0,a_{3})$, then $x_0>a$, $g^{\prime\prime}(x)<0$ on $(a,x_0)$, and, $g^{\prime\prime}(x)>0$ on $(x_0,\infty)$.
Indeed, if $\Theta>0$ and $a\in(0,a_{3})$, it follows from $h_{2}(0)>0$ that the function $\mathbb{R}_{+}\ni x\mapsto h_{2}(x)$ is strictly decreasing and $a_{3}>0$ given by \eqref{def.a} is its unique zero. Hence, one knows that $-g^{\prime\prime}(a+)=h_{2}(a)>0$, which yields $x_0>a$.

\item[(2-2-3)] If $\Theta>0$, $c_+\leq0$ and $a\in[a_{3},\infty)$, then $g^{\prime\prime}(x)\geq 0$ for all $x>a$.
Indeed, if $\Theta>0$ and $a\in[a_{3},\infty)$, one knows that $-g^{\prime\prime}(a+)=h_{2}(a)\leq 0$ for $a\in[a_{3},\infty)$, which means $x_0\leq a$.
\end{itemize}
\end{itemize}
\end{itemize}
Putting together all the above arguments leads to the desired result of Proposition \ref{P.0.3}.

\subsection{Proof of Theorem \ref{thm2.1}}
\label{sec5.4}

Let $\mathcal{M}_{\zeta}$ and $\psi(x,y)$ be given respectively by \eqref{M} and \eqref{psi}. 

We first consider Case (i) of Proposition \ref{P.0.3}. We are to prove that $\mathcal{M}_{\zeta}$ is a singleton set and then identify the unique $(z_{1},z_{2})\in\mathcal{M}_{\zeta}$ explicitly.

To start, we characterize the set $\mathcal{N}$ (Actually, if $\mathcal{N}$ is identified to be a singleton, then by Proposition \ref{P.0.4} and the relation $\mathcal{M}_{\zeta}\subseteq \mathcal{N}$ one has $\mathcal{M}_{\zeta}= \mathcal{N}$). For any $(z_{1},z_{2})\in\mathcal{N}$, either $z_{1}=0$ or $z_{1}>0$ holds true.
\begin{itemize}
\item [(1)] We first check whether or not there is a
$(z_{1},z_{2})\in\mathcal{N}$ such that $z_{1}>0$. If there is a $(z_{1},z_{2})\in\mathcal{N}$ with $z_{1}>0$, then we have \eqref{psi=beta} and $g^{\prime}(z_1)=g^{\prime}(z_2)$, which forces us to conclude that $0<z_{1}\leq a\leq z_{2}<\inf\{x\geq a;g^{\prime}(x)\geq g^{\prime}(0)\}$.
Then, it holds that $z_1=(g^{\prime})_-^{-1}(g^{\prime}(z_1))=(g^{\prime})_-^{-1}(g^{\prime}(z_2))$.
Hence, \eqref{psi=beta} can be rewritten as
\begin{align}
\label{psi=beta.1}
\phi(z_2)=\beta,
\end{align}
where the unary function $\phi(x)$ is defined by \eqref{def.phi.} with $a\leq x\leq a_{4}$ (note that $a_{4}$ is guaranteed to be finite since $g^{\prime}(0)$ is finite and $g^{\prime}$ is strictly increasing on $(a,\infty)$ with $g^{\prime}(\infty)=\infty$).
One can verify that
\begin{align}
\phi^{\prime}(x)=\int_{(g^{\prime})_-^{-1}(g^{\prime}(x))}^{x} g^{\prime}(s)\ds\frac{g^{\prime\prime}(x)}{(g^{\prime}(x))^2},\quad x\in(a,a_{4}),\nonumber
\end{align}
which inherits from $g^{\prime\prime}(x)$ the property of being positive on $(a,a_{4})$. That is to say, the unary function $\phi(x)$ defined by \eqref{def.phi.} is continuous and strictly increasing on $[a,a_{4}]$ with $\phi(a)=\psi((g^{\prime})_-^{-1}(g^{\prime}(a)),a)=\psi(a,a)=0$. Hence
\begin{itemize}
\item[(1-1)] if $\phi\left(a_{4}\right)=\psi\left((g^{\prime})_-^{-1}\big(g^{\prime}(a_{4})\big),a_{4}\right)=\psi(0,a_{4})>\beta$, by the intermediate value theorem, we know that there exists a unique $z_{2}=\phi^{-1}(\beta)\in(a,a_{4})$ with $z_{1}=(g^{\prime})_-^{-1}(g^{\prime}(z_2))\in(0,a)$ such that \eqref{psi=beta.1} holds true. Hence, the point $(z_{1},z_{2})$ with $z_{2}=\phi^{-1}(\beta)\in(a,a_{4})$ and $z_{1}=(g^{\prime})_-^{-1}(g^{\prime}(z_2))$ is the unique solution of \eqref{psi=beta} such that $z_{1}>0$ and $g^{\prime}(z_1)=g^{\prime}(z_2)$. Here, $\phi^{-1}$ denotes the well-defined inverse function of $\phi$ given by \eqref{def.phi.}.
\item[(1-2)] if $\phi(a_4)=\psi\left((g^{\prime})_-^{-1}\big(g^{\prime}(a_{4})\big),a_{4}\right)=\psi(0,a_{4})\leq \beta$, there is no solution $(z_{1},z_{2})$ of \eqref{psi=beta} such that $z_{1}>0$ and $g^{\prime}(z_1)=g^{\prime}(z_2)$.
\end{itemize}

\item[(2)] We then check whether or not there is a
$(z_{1},z_{2})\in\mathcal{N}$ such that $z_{1}=0$. If there is a $(z_{1},z_{2})\in\mathcal{N}$ with $z_{1}=0$, then, \eqref{z2.int.} holds true with $z_{1}=0$, that is
\begin{equation}
\label{psi=beta.0}
\psi(0,z_2)=
\begin{cases}
\phi_0(z_2):=
\int_0^{z_2}\left(1-\frac{g'(s)}{g'(z_2)}\right)\ds=\beta,& z_2\in[a,a_4],\\
\phi(z_2)=\beta,&z_2\in[a_4,\infty].
\end{cases}
\end{equation}
It is easy to verify that
the function $(0,\infty)\ni x\mapsto \psi(0,x)$ is strictly increasing on $(a,\infty)$, $\psi(0,a)<0$ and $\psi(0,\infty)=\infty$. Hence
\begin{itemize}
\item[(2-1)] if $\psi(0,a_{4})>\beta$, by the intermediate value theorem, we know that there exists a unique $z_{2}\in(a,a_{4})$ such that \eqref{psi=beta.0} holds true. Hence, the point $(0,z_{2})$ with $z_{2}\in(a,a_{4})$ is the unique solution of \eqref{psi=beta} such that $z_{1}=0$. 

\item[(2-2)] if $\psi(0,a_{4})\leq \beta$, by the intermediate value theorem, we know that there exists a unique $z_{2}\in[a_{4},\infty)$ such that \eqref{psi=beta.0} holds true. Hence, the point $(0,z_{2})$ with $z_{2}\in[a_{4},\infty)$ is the unique solution of \eqref{psi=beta} such that $z_{1}=0$.
\end{itemize}
\end{itemize}
Summing up the above arguments, we arrive at our conclusion as follows.
\begin{itemize}
\item[(a)] If $\psi(0,a_{4})>\beta$, the set $\mathcal{N}$ is composed of two points, i.e.,
$$\mathcal{N}=\{((g^{\prime})_-^{-1}(g^{\prime}(\phi^{-1}(\beta))),\phi^{-1}(\beta)),(0,\phi_{0}^{-1}(\beta))\}.$$
Here, $\phi_0^{-1}$ denotes the well-defined inverse function of $\phi_0$.
Due to the fact that $$g^{\prime}(s)>g^{\prime}(\phi^{-1}(\beta)), \quad\text{for all}\quad s\in[0,(g^{\prime})_-^{-1}(g^{\prime}(\phi^{-1}(\beta)))),$$
one can verify that
\begin{align}
\beta &= \phi(\phi^{-1}(\beta))=
\int_{(g^{\prime})_-^{-1}(g^{\prime}(\phi^{-1}(\beta)))}^{\phi^{-1}(\beta)}\left(1-\frac{g^{\prime}(s)}{g^{\prime}(\phi^{-1}(\beta))}\right)\ds
\nonumber\\
&>
\int_{0}^{\phi^{-1}(\beta)}\left(1-\frac{g^{\prime}(s)}{g^{\prime}(\phi^{-1}(\beta))}\right)\ds=\phi_{0}(\phi^{-1}(\beta)),\nonumber
\end{align}
which implies
\begin{align}
\label{bijiao}
\phi_{0}^{-1}(\beta)>\phi^{-1}(\beta).
\end{align}
Since both points of $\mathcal{N}$ are solutions to \eqref{z2.int.}, by \eqref{bijiao} and the fact that $g^{\prime}(x)$ is strictly increasing on $(a,a_{4})$, one can get
\begin{align}
\zeta(0,\phi_{0}^{-1}(\beta))
&= 1/g^{\prime}(\phi_{0}^{-1}(\beta))
\nonumber\\
&<1/g^{\prime}(\phi^{-1}(\beta))=\zeta((g^{\prime})_-^{-1}(g^{\prime}(\phi^{-1}(\beta))),\phi^{-1}(\beta)),\nonumber
\end{align}
which together with the fact that $\emptyset\neq \mathcal{M}_{\zeta}\subseteq \mathcal{N}$ implies that
$$\mathcal{M}_{\zeta}=\{((g^{\prime})_-^{-1}(g^{\prime}(\phi^{-1}(\beta))),\phi^{-1}(\beta))\}.$$

\item[(b)] If $\psi(0,a_{4})\leq \beta$, the set $\mathcal{N}$ is composed of only one point, i.e.,
$$\mathcal{N}=\{(0,\phi^{-1}(\beta))\}=\{((g')^{-1}_-(\phi^{-1}(\beta)),\phi^{-1}(\beta))\},$$
which combined with the fact that $\emptyset\neq \mathcal{M}_{\zeta}\subseteq \mathcal{N}$ yields that
$$\mathcal{M}_{\zeta}=\{((g')^{-1}_-(\phi^{-1}(\beta)),\phi^{-1}(\beta))\}.$$
\end{itemize}

For Cases (ii) and (iii), one can derive the desired results by adopting a similar argument as the one used for the Case (i).

{
We next discuss Case (iv) of Proposition \ref{P.0.3} in which $g(x)$ is concave on $(0,a_1)$, convex on $(a_1,a)$, concave on $(a,x_0)$ and convex on $(x_0,\infty)$. Let $x_{1}$ and $x_{2}$ be defined by \eqref{x1.def}-\eqref{x2.def}. To simplify the analysis, we show the following six claims.

\begin{itemize}
\item[(1)] 
$\{(z_1,z_2)\in\mathcal{N}: z_2\in(a, x_0]\cup [0,a_1]\} \cap \mathcal{M}_{\zeta}= \emptyset$.
\item[(2)] $\{(z_1,z_2)\in\mathcal{N}: z_2\in ((g')^{-1}_2( g'(x_2)),a] \} \cap \mathcal{M}_{\zeta}= \emptyset$.
\item[(3)] $\{(z_1,z_2)\in\mathcal{N}: z_2\in [a_1,(g')^{-1}_2( g'(x_2))], z_1\neq (g')^{-1}_1(g'(z_2)) \} \cap \mathcal{M}_{\zeta}= \emptyset$.
\item [(4)] 
$\{(z_1,z_2)\in\mathcal{N}: z_2\in [x_0,x_{1}),z_1\neq (g')^{-1}_3(g'(z_2))\} \cap \mathcal{M}_{\zeta}= \emptyset$.
\item[(5)] 
$\{(z_1,z_2)\in\mathcal{N}: z_2=x_1,z_1\notin\{ (g')^{-1}_1(g'(z_2)), (g')^{-1}_3(g'(z_2))\}
\} \cap \mathcal{M}_{\zeta}= \emptyset$.
\item[(6)] 
$\{(z_1,z_2)\in\mathcal{N}: z_2\in (x_1,\infty),z_1\neq (g')^{-1}_1(g'(z_2)) 
\} \cap \mathcal{M}_{\zeta}= \emptyset$.
\end{itemize}
Obviously, $\{(z_1,z_2)\in\mathcal{N}: z_2\in[0,a_1]\}= \emptyset$. To prove claim (1), assume that $(z_1,z_2)\in\mathcal{N}$ 
is such that $z_2\in (a, x_0]$. By the definition of $\mathcal{N}$, $z_1\in[0,a_1)$.
Due to the fact of
$$g^{\prime}(s)>g^{\prime}(z_2),\quad\text{for all}\quad s\in((g^{\prime})^{-1}_2(g^{\prime}(z_2)),z_2) \cup [0,(g^{\prime})^{-1}_1(g^{\prime}(z_2))),$$ 
one can verify that
\begin{align}\label{56}
\beta
&=
\int_{z_1}^{z_2}\Big(1-\frac{g^{\prime}(s)}{g^{\prime}(z_2)}\Big)\ds
\nonumber\\
&= 
\Big(\int_{z_1}^{(g^{\prime})^{-1}_1(g^{\prime}(z_2))}
+\int_{(g^{\prime})^{-1}_1(g^{\prime}(z_2))}^{(g^{\prime})^{-1}_2(g^{\prime}(z_2))}
+\int_{(g^{\prime})^{-1}_2(g^{\prime}(z_2))}^{z_2}\Big)\Big(1-\frac{g^{\prime}(s)}{g^{\prime}(z_2)}\Big)\ds
\nonumber\\
&< 
\int_{(g^{\prime})^{-1}_1(g^{\prime}(z_2))}^{(g^{\prime})^{-1}_2(g^{\prime}(z_2))}\Big(1-\frac{g^{\prime}(s)}{g^{\prime}(z_2)}\Big)\ds=\psi((g^{\prime})^{-1}_1(g^{\prime}(z_2)),(g^{\prime})^{-1}_2(g^{\prime}(z_2))),
\end{align}
which implies that there exists a $(z'_1,z'_2)\in\mathcal{N}$ such that $0\leq z'_1<a_1<z'_2\leq a$ and $z'_1=(g^{\prime})^{-1}_1(g^{\prime}(z'_2))$. 
By \eqref{56}, it holds that
\begin{align}
z'_2<(g^{\prime})^{-1}_2(g^{\prime}(z_2)).
\end{align}
Then, by the fact that $g^{\prime}(x)$ is strictly increasing on $(a_1,a)$, one can get
\begin{align}
\zeta(z'_1,z'_2)=1/g^{\prime}(z'_2)>1/g^{\prime}(z_2)=\zeta(z_1,z_2),
\end{align}
which means that any $(z_1,z_2)\in\mathcal{N}$ such that $z_2\in(a, x_0]$ satisfies $(z_1,z_2)\notin\mathcal{M}_{\zeta}$. Hence, (1) holds true. The claims (2)-(6) can be proved by similar arguments combined with the definition of $x_1$ and $x_2$. We hence omit their proofs.
By the above claims (1)-(6), we know that $\mathcal{M}_{\zeta}\subseteq \cup_{i=1}^{4}\mathcal{R}_i$, where
\begin{align*}
&\mathcal{R}_1:=\{(z_1,z_2)\in \mathcal{N}:z_2\in [a_1,(g')^{-1}_2( g'(x_2))],
z_1=(g')^{-1}_1(g'(z_2))\}, \\
&\mathcal{R}_2:=\{(z_1,z_2)\in \mathcal{N}:z_2\in [x_0,x_{1}),z_1= (g')^{-1}_3(g'(z_2))\}, \\
&\mathcal{R}_3:=\{(z_1,z_2)\in \mathcal{N}: z_2=x_1,z_1\in\{ (g')^{-1}_1(g'(z_2)), (g')^{-1}_3(g'(z_2))\} \},
\\
&\mathcal{R}_4:=\{(z_1,z_2)\in \mathcal{N}: z_2\in (x_1,\infty),z_1= (g')^{-1}_1(g'(z_2)) \}.
\end{align*}
The forms of $(\mathcal{R}_{i})_{1\leq i\leq 4}$ motivate us to define $\omega_{1}$ and $\omega_{2}$ through \eqref{addnewomega1} and \eqref{addnewomega2}.
Then, we have
\begin{align*}
&\mathcal{R}_1=
\begin{cases}
 \{(\tilde{z}_1,\tilde{z}_2)\},\,\text{ if } \beta\leq \omega_{1}(x_{2}),\\
 \emptyset, \text{ if else}.
\end{cases} \\
&\cup_{i=2}^{4}\mathcal{R}_i=\begin{cases}
 \{(\bar{z}_1,\bar{z}_2)\},\,\text{ if } \beta<\omega_{2}(x_1),\\
 \{(\hat{z}_1:= (g')^{-1}_1(g'(\omega_{2}^{-1}(\beta))),\hat{z}_2:=\omega_{2}^{-1}(\beta))\},\,\text{ if } \beta>\omega_{2}(x_1),
 \\
 \{(\bar{z}_1,\bar{z}_2),(\hat{z}_1,\hat{z}_2)\},\,\text{ if } \beta=\omega_{2}(x_1).
\end{cases}
\end{align*}
Therefore, we have the following observations:
\begin{itemize}
\item If $g^{\prime}(\omega_1^{-1}(\beta))<g^{\prime}(\omega_{2}^{-1}(\beta))$ and $\beta<\omega_2(x_1)$, then
$\emptyset\neq \mathcal{M}_{\zeta}\subseteq \{(\tilde{z}_{1},\tilde{z}_{2}),(\bar{z}_{1},\bar{z}_{2})\}$, and
\begin{align}
\zeta(\tilde{z}_{1},\tilde{z}_{2})=1/g'(\omega_{1}^{-1}(\beta))
 >
 1/g'(\omega_{2}^{-1}(\beta))
 =\zeta(\bar{z}_{1},\bar{z}_{2}).
 \nonumber
\end{align}
Hence, $\mathcal{M}_{\zeta}=\{(\tilde{z}_{1},\tilde{z}_{2})\}$.
\item If $g^{\prime}(\omega_1^{-1}(\beta))> g^{\prime}(\omega_{2}^{-1}(\beta))$ and $\beta< \omega_{2}(x_1)$, then $\emptyset\neq \mathcal{M}_{\zeta}\subseteq \{(\tilde{z}_{1},\tilde{z}_{2}),(\bar{z}_{1},\bar{z}_{2})\}$, and
\begin{align}
\zeta(\tilde{z}_{1},\tilde{z}_{2})=1/g'(\omega_{1}^{-1}(\beta))
 < 
 1/g'(\omega_{2}^{-1}(\beta))
 =\zeta(\bar{z}_{1},\bar{z}_{2}).
 \nonumber
\end{align}
Hence, $\mathcal{M}_{\zeta}=\{(\bar{z}_{1},\bar{z}_{2})\}$.
\item 
If $g^{\prime}(\omega_1^{-1}(\beta))= g^{\prime}(\omega_{2}^{-1}(\beta))$ and $\beta< \omega_{2}(x_1)$, then $\emptyset\neq \mathcal{M}_{\zeta}\subseteq \{(\tilde{z}_{1},\tilde{z}_{2}),(\bar{z}_{1},\bar{z}_{2})\}$, 
and
\begin{align}
\zeta(\tilde{z}_{1},\tilde{z}_{2})=1/g'(\omega_{1}^{-1}(\beta))
 = 
 1/g'(\omega_{2}^{-1}(\beta))
 =\zeta(\bar{z}_{1},\bar{z}_{2}).
 \nonumber
\end{align}
Hence, $\mathcal{M}_{\zeta}=\{(\tilde{z}_{1},\tilde{z}_{2}),(\bar{z}_{1},\bar{z}_{2})\}$.

\item $g^{\prime}(\omega_1^{-1}(\beta))> g^{\prime}(\omega_{2}^{-1}(\beta))$ and $\beta=\omega_2(x_1)$ can not hold true simultaneously.
\item If $g^{\prime}(\omega_1^{-1}(\beta))= g^{\prime}(\omega_{2}^{-1}(\beta))$ and $\beta=\omega_2(x_1)$, then $x_2\leq x_1$, $\emptyset\neq \mathcal{M}_{\zeta}\subseteq \{(\tilde{z}_{1},\tilde{z}_{2}),(\bar{z}_{1},\bar{z}_{2}), (\hat{z}_1,\hat{z}_2)\}$.
\begin{itemize}
 \item[$\bullet$] If $x_{2}<x_{1}$, then $\hat{z}_1=\tilde{z}_1$ and $\hat{z}_2=\tilde{z}_2$ (since $\omega_1\equiv\omega_2$ on $[x_1,\infty)$), and 
\begin{align}
\zeta(\bar{z}_1,\bar{z}_2)=1/g'(\omega^{-1}_2(\beta))=1/g'(\omega^{-1}_1(\beta))=\zeta(\tilde{z}_1,\tilde{z}_2).\nonumber
\end{align}
Hence, $\mathcal{M}_{\zeta}=\{(\tilde{z}_1,\tilde{z}_2),(\bar{z}_1,\bar{z}_2)\}.$
 \item[$\bullet$] If $x_{2}=x_{1}$, then $\beta=\omega_{1}(x_{2})$ and $(\tilde{z}_1,\tilde{z}_2)$ is understood as a two-point set $\{(\tilde{z}_1,(g^{\prime})_{2}^{-1}( g^{\prime}(x_2))),(\tilde{z}_1,x_2)\}$
 with $\tilde{z}_1=(g^{\prime})_{1}^{-1}(g^{\prime}(\omega_{1}^{-1}(\beta)))=(g^{\prime})_{1}^{-1}(g^{\prime}(\omega_{2}^{-1}(\beta)))=\hat{z}_1.$ In addition, 
 \begin{align}
\zeta(\bar{z}_1,\bar{z}_2)=1/g'(x_{2})=\zeta(\tilde{z}_1,(g^{\prime})_{2}^{-1}( g^{\prime}(x_2)))=\zeta(\tilde{z}_1,x_2).\nonumber
\end{align}
Hence, $\mathcal{M}_{\zeta}=\{(\tilde{z}_1,\tilde{z}_2),(\bar{z}_1,\bar{z}_2)\}=\{(\tilde{z}_1,(g^{\prime})_{2}^{-1}( g^{\prime}(x_2))),(\tilde{z}_1,x_2),(\bar{z}_1,\bar{z}_2)\}.$
\end{itemize}

\item If $g^{\prime}(\omega_1^{-1}(\beta))< g^{\prime}(\omega_{2}^{-1}(\beta))$ and $\beta=\omega_2(x_1)$, then $x_1<x_2$, $\emptyset\neq\mathcal{M}_{\zeta}\subseteq\{(\tilde{z}_1,\tilde{z}_2),(\bar{z}_1,\bar{z}_2),(\hat{z}_1,\hat{z}_2)\}$, and 
\begin{align}
\zeta(\tilde{z}_1,\tilde{z}_2)=1/g'(\omega^{-1}_1(\beta))>1/g'(\omega^{-1}_2(\beta))=\zeta(\bar{z}_1,\bar{z}_2)=\zeta(\hat{z}_1,\hat{z}_2).\nonumber
\end{align}
Hence, $\mathcal{M}_{\zeta}=\{(\tilde{z}_1,\tilde{z}_2)\}$.

\item 
If either of the following cases is the case 
\begin{itemize}
 \item[$\bullet$] $\beta> \omega_{2}(x_1)$ and $x_{1}\geq x_{2}$,
 \item[$\bullet$] $\beta>\omega_{2}(x_1)$, $x_{1}< x_{2}$, and $\hat{z}_{2}\geq x_{2}$,
\end{itemize}
then $g'(\omega_{1}^{-1}(\beta))=g'(\omega_{2}^{-1}(\beta))$ (since $\omega_{1}\equiv \omega_{2}$ on $[x_{1}\vee x_2,\infty)$), and then $$\hat{z}_{1}=\tilde{z}_{1},\hat{z}_{2}=\tilde{z}_{2},$$
i.e., $\cup_{i=1}^{4}\mathcal{R}_i=\{(\tilde{z}_{1},\tilde{z}_{2})\}$. Hence,
$\mathcal{M}_{\zeta}=\{(\tilde{z}_{1},\tilde{z}_{2})\}$.
\item If $\beta>\omega_{2}(x_1)$, $x_{1}< x_{2}$, and $\hat{z}_{2}<x_{2}$, then
\begin{align}
 &\omega_{1}(x_{2})=\omega_{2}(x_{2})\geq \omega_{2}(\hat{z}_2)=\beta> \omega_2(x_1),\nonumber
\end{align}
and
\begin{align}
\beta=\omega_{2}(\hat{z}_2)=&\Big(\int_{\hat{z}_1}^{(g')_{2}^{-1}(g'(\hat{z}_2))}+\int_{(g')_{2}^{-1}(g'(\hat{z}_2))}^{\hat{z}_{2}}\Big)\Big(1-\frac{g^{\prime}(s)}{g^{\prime}(\hat{z}_2)}\Big)\ds\nonumber\\
< &
\int_{\hat{z}_1}^{(g')_{2}^{-1}(g'(\hat{z}_2))}\Big(1-\frac{g^{\prime}(s)}{g^{\prime}(\hat{z}_2)}\Big)\ds \quad (\text{since } \hat{z}_{2}<x_{2})\nonumber\\
=& \omega_{1}((g')_{2}^{-1}(g'(\hat{z}_2))),\nonumber
\end{align}
which implies that
\begin{align}
 \omega_{1}^{-1}(\beta) < (g')_{2}^{-1}(g'(\hat{z}_2)),\nonumber
\end{align}
that is
$g'( \omega_{1}^{-1}(\beta))<g'(\hat{z}_{2})=g'(\omega^{-1}_2(\beta))$. Therefore
\begin{align}
 \zeta(\tilde{z}_1,\tilde{z}_{2})=1/g'(\tilde{z}_2)>1/g'(\hat{z}_{2})=\zeta(\hat{z}_1,\hat{z}_{2}).\nonumber
\end{align}
Hence,
$\mathcal{M}_{\zeta}=\{(\tilde{z}_{1},\tilde{z}_{2})\}$.
\end{itemize}
The proof is complete.}

\subsection{Proof of Proposition \ref{prop.4.1}}
\label{sec5.5}

Recall that the function $g^{\prime\prime}(x)$ is given by
\begin{align}
\label{g''.v3}
g^{\prime\prime}(x)
&= 
\left[(1-c_{+})g^-(0)(\theta_2^+)^2e^{\theta_2^+{(x-a)}}-(g^+(0)-c_+g^-(0))(\theta_1^+)^2e^{-\theta_1^+{(x-a)}} \right]\mathbf{1}_{\{x>a\}}
\nonumber\\
&
+\left[(g^-(0)-c_{-}g^+(0))(\theta_2^-)^2e^{\theta_2^-{(x-a)}}-(1-c_-)g^+(0)(\theta_1^-)^2e^{-\theta_1^-{(x-a)}}\right]\mathbf{1}_{\{x<a\}}.
\end{align}
We first discuss the sign of $g^{\prime\prime}(x)$ for $x\in(0,a)$.
When $x\in (0,a)$, using \eqref{scale.fun.} and \eqref{g''.v3} we have
\begin{align}\label{g''.v4}
g^{\prime\prime}(x)
&=
(g^-(0)-c_{-}g^+(0))(\theta_2^-)^2e^{\theta_2^-{(x-a)}}-(1-c_-)g^+(0)(\theta_1^-)^2e^{-\theta_1^-{(x-a)}}
\nonumber\\
&=
(1-c_-)e^{(\theta_1^--\theta_2^-)a}\left[(\theta_2^-)^2e^{\theta_2^-x}-(\theta_1^-)^2e^{-\theta_1^-x}\right].
\end{align}
Since $\theta^-_2\geq \theta^-_1$, one sees that the function
$$\mathbb{R}_{+}\ni x\mapsto (1-c_-)e^{(\theta_1^--\theta_2^-)a}\left[(\theta_2^-)^2e^{\theta_2^-x}-(\theta_1^-)^2e^{-\theta_1^-x}\right],$$
is non-negative at $x=0$ and is strictly increasing on $(0,a)$. Hence, one has $g^{\prime\prime}(x)>0$ on $(0,a]$.
In the sequel, we discuss the sign of $g^{\prime\prime}(x)$ for $x>a$.
When $x>a$, by \eqref{scale.fun.} and \eqref{g''.v3}, it holds that
\begin{align}\label{g''.v5}
g^{\prime\prime}(x)
&= 
(1-c_{+})g^-(0)(\theta_2^+)^2e^{\theta_2^+{(x-a)}}-(g^+(0)-c_+g^-(0))(\theta_1^+)^2e^{-\theta_1^+{(x-a)}},
\end{align}
and
\begin{align}\label{g''(a)3}
g^{\prime\prime}(a)
&= 
(1-c_{+})g^-(0)(\theta_2^+)^2-(g^+(0)-c_+g^-(0))(\theta_1^+)^2
\nonumber\\
&\geq
[(1-c_{+})g^-(0)-(g^+(0)-c_+g^-(0))](\theta_1^+)^2
\nonumber\\
&= 
[g^-(0)-g^+(0)](\theta_1^+)^2>0.
\end{align}
It is seen that the function $g^{\prime\prime}(x)$ is strictly increasing on $(a,\infty)$, which together with \eqref{g''(a)3} implies that $g^{\prime\prime}(x)>0$ on $(a,\infty)$. The proof is complete.

\subsection{Proof of Proposition \ref{prop.4.3}}
\label{sec5.6}

Recall that the the function $g^{\prime\prime}(x)$ is given by \eqref{g''.v3}.
\begin{itemize}
\item [(1)]
We first discuss the sign of $g^{\prime\prime}(x)$ for $x\in(0,a)$.
When $x\in (0,a)$, we have \eqref{g''.v4} holds.
It is seen that the function
$$\mathbb{R}_{+}\ni x\mapsto (1-c_-)e^{(\theta_1^--\theta_2^-)a}\left[(\theta_2^-)^2e^{\theta_2^-x}-(\theta_1^-)^2e^{-\theta_1^-x}\right],$$
is strictly increasing with its unique zero $a_1>0$ given by \eqref{def.a}, where we have used the fact that $\theta_1^->\theta_2^-$.
\begin{itemize}
\item [(1-1)]If $0<a\leq a_{1}$, one has $g^{\prime\prime}(x)<0$ on $(0,a)$.
\item[(1-2)] If $a>a_{1}$, one has $g^{\prime\prime}(x)<0$ on $[0,a_{1})$ and $g^{\prime\prime}(x)>0$ on $(a_{1},a]$.
\end{itemize}
\item[(2)] In the sequel, we discuss the sign of $g^{\prime\prime}(x)$ for $x>a$.
When $x>a$, we have \eqref{g''.v5} 
and
\begin{align}\label{g''(a)4}
g^{\prime\prime}(a)
&= 
(1-c_{+})g^-(0)(\theta_2^+)^2-(g^+(0)-c_+g^-(0))(\theta_1^+)^2
\nonumber\\
&\geq 
[(1-c_{+})g^-(0)-(g^+(0)-c_+g^-(0))](\theta_1^+)^2
\nonumber\\
&= 
[g^-(0)-g^+(0)](\theta_1^+)^2>0,
\end{align}
where we have used the facts that $\theta_1^+\leq \theta_2^+$ and $g^-(0)>g^+(0)>0$.
It is seen that the function $g^{\prime\prime}(x)$ is strictly increasing on $(a,\infty)$, which together with \eqref{g''(a)4} implies that $g^{\prime\prime}(x)>0$ on $(a,\infty)$.
\end{itemize}
Putting together all the above arguments leads to the desired result of Proposition \ref{prop.4.3}.

\subsection{Proof of Proposition \ref{prop.4.5}}\label{sec5.7}
Using similar arguments as those in the proof of Proposition \ref{prop.4.3}, one has the following observations. 
\begin{itemize}
\item [(1)]
We have $g^{\prime\prime}(x)>0$ on $(0,a)$.
\item [(2)]In the sequel, we discuss the sign of $g^{\prime\prime}(x)$ for $x>a$.

\begin{itemize}
\item [(2-1)] Suppose $c_+>0$, in which case we have the following conclusions.
\begin{itemize}
\item[(2-1-1)] If $c_+>0$ and $a\geq a_{2}$, by \eqref{g''.v5} we have $g^{\prime\prime}(x)>0$ for all $x>a$.
\item [(2-1-2)] If $c_+>0$ and $0<a<a_{3}\wedge a_{2}$, then $x_0> a$, $g^{\prime\prime}(x)<0$ on $(a,x_0)$, and, $g^{\prime\prime}(x)>0$ on $(x_0,\infty)$.
\item [(2-1-3)] If $c_+>0$ and $0<a_{3}\leq a\leq a_{2}$, then $g^{\prime\prime}(x)>0$ for all $x>a$. 
\end{itemize}
\item [(2-2)] Suppose $c_+\leq0$, in which case
the function $g^{\prime\prime}(x)$ is strictly increasing with its unique zero $x_0$ given by \eqref{def.x0}. We have the following conclusions.
\begin{itemize}
\item[(2-2-1)] If $\Theta\leq 0$, $c_+\leq0$ and $a>0$, then $x_0>a$, $g^{\prime\prime}(x)<0$ on $(a,x_0)$, and, $g^{\prime\prime}(x)>0$ on $(x_0,\infty)$. 
\item[(2-2-2)] If
$\Theta>0$, 
$c_+\leq0$ and $a\in(0,a_{3})$, then $x_0>a$, $g^{\prime\prime}(x)<0$ on $(a,x_0)$, and, $g^{\prime\prime}(x)>0$ on $(x_0,\infty)$. 
\end{itemize}

\begin{itemize}
\item[(2-2-3)] If
$\Theta>0$, 
$c_+\leq0$ and $a\in[a_{3},\infty)$, then $g^{\prime\prime}(x)\geq 0$ for all $x>a$. 
\end{itemize}
\end{itemize}
\end{itemize}
Putting together all the above arguments leads to the desired result of Proposition \ref{prop.4.5}.

\subsection{Proof of Proposition \ref{prop.2.4.add}}
\label{sec5.8}

We exclusively present the proof for case (iv) of Subsection \ref{sec3.1}, as proofs for the remaining cases of Subsections \ref{sec3.1}-\ref{sec3.4} adhere to patterns that are either similar or simpler in nature.
Recalling that $a_{1}=(g^{\prime})_{2}^{-1}g^{\prime}(x_{2})$ and $x_{0}=x_{1}$ can not happen simultaneously and then using Theorem \ref{thm2.1}, one obtains
\begin{itemize}
\item If $a_{1}<(g^{\prime})_{2}^{-1}g^{\prime}(x_{2})$ and $x_{0}<x_{1}$, then, for sufficiently small $\beta$, we have
\begin{equation*}
\mathcal{M}_{\zeta}=
\begin{cases}
\{(\tilde{z}_1,\tilde{z}_2)\}, & \text{if $g^{\prime}(\omega_1^{-1}(\beta))<g^{\prime}(\omega_{2}^{-1}(\beta))$, 
}\\
\{(\Bar{z}_1,\Bar{z}_1)\}, & \text{if $g^{\prime}(\omega_1^{-1}(\beta))> g^{\prime}(\omega_{2}^{-1}(\beta))$,}
\\
\{(\tilde{z}_1,\tilde{z}_2),(\Bar{z}_1,\Bar{z}_1)\}, & \text{if $g^{\prime}(\omega_1^{-1}(\beta))=g^{\prime}(\omega_{2}^{-1}(\beta))$, }
\end{cases}
\end{equation*}
where $(\tilde{z}_1,\tilde{z}_2)$ and $(\Bar{z}_1,\Bar{z}_1)$ are defined in Theorem \ref{thm2.1},
and
\begin{align}
\label{addnew204}
\lim_{\beta\rightarrow0+}\tilde{z}_{2}=\lim_{\beta\rightarrow0+}\tilde{z}_{1}
=
a_1,\quad 
\lim_{\beta\rightarrow0+}\Bar{z}_{2}=\lim_{\beta\rightarrow0+}\Bar{z}_{1}
=
x_0.
\end{align}
\item If $a_{1}=(g^{\prime})_{2}^{-1}g^{\prime}(x_{2})$ and $x_{0}<x_{1}$, then, for sufficiently small $\beta$, we have
\begin{equation*}
\mathcal{M}_{\zeta}=
\{((g^{\prime})_{3}^{-1}(g^{\prime}(\omega_{2}^{-1}(\beta))),\omega_{2}^{-1}(\beta))\}, 
\end{equation*}
with the second equality of \eqref{addnew204} holds true.
\item If $a_{1}<(g^{\prime})_{2}^{-1}g^{\prime}(x_{2})$ and $x_{0}=x_{1}$, then, for sufficiently small $\beta$, we have
\begin{equation*}
\mathcal{M}_{\zeta}=
\{((g^{\prime})_{1}^{-1}(g^{\prime}(\omega_1^{-1}(\beta))),\omega_1^{-1}(\beta))\},
\end{equation*}
with \eqref{addnew204} holds true.
\end{itemize}
Therefore, $\lim_{\beta\rightarrow0+}\max_{(z_1,z_2)\in\mathcal{M}_{\zeta}}(z_2-z_1)=0$ as desired.

\subsection{Proof of Proposition \ref{prop.2.5.add}}
\label{sec5.9}

We exclusively present the proof for case (iv) of Subsection \ref{sec3.1}, as proofs for the remaining cases of Subsections
\ref{sec3.1}-\ref{sec3.4} adhere to patterns that are either similar or simpler in nature.
Recall that $\omega_{1}$ coincides with $\omega_{2}$ on $[x_{1}\vee x_{2},\infty)$. By Theorem \ref{thm2.1}, one has
\begin{itemize}
\item[(a)] If $x_{2}\leq x_{1}$ and $g^{\prime}(a_{1})<g^{\prime}(x_{0})$, then
\begin{align}
z_{2}-z_{1}=
\begin{cases}
\omega_{2}^{-1}(\beta)-(g^{\prime})_{3}^{-1}(g^{\prime}(\omega_{2}^{-1}(\beta))), & \beta\in [\omega_{1}(x_{2}),\omega_{2}(x_{1})),\\
\omega_{2}^{-1}(\beta)-(g^{\prime})_{1}^{-1}(g^{\prime}(\omega_{2}^{-1}(\beta))),&\beta\in[\omega_{2}(x_{1}),\infty),\\
\omega_{1}^{-1}(\beta)-(g^{\prime})_{1}^{-1}(g^{\prime}(\omega_{1}^{-1}(\beta))),& \beta\in[0,\omega_{1}((g^{\prime})_{2}^{-1}(g^{\prime}(x_{0})))),
\end{cases}
\end{align}
which implies that $z_{2}-z_{1}$ is increasing in $\beta$ on $[0,\omega_{1}((g^{\prime})_{2}^{-1}(g^{\prime}(x_{0}))))\cup[\omega_{2}(x_{2}),\infty)$.

\item[(b)] If $x_{2}> x_{1}$ and $g^{\prime}(a_{1})<g^{\prime}(x_{0})$, then
\begin{align}
z_{2}-z_{1}=
\omega_{1}^{-1}(\beta)-(g^{\prime})_{1}^{-1}(g^{\prime}(\omega_{1}^{-1}(\beta))), \quad \beta\in [0,\omega_{1}((g^{\prime})_{2}^{-1}(g^{\prime}(x_{0}))))\cup[\omega_{1}(x_{1}),\infty),
\end{align}
which implies that $z_{2}-z_{1}$ is increasing in $\beta$ on $[0,\omega_{1}((g^{\prime})_{2}^{-1}(g^{\prime}(x_{0}))))\cup[\omega_{1}(x_{1}),\infty)$.

\item[(c)] If $x_{2}\leq x_{1}$ and $g^{\prime}(a_{1})\geq g^{\prime}(x_{0})$, then
\begin{align}
z_{2}-z_{1}=
\begin{cases}
\omega_{2}^{-1}(\beta)-(g^{\prime})_{3}^{-1}(g^{\prime}(\omega_{2}^{-1}(\beta))), & \beta\in [0,\omega_{2}(a_{5})]\cup[\omega_{2}(x_{2}),\omega_{2}(x_{1})),\\
\omega_{2}^{-1}(\beta)-(g^{\prime})_{1}^{-1}(g^{\prime}(\omega_{2}^{-1}(\beta))),&\beta\in[\omega_{2}(x_{1}),\infty),
\end{cases}
\end{align}
which implies that $z_{2}-z_{1}$ is increasing in $\beta$ on $[0,\omega_{2}(a_{5})\cup[\omega_{2}(x_{2}),\infty)$.

\item[(d)] If $x_{2}> x_{1}$ and $g^{\prime}(a_{1})\geq g^{\prime}(x_{0})$, then
\begin{align}
z_{2}-z_{1}=
\begin{cases}
\omega_{2}^{-1}(\beta)-(g^{\prime})_{3}^{-1}(g^{\prime}(\omega_{2}^{-1}(\beta))), & \beta\in [0,\omega_{2}(a_{5})],
\\
\omega_{1}^{-1}(\beta)-(g^{\prime})_{1}^{-1}(g^{\prime}(\omega_{1}^{-1}(\beta))), & \beta\in [\omega_{2}(x_{1}),\infty),
\end{cases}
\end{align}
which implies that $z_{2}-z_{1}$ is increasing in $\beta$ on $[0,\omega_{2}(a_{5})]\cup[\omega_{2}(x_{1}),\infty)$.
\end{itemize}

In the sequel, we only check the increasing property of $\max _{(z_1,z_2)\in\mathcal{M}_{\zeta}}(z_2-z_1)$ and $\min_{(z_1,z_2)\in\mathcal{M}_{\zeta}}(z_2-z_1)$ (with respect to $\beta$) over the interval $[\omega_{1}((g^{\prime})_{2}^{-1}(g^{\prime}(x_{0}))), \omega_{1}(x_{2})]$ under the above Case (a), because the corresponding proofs needed for the above Cases (b)-(d) are quite similar.

When $\beta\in[\omega_1((g')^{-1}_2(g'(x_0))),\omega_{1}(x_2)]$, by Theorem \ref{thm2.1}, for $(z_1,z_2)\in\mathcal{M}_{\zeta}$, we have
\begin{equation}\label{eq:prop2.5:a}
z_2-z_1=
\begin{cases}
\omega_1^{-1}(\beta)-(g^{\prime})_{1}^{-1}(g^{\prime}(\omega_1^{-1}(\beta))), & \text{if $g^{\prime}(\omega_1^{-1}(\beta))<g^{\prime}(\omega_{2}^{-1}(\beta))$,
}
\\
\omega_{2}^{-1}(\beta)-(g^{\prime})_{3}^{-1}(g^{\prime}(\omega_{2}^{-1}(\beta))), & \text{if $g^{\prime}(\omega_1^{-1}(\beta)) > g^{\prime}(\omega_{2}^{-1}(\beta))$,
}
\\
\omega_{2}^{-1}(\beta)-(g^{\prime})_{3}^{-1}(g^{\prime}(\omega_{2}^{-1}(\beta)))\text{ or } 
\\
\omega_1^{-1}(\beta)-(g^{\prime})_{1}^{-1}(g^{\prime}(\omega_1^{-1}(\beta))), & \text{if $g^{\prime}(\omega_1^{-1}(\beta))=g^{\prime}(\omega_{2}^{-1}(\beta))$,
}
\end{cases}
\end{equation}
Denote three sets as
$$\Lambda^{+(-,0)}:=\{\beta\in[\omega_1((g')^{-1}_2(g'(x_0))),\omega_{1}(x_2)];\,g'(\omega_1^{-1}(\beta))>(<,=)g'(\omega^{-1}_2(\beta))\}.$$
Hence, it holds that $[\omega_1((g')^{-1}_2(g'(x_0))),\omega_{1}(x_2)]=\Lambda^+\cup\Lambda^-\cup\Lambda^0.$ 
We next demonstrate that $\Lambda^{\pm}$ can each be expressed as a countable union of disjoint open intervals, and, $\Lambda^0$ can be expressed as a countable union of disjoint closed intervals. For any $\beta_0\in\Lambda^+$, we have $g'(\omega^{-1}_1(\beta_0))>g'(\omega^{-1}_2(\beta_0))$. According to the sign-preserving property of continuous functions, there exists a constant $\epsilon>0$ such that 
$(\beta_0-\epsilon,\beta_0+\epsilon)\subseteq [\omega_1((g')^{-1}_2(g'(x_0))),\omega_2(x_2)]$ and for any $\beta\in(\beta_0-\epsilon,\beta_0+\epsilon)$, it holds that $g'(\omega^{-1}_1(\beta))>g'(\omega^{-1}_2(\beta))$.
Denote 
\begin{align}\label{over.ep}
\overline{\epsilon}(\beta_0)
&:= 
\sup\{{\overline{\epsilon}>0}; \text{ it holds that } (\beta_0-\epsilon,\beta_0+\overline{\epsilon})\subseteq [\omega_1((g')^{-1}_2(g'(x_0))),\omega_{1}(x_2)], 
\nonumber\\
&\hspace{0.5cm}
\text{and, } g'(\omega^{-1}_1(\beta))>g'(\omega^{-1}_2(\beta))
\text{ for any }\beta\in(\beta_0-\epsilon,\beta_0+\overline{\epsilon})\},
\end{align}
and
\begin{align}\label{under.ep}
\underline{\epsilon}(\beta_0)
&:= 
\sup\{{\underline{\epsilon}>0}; \text{ it holds that }(\beta_0-\underline{\epsilon},\beta_0+\epsilon)\subseteq [\omega_1((g')^{-1}_2(g'(x_0))),\omega_{1}(x_2)], 
\nonumber\\
&\hspace{0.5cm}
\text{and, }g'(\omega^{-1}_1(\beta))>g'(\omega^{-1}_2(\beta)) \text{ for any }\beta\in(\beta_0-\underline{\epsilon},\beta_0+\epsilon)\}.
\end{align}
With the continuity of $g'$ on $\mathbb{R}_+$ and $(\omega^{-1}_i)_{i=1,2}$ on $[\omega_1((g')^{-1}_2(g'(x_0))),\omega_{1}(x_2)]$, we derive that $g'(\omega^{-1}_1(\beta_0-\underline{\epsilon}(\beta_0)))=g'(\omega^{-1}_2(\beta_0-\underline{\epsilon}(\beta_0)))$ and $g'(\omega^{-1}_1(\beta_0+\overline{\epsilon}(\beta_0)))=g'(\omega^{-1}_2(\beta_0+\overline{\epsilon}(\beta_0)))$. Additionally, it holds that $\beta_0-\underline{\epsilon}(\beta_0)\leq \beta_0-\epsilon<\beta_0<\beta_0+\epsilon\leq \beta_0+\overline{\epsilon}(\beta_0),$
and $(\beta_0-\underline{\epsilon}(\beta_0),\beta_0+\overline{\epsilon}(\beta_0))$ is a non-empty subset of $\Lambda^+$. In this vein, for any $\beta_0\in\Lambda^+$, one can find a non-empty open subset of $\Lambda^+$ that contains $\beta_0$.
In addition, according to the 
denseness of rational numbers in the set of real numbers, there is at least one rational number in $(\beta_0-\underline{\epsilon}(\beta_0),\beta_0+\overline{\epsilon}(\beta_0))$. However, rational numbers are countable, implying that $\Lambda^+$ is a countable union of disjoint open intervals. Using similar arguments, one can obtain that $\Lambda^-$ is also a countable union of disjoint open intervals. Furthermore, it follows from $[\omega_1((g')^{-1}_2(g'(x_0))),\omega_{1}(x_2)]=\Lambda^+\cup\Lambda^-\cup\Lambda^0$ that $\Lambda^0$ is a countable union of disjoint closed intervals. More specifically, the three sets $\Lambda^+$, $\Lambda^-$ and $\Lambda^0$ can be expressed as
\begin{align}
\Lambda^{+}=\bigcup_{j=1}^{\infty}(r_{j}^{-},r_{j}^{+}),\quad \Lambda^{-}=\bigcup_{j=1}^{\infty}(s_{j}^{-},s_{j}^{+}),\quad \Lambda^{0}=\bigcup_{j=1}^{\infty}[t_{j}^{-},t_{j}^{+}],\nonumber
\end{align}
where $(r_{j}^{-},r_{j}^{+},s_{j}^{-},s_{j}^{+}, t_{j}^{-},t_{j}^{+})_{j\geq 1}$ are real numbers satisfying $r_{j}^{+}\leq r_{j+1}^{-},s_{j}^{+}\leq s_{j+1}^{-}$ and $t_{j}^{+}\leq t_{j+1}^{-}$ for any $j\geq 1$. 

We next demonstrate that $\max _{(z_1,z_2)\in\mathcal{M}_{\zeta}}(z_2-z_1)$ and $\min_{(z_1,z_2)\in\mathcal{M}_{\zeta}}(z_2-z_1)$
are indeed increasing functions of $\beta$ on $[\omega_1((g')^{-1}_2(g'(x_0))),{\omega_{1}}(x_2)]$ under Case (a). Actually, by definition, one knows that if 
$\beta=\omega_1((g')^{-1}_2(g'(x_0)))$, then $g'(\omega^{-1}_1(\beta))=g'(x_0)<g'(\omega^{-1}_2(\beta))$, which yields $\omega_1((g')^{-1}_2(g'(x_0)))\in\Lambda^-$, implying that $\Lambda^-\neq \emptyset$.
Similarly, if $\beta=\omega_1(x_2)$, then $g'(\omega^{-1}_1(\beta))=g'(x_2)>g'(\omega^{-1}_2(\beta))$, which yields $\omega_1(x_2)\in\Lambda^+$, implying that $\Lambda^+\neq \emptyset.$ Hence, one has $\Lambda^0\neq\emptyset.$
Then by \eqref{eq:prop2.5:a}, a point $(z_1,z_2)\in\mathcal{M}_{\zeta}$ satisfies that 
$$z_2-z_1=\omega_{2}^{-1}(\beta)-(g^{\prime})_{3}^{-1}(g^{\prime}(\omega_{2}^{-1}(\beta))),$$
if $\beta\in(r^-_j,r^+_j)$ for some $j\geq1$;
and,
\begin{align}\label{eq:prop2.5:c}
z_2-z_1=\omega_{2}^{-1}(\beta)-(g^{\prime})_{3}^{-1}(g^{\prime}(\omega_{2}^{-1}(\beta))),
\end{align}
or
\begin{align}\label{eq:prop2.5:f}
z_2-z_1=\omega_1^{-1}(\beta)-(g^{\prime})_{1}^{-1}(g^{\prime}(\omega_1^{-1}(\beta))).
\end{align}
if $\beta\in[t^-_j,t^+_j]$ for some $j\geq1;$
and,
$$z_2-z_1=\omega_1^{-1}(\beta)-(g^{\prime})_{1}^{-1}(g^{\prime}(\omega_1^{-1}(\beta))),$$ if $\beta\in(s^-_j,s^+_j)$ for some $j\geq1$. 
Let us introduce two functions as
$$\varphi_1(u):=\omega_1((g')^{-1}_2(u))\quad\text{and}\quad\varphi_2(u):=\omega_2((g')^{-1}_4(u)),\text{ for }u\in[g'(x_0),g'(x_2)].$$
It can be checked that the functions $\varphi_1$ and $\varphi_2$ are strictly increasing on $[g'(x_0),g'(x_2)]$, {and inherit differentiability from $\omega_1$, $\omega_2$, $(g')_2^{-1}$ and $(g')_4^{-1}$.}
In addition, it holds that 
\begin{align}
[\omega_1((g')^{-1}_2(g'(x_0))),\omega_{1}(x_2)]&\subseteq [\varphi_1(g'(x_0)),\varphi_1(g'(x_2))],\nonumber\\
[\omega_1((g')^{-1}_2(g'(x_0))),\omega_{1}(x_2)]&\subseteq [\varphi_2(g'(x_0)),\varphi_2(g'(x_2))].\nonumber
\end{align}

If there is some $i,j,k\geq1$ such that $r^-_i<r^+_i=t^-_j=t^+_j=s^-_k<s^+_k$, 
it holds that
\begin{align}
\label{addnewa.33}
 g^{\prime}(\omega^{-1}_1(t^-_j))=g^{\prime}(\omega^{-1}_2(t^-_j)),
\end{align}
and, there exists some $u_0\in[g'(x_0),g'(x_2)]$ such that 
\begin{align}
\label{addnewa.34}
 \varphi_1(u_0)=t^-_j\,\,(\iff u_0=\varphi_1^{-1}(t^-_j)=g^{\prime}(\omega^{-1}_1(t^-_j))).
\end{align}
In addition, by the continuity and strict increasing property of $\varphi_{1}$, there exists small $\epsilon>0$ such that
\begin{align}
\varphi_{1}(u)\in (r^-_i,r^+_i), \quad \text{ for } u\in(u_0-\epsilon,u_0)\subseteq[g'(x_0),g'(x_2)],\nonumber
\end{align}
which implies that
\begin{align}
u=g^{\prime}(\omega_1^{-1}(\omega_1((g')^{-1}_2(u))))&=
g^{\prime}(\omega_1^{-1}(\varphi_{1}(u)))
\nonumber\\&
>g^{\prime}(\omega^{-1}_2(\varphi_{1}(u)))
=g^{\prime}(\omega^{-1}_2(\omega_1((g')^{-1}_2(u)))),\,\, \text{ for } u\in(u_0-\epsilon,u_0),\nonumber
\end{align}
which, by the increasing property of $(g^{\prime})^{-1}_{4}$, gives
\begin{align}
 (g^{\prime})^{-1}_{4}(u)>\omega^{-1}_2(\omega_1((g')^{-1}_2(u))),\quad \text{ for } u\in(u_0-\epsilon,u_0),\nonumber
\end{align}
which can be equivalently written as
\begin{align}
\label{addnewa.35}
 \omega_{2}((g^{\prime})^{-1}_{4}(u))>\omega_1((g')^{-1}_2(u)), \,\,\quad \text{ for } u\in(u_0-\epsilon,u_0).
\end{align}
Hence, by \eqref{addnewa.33}-\eqref{addnewa.35}, we have
\begin{align*}
\varphi^{\prime}_1(u_0)&=\lim_{u\uparrow u_0}\frac{\varphi_1(u)-\varphi_1(u_0)}{u-u_0}
=\lim_{u\rightarrow g'(\omega^{-1}_1(t^-_j))-}\frac{\omega_1((g')^{-1}_2(u))-\omega_1((g')^{-1}_2(g'(\omega^{-1}_1(t^-_j))))}{u-g'(\omega^{-1}_1(t^-_j))}
\nonumber\\
&=\lim_{u\rightarrow g'(\omega^{-1}_1(t^-_j))-}\frac{\omega_1((g')^{-1}_2(u))-t^-_j}{u-g'(\omega^{-1}_1(t^-_j))}
\geq \lim_{u\rightarrow g'(\omega^{-1}_2(t^-_j))-}\frac{\omega_2((g')^{-1}_4(u))-t^-_j}{u-g'(\omega^{-1}_2(t^-_j))}=\varphi'_2(u_0).\nonumber
\end{align*}
Then, by $\varphi^{\prime}_1(u_0)\geq\varphi^{\prime}_2(u_0)$, \eqref{addnewomega1}-\eqref{addnewomega2}, $(g')^{-1}_2(u_0)\in (a_1,(g')^{-1}_2(g^{\prime}(x_2)))$, and $(g')^{-1}_4(u_0)\in(x_0,x_{1})$, one obtains 
\begin{align}
\int_{(g')^{-1}_1(u_0)}^{(g')^{-1}_2(u_0)}g'(s)\ds\geq\int_{(g')^{-1}_3(u_0)}^{(g')^{-1}_4(u_0)}g'(s)\ds,\nonumber
\end{align}
which together with $\varphi_1(u_0)=t^-_j$ (that is, $u_0=g^{\prime}(\omega_1^{-1}(t^-_j))=g^{\prime}(\omega_2^{-1}(t^-_j))$) gives
\begin{align}\label{eq:prop2.5:d}
\int_{(g^{\prime})_{1}^{-1}(g^{\prime}(\omega_1^{-1}(t^-_j))}^{\omega_1^{-1}(t^-_j)}g'(s)\ds\geq\int_{(g^{\prime})_{3}^{-1}(g^{\prime}(\omega_2^{-1}(t^-_j))}^{\omega_2^{-1}(t^-_j)}g'(s)\ds.
\end{align}
In addition, since the two points $(z_1,z_2)=((g^{\prime})_{1}^{-1}(g^{\prime}(\omega_1^{-1}(t^-_j)),\omega_1^{-1}(t^-_j))$ (which satisfies \eqref{eq:prop2.5:f}) and $(z_1,z_2)=((g^{\prime})_{3}^{-1}(g^{\prime}(\omega_2^{-1}(t^-_j)),\omega_2^{-1}(t^-_j)) $ (which satisfies \eqref{eq:prop2.5:c}) are subject to $\psi(z_1,z_2)=\beta=t^-_j$, we have
\begin{align}
\int_{(g^{\prime})_{1}^{-1}(g^{\prime}(\omega_1^{-1}(t^-_j))}^{\omega_1^{-1}(t^-_j)}\left(1-\frac{g'(s)}{g'(\omega_1^{-1}(t^-_j))}\right)\ds
&= 
\int_{(g^{\prime})_{3}^{-1}(g^{\prime}(\omega_2^{-1}(t^-_j))}^{\omega_2^{-1}(t^-_j)}\left(1-\frac{g'(s)}{g'(\omega_2^{-1}(t^-_j))}\right)\ds,\nonumber
\end{align}
which, combined with \eqref{eq:prop2.5:d} and the fact that $g'(\omega^{-1}_1(t^-_j))=g'(\omega^{-1}_2(t^-_j))$, yields
\begin{align}
\label{eq:prop2.5:e}
\omega_{2}^{-1}(\beta)-(g^{\prime})_{3}^{-1}(g^{\prime}(\omega_{2}^{-1}(\beta)))\leq\omega_1^{-1}(\beta)-(g^{\prime})_{1}^{-1}(g^{\prime}(\omega_1^{-1}(\beta))),\quad \text{for } \beta=t^-_j.
\end{align}
Hence, the two functions $(r^-_j,s^+_j)\ni\beta\mapsto \max(\min)_{(z_1,z_2)\in\mathcal{M}_{\zeta}}(z_2-z_1)$ are increasing. 

If there is some $i,j,k\geq 1$ such that $r^-_i<r^+_i=t^-_j<t^+_j=s^-_k<s^+_k$, we have $g'(\omega^{-1}_1(\beta))=g'(\omega^{-1}_2(\beta))$ for any $\beta\in[t^-_j,t^+_j]$. Then, by a similar argument that is used in deriving \eqref{addnewa.35}, one can get
\begin{align}
 \omega_{2}((g^{\prime})^{-1}_{4}(u))=\omega_1((g')^{-1}_2(u)), \quad u\in[\varphi_1^{-1}(t^-_j),\varphi_1^{-1}(t^+_j)]. 
\end{align}
This implies that for any $\beta\in(t^-_j,t^+_j)$ and $v=\varphi_1^{-1}(\beta)=g^{\prime}(\omega^{-1}_1(\beta))$,
\begin{align}
\varphi'_1(v)&=\lim_{u\rightarrow g'(\omega^{-1}_1(\beta))}\frac{\omega_1((g')^{-1}_2(u))-\beta}{u-g'(\omega^{-1}_1(\beta))}= \lim_{u\rightarrow g'(\omega^{-1}_2(\beta))}\frac{\omega_2((g')^{-1}_4(u))-\beta}{u-g'(\omega^{-1}_2(\beta))}=\varphi'_2(v),\nonumber 
\end{align}
and, for $u_1:=\varphi^{-1}_1(t^-_j)$,
\begin{align}
\varphi'_1(u_1)&=\lim_{u\rightarrow g'(\omega^{-1}_1(t^-_j))+}\frac{\omega_1((g')^{-1}_2(u))-t^-_j}{u-g'(\omega^{-1}_1(t^-_j))}= \lim_{u\rightarrow g'(\omega^{-1}_2(t^-_j))+}\frac{\omega_2((g')^{-1}_4(u))-t^-_j}{u-g'(\omega^{-1}_2(t^-_j))}=\varphi'_2(u_1),\nonumber
\end{align}
and, for $u_2:=\varphi^{-1}_1(t^+_j)$,
\begin{align}
 \varphi'_1(u_2)&=\lim_{u\rightarrow g'(\omega^{-1}_1(t^+_j))-}\frac{\omega_1((g')^{-1}_2(u))-t^+_j}{u-g'(\omega^{-1}_1(t^+_j))}=\lim_{u\rightarrow g'(\omega^{-1}_2(t^+_j))-}\frac{\omega_2((g')^{-1}_4(u))-t^+_j}{u-g'(\omega^{-1}_2(t^+_j))}=\varphi'_2(u_2),\nonumber
\end{align}
In this way, it holds that $\varphi'_1(u)=\varphi'_2(u)$ for $u\in[\varphi^{-1}_1(t^-_j),\varphi^{-1}_1(t^+_j)]$,
which, combined with arguments similar to those leading to \eqref{eq:prop2.5:e}, imply
$$\omega_{2}^{-1}(\beta)-(g^{\prime})_{3}^{-1}(g^{\prime}(\omega_{2}^{-1}(\beta)))=\omega_1^{-1}(\beta)-(g^{\prime})_{1}^{-1}(g^{\prime}(\omega_1^{-1}(\beta)),\quad \beta\in[t^-_j,t^+_j].$$
Therefore, the two functions 
$\max(\min)_{(z_1,z_2)\in\mathcal{M}_{\zeta}}(z_2-z_1)$ is increasing in $\beta$ on $(r^-_j,s^+_j)$.

If there is some $i,j,k\geq 1$ such that $s_{k}^{-}< s_{k}^{+}=t_{j}^{-}\leq t_{j}^{+}=r_{i}^{-}<r_{i}^{+}$, one can derive the desired results by using similar arguments. In this vein, one obtains the increasing property of $\max(\min)_{(z_1,z_2)\in\mathcal{M}_{\zeta}}(z_2-z_1)$ in $\beta$ over the interval $[\omega_1((g')^{-1}_2(g'(x_0))),\omega_1(x_2)]$ under Case (a). 
The proof is complete.

\subsection{Proof of Proposition \ref{prop.3.4}}
\label{sec5.10}

By \eqref{theta12} and \eqref{def.a}, we have
\begin{align}
\label{limc+}
\lim_{\mu_-\rightarrow\infty}c_+
&=\lim_{\mu_-\rightarrow\infty}\frac{\theta_2^+-\theta_2^-}{\theta_2^++\theta_1^+}=\lim_{\mu_-\rightarrow\infty}\frac{\theta_2^+-\frac{\sqrt{\mu_-^2+2q\sigma_-^2}-\mu_-}{\sigma_-^2}}{\theta_2^++\theta_1^+}=\frac{\theta_2^+}{\theta_2^++\theta_1^+}>0,\\
\label{lima1}
\lim_{\mu_-\rightarrow\infty}a_1
&=
\lim_{\mu_-\rightarrow\infty}\frac{\ln{\frac{\sqrt{\mu_-^2+2q\sigma_-^2}+\mu_-}{\sqrt{\mu_-^2+2q\sigma_-^2}-\mu_-}}}{\frac{\sqrt{\mu_-^2+2q\sigma_-^2}}{\sigma_-^2}}=
\lim_{\mu_-\rightarrow\infty}\frac{2\ln\left(\sqrt{\mu_-^2+2q\sigma_-^2}+\mu_-\right)-\ln\left(2q\sigma_-^2\right)}{\frac{\sqrt{\mu_-^2+2q\sigma_-^2}}{\sigma_-^2}}=0,\\
\label{limamu}
\lim_{\mu_-\rightarrow\infty}a_2
&=
\lim_{\mu_-\rightarrow\infty}\frac{\ln\left(\theta_2^++\frac{\sqrt{\mu_-^2+2q\sigma_-^2}+\mu_-}{\sigma_-^2}\right)-\ln\left(\theta_2^+-\frac{\sqrt{\mu_-^2+2q\sigma_-^2}-\mu_-}{\sigma_-^2}\right)}{\frac{2\sqrt{\mu_-^2+2q\sigma_-^2}}{\sigma_-^2}}=0,
\end{align}
which combined with Proposition \ref{P.0.3} implies that the function $g(x)$ is concave on $(0,a_1)$ and convex on $(a_1,\infty)$ (since only the Case (iii) of Proposition \ref{P.0.3} can happen) if $\mu_-$ is large enough. 
Then, by the proof of Theorem \ref{thm2.1}, for any $(z_1,z_2)\in\mathcal{M}_{\zeta}\neq \emptyset$, we have $0\leq z_1<a_1<a-\beta$ (note that $\lim_{\mu_-\rightarrow\infty}a_1=0$ and $a-\beta>0$) if $\mu_-$ is large.

We next use proof by contradiction to show that, for any admissible strategy $(\Tilde{z}_1,\Tilde{z}_2)$ satisfying $\beta\leq \Tilde{z}_1+\beta<a< \Tilde{z}_2$, it holds that
$(\Tilde{z}_1,\Tilde{z}_2)\notin\mathcal{M}_{\zeta}$ as $\mu_-$ is large enough. Denote by $(z_1,z_2)$ an arbitrary admissible strategy such that $\beta\leq z_1+\beta< z_2\leq a$. Suppose $(\Tilde{z}_1,\Tilde{z}_2)\in\mathcal{M}_{\zeta}$. By Remark \ref{rem2.1} and \eqref{11}, we have
\begin{align}
\label{V/V}
&
\liminf_{\mu_-\rightarrow\infty}\frac{\zeta({z_1},{z_2})}{\zeta({\Tilde{z}_1},{\Tilde{z}_2)}}
=\liminf_{\mu_-\rightarrow\infty}\frac{g'_q(\Tilde{z}_2)(z_2-z_1-\beta)}{g_q(z_2)-g_q(z_1)}
\nonumber\\
=&
\liminf_{\mu_-\rightarrow\infty}\frac{\left[(1-c_{+})g^-_q(0)\theta_2^+e^{\theta_2^+{(\Tilde{z}_2-a)}}+(g^+_q(0)-c_+g^-_q(0))\theta_1^+e^{-\theta_1^+{(\Tilde{z}_2-a)}}\right](z_2-z_1-\beta)}{(g^-_q(0)-c_-g^+_q(0))(e^{\theta^-_2(z_2-a)}-e^{\theta^-_2(z_1-a)})-(1-c_-)g^+_q(0)(e^{-\theta^-_1(z_2-a)-e^{\theta^-_1(z_1-a)})}}
\nonumber\\
=&
\liminf_{\mu_-\rightarrow\infty}\left\{\frac{(c_-e^{-\theta_2^-a}+(1-c_-)e^{\theta_1^-a})\left[(1-c_+)\theta_2^{+}e^{\theta_2^+{(\Tilde{z}_2-a)}}-c_+\theta_1^+e^{-\theta_1^+(\Tilde{z}_2-a)}\right](z_2-z_1-\beta)}{(1-c_-)e^{(\theta_1^--\theta_2^-)a}(e^{\theta_2^-z_2}-e^{\theta_2^-z_1}-e^{-\theta_1^-z_2}+e^{-\theta_1^-z_1})}\right.
\nonumber\\
&
\hspace{1.5cm}\left.+\frac{\theta_1^+e^{-\theta_2^-a}e^{-\theta_1^+(\Tilde{z}_2-a)}(z_2-z_1-\beta)}{(1-c_-)e^{(\theta_1^--\theta_2^-)a}(e^{\theta_2^-z_2}-e^{\theta_2^-z_1}-e^{-\theta_1^-z_2}+e^{-\theta_1^-z_1})}\right\}
\nonumber\\
=&
\liminf_{\mu_-\rightarrow\infty}\left\{\frac{\left[(1-c_+)\theta_2^{+}e^{\theta_2^+{(\Tilde{z}_2-a)}}-c_+\theta_1^+e^{-\theta_1^+(\Tilde{z}_2-a)}\right](z_2-z_1-\beta)}{e^{-\theta_2^-a}(e^{\theta_2^-z_2}-e^{\theta_2^-z_1}-e^{-\theta_1^-z_2}+e^{-\theta_1^-z_1})}
\right.
\nonumber\\
&\hspace{1.5cm}+\frac{\theta_1^+e^{-\theta_1^+(\Tilde{z}_2-a)}(z_2-z_1-\beta)}{(1-c_-)e^{\theta_1^-a}(e^{\theta_2^-z_2}-e^{\theta_2^-z_1}-e^{-\theta_1^-z_2}+e^{-\theta_1^-z_1})}
\nonumber\\
&\hspace{1.5cm}
\left.+\frac{c_-\left[(1-c_+)\theta_2^{+}e^{\theta_2^+{(\Tilde{z}_2-a)}}-c_+\theta_1^+e^{-\theta_1^+(\Tilde{z}_2-a)}\right](z_2-z_1-\beta)}{(1-c_-)e^{\theta_1^-a}(e^{\theta_2^-z_2}-e^{\theta_2^-z_1}-e^{-\theta_1^-z_2}+e^{-\theta_1^-z_1})}\right\}. 
\end{align}
Recall that 
\begin{align}\label{eq:prop3.2:a}
(1-c_-)e^{(\theta_1^--\theta_2^-)a}(e^{\theta_2^-z_2}-e^{\theta_2^-z_1}-e^{-\theta_1^-z_2}+e^{-\theta_1^-z_1})=g(z_2)-g(z_1)>0,
\end{align}
which implies that the denominators of three fractions on the right hand side of \eqref{V/V} are positive.
In view of \eqref{limc+} as well as the facts that $\lim_{\mu_-\rightarrow\infty}\theta_2^-=0$ and $\lim_{\mu_-\rightarrow\infty}\theta_1^-=\infty$, we have
\begin{align}
\label{eq:prop3.2:d}
\lim_{\mu_-\rightarrow\infty}\frac{\left[(1-c_+)\theta_2^{+}e^{\theta_2^+{(\Tilde{z}_2-a)}}-c_+\theta_1^+e^{-\theta_1^+(\Tilde{z}_2-a)}\right](z_2-z_1-\beta)}{e^{-\theta_2^-a}(e^{\theta_2^-z_2}-e^{\theta_2^-z_1}-e^{-\theta_1^-z_2}+e^{-\theta_1^-z_1})}
=\infty,
\end{align}
where we have used the fact that
\begin{align}
&\lim_{\mu_-\rightarrow\infty}\left[(1-c_+)\theta_2^{+}e^{\theta_2^+{(\Tilde{z}_2-a)}}-c_+\theta_1^+e^{-\theta_1^+(\Tilde{z}_2-a)}\right](z_2-z_1-\beta)
\nonumber\\
=&\left[(1-\frac{\theta_2^+}{\theta_2^++\theta_1^+})\theta_2^{+}e^{\theta_2^+{(\Tilde{z}_2-a)}}-\frac{\theta_2^+}{\theta_2^++\theta_1^+}\theta_1^+e^{-\theta_1^+(\Tilde{z}_2-a)}\right](z_2-z_1-\beta)
\nonumber\\
>&\left[(1-\frac{\theta_2^+}{\theta_2^++\theta_1^+})\theta_2^{+}-\frac{\theta_2^+}{\theta_2^++\theta_1^+}\theta_1^+\right](z_2-z_1-\beta)=0.\nonumber
\end{align}
In addition, it follows from \eqref{eq:prop3.2:a} and $\theta^+_1>0$ that
\begin{align}\label{eq:prop3.2:b}
\frac{\theta_1^+e^{-\theta_1^+(\Tilde{z}_2-a)}(z_2-z_1-\beta)}{(1-c_-)e^{\theta_1^-a}(e^{\theta_2^-z_2}-e^{\theta_2^-z_1}-e^{-\theta_1^-z_2}+e^{-\theta_1^-z_1})}>0, \quad \text{ for all }\,\mu_{-}\in\mathbb{R}_{+}.
\end{align}
We also note that $\lim_{\mu_-\rightarrow\infty}c_-=1$ due to $\lim_{\mu_-\rightarrow\infty}\theta^-_1=\infty$ and $\lim_{\mu_-\rightarrow\infty}\theta^-_2=0$. Hence, there exists $K_{0}\in\mathbb{R}_{+}$ such that $c_{-}>0$ for all $\mu_{-}\in [K_{0},\infty)$. This, together with $\theta^+_{i=1,2}>0$, \eqref{limc+} and \eqref{eq:prop3.2:a}, implies
\begin{align}\label{eq:prop3.2:c}
&
\frac{c_-\left[(1-c_+)\theta_2^{+}e^{\theta_2^+{(\Tilde{z}_2-a)}}-c_+\theta_1^+e^{-\theta_1^+(\Tilde{z}_2-a)}\right](z_2-z_1-\beta)}{(1-c_-)e^{\theta_1^-a}(e^{\theta_2^-z_2}-e^{\theta_2^-z_1}-e^{-\theta_1^-z_2}+e^{-\theta_1^-z_1})}
\nonumber\\
&\geq
\frac{c_-\left[(1-c_+)\theta_2^{+}-c_+\theta_1^+\right](z_2-z_1-\beta)}{(1-c_-)e^{\theta_1^-a}(e^{\theta_2^-z_2}-e^{\theta_2^-z_1}-e^{-\theta_1^-z_2}+e^{-\theta_1^-z_1})}=0,\quad \text{ for all }\, \mu_{-}\in [K_{0},\infty).
\end{align}
Hence, putting together with \eqref{V/V} and \eqref{eq:prop3.2:d}-\eqref{eq:prop3.2:c}, we obtain
\begin{align}
\liminf_{\mu_-\rightarrow\infty}\frac{\zeta(z_1,z_2)}{\zeta(\Tilde{z}_1,\Tilde{z}_2)}=\infty.
\end{align}
Then, there exists a constant $K>K_{0}$ such that, when $\mu_->K$, we have that $\zeta({z_1},{z_2})>\zeta({\Tilde{z}_1},{\Tilde{z}_2})$, which contradicts $(\Tilde{z}_1,\Tilde{z}_2)\in\mathcal{M}_{\zeta}$. Therefore, when $\mu_->K$, any admissible strategy $(\Tilde{z}_1,\Tilde{z}_2)$ satisfying $\beta\leq \Tilde{z}_1+\beta<a< \Tilde{z}_2$ can not be the global maximizer of $\zeta$, i.e., $(\Tilde{z}_1,\Tilde{z}_2)\notin\mathcal{M}_{\zeta}$. This 
yields that $\beta \leq z_1+\beta<z_2\leq a$ for any $(z_1,z_2)\in\mathcal{M}_{\zeta}\neq \emptyset$ when $\mu_->K$. The proof is complete.

 \section*{Acknowledgments}
The authors thank the editors and anonymous reviewers for their insightful and constructive feedback, which has significantly contributed to the improvement of this manuscript.


\begin{thebibliography}{99}
\bibitem{AT09}
\small{Albrecher, H. and Thonhauser, S. (2009). Optimality results for dividend problems in insurance. \emph{Revista de la Real Academia de Ciencias Exactas, Fisicas y Naturales. Serie A.
Matematicas}, {\bf103(2)}, 295-320.}

\bibitem{AngBeK02b} {\small Ang, A. and Bekaert, G. (2002). International asset allocation with regime shifts. \emph{Review of Financial Studies}, {\bf15}, 1137-1187.}

\bibitem{AngTim} {\small Ang, A. and Timmermann, A. (2012). Regime changes and financial markets. \emph{Annual Review of Financial Economics}, {\bf4}, 313-337.}

\bibitem{Azcue15} {\small Azcue, P. and Muler, N. (2015). Optimal dividend payment and regime switching in a compound Poisson risk model. \emph{SIAM J Control Optim}, {\bf53(5)}, 3270-3298.}

\bibitem{BaiGuo} {\small Bai, L. and Guo, J. (2010).
Optimal dividend payments in the classical risk model when
payments are subject to both transaction costs and taxes. \emph{Scandinavian Actuarial Journal}, {\bf 1}, 36-55. 
}

\bibitem{Bai10} {\small Bai, L. and Paulsen, J. (2010). Optimal dividend policies with transaction costs for a class of diffusion processes. \emph{SIAM J Control Optim}, {\bf48(8)}, 4987-5008.}




\bibitem{BK14}
\small{Bayraktar, E., Kyprianou, A.E. and Yamazaki, K. (2014). Optimal dividends in the dual model under transaction costs. \emph{Insurance: Mathematics and Economics}, {\bf54}, 133-143.}

\bibitem{Benes1980}
\small{Benes, V.E., Shepp, L. A. and Witsenhausen, H. S. (1980). Some solvable stochastic control problems. \emph{Stochastics} {\bf 4}, 134-160.}


\bibitem{BANPS2015}
{\small
Borodin, A. and Salminen, P. (2015). \emph{Handbook of Brownian motion-facts and formulae}. {Springer Science and Business Media}, 251-330.}


\bibitem{Cade06}
\small{Cadenillas, A., Choulli, T., Taksar, M. and Zhang, L., (2006). Classical and impulse stochastic control for the optimization of the dividend and risk policies of an insurance firm. \emph{Mathematical Finance}, {\bf 16(1)}, 181–202.}

\bibitem{Chen22}
\small{Chen, Z. and Epstein L.G. (2022). A central limit theorem for sets of probability measures. \emph{Stochastic Processes and their Applications}, {\bf 152}, 424-451.}

\bibitem{Chen23}
\small{Chen, Z., Epstein L.G. and Zhang, G. (2023). A central limit theorem, loss aversion and multi-armed bandits. \emph{Journal of Economic Theory}, {\bf 209}, 105645.}


\bibitem{Chen}
\small{Chen, Z., Wu, P. and Zhou, X. (2025+). Optimal state equation for the control of a diffusion with two distinct dynamics. IEEE Transactions on Automatic Control, doi: 10.1109/TAC.2025.3610151.}

\bibitem{deFinetti}
\small{De Finetti, B. (1957). Su un’impostazione alternativa della teoria collettiva del rischio. In: Transactions of the XVth International Congress of Actuaries, {\bf 2}, 433–443.}

\bibitem{Gairat}
\small{Gairat, A. and Shcherbakov, V. (2017). Density of skew Brownian motion and its functionals with application in finance. \emph{Mathematical Finance}, {\bf 27(4)}, 1069-1088.
}

\bibitem{Gerber}
\small{Gerber, H.U. and Shiu, E.S.W. (2006). On optimal dividends: From reflection to refraction. \emph{Journal of Computational and Applied Mathematics}, {\bf 186(1)}, 4–22.}


\bibitem{Hunting}
\small{Hunting, H. and Paulsen, J. (2013). Optimal dividend policies with transaction costs for a class of jump-diffusion processes. \emph{Finance and Stochastics}, {\bf 17}, 73-106. }


\bibitem{Jacod03} {\small
Jacod, J. and Shiryaev, A. (2003). Limit Theorems for Stochastic Processes, 2nd edn. Springer, Berlin.}

\bibitem{JiaPis2012} {\small Jiang, Z. and Pistorius, M. R. (2012). Optimal dividend distribution under Markov regime
switching. \emph{ Finance and Stochastics}. {\bf16}, 449-476.}

\bibitem{Karatzas1984} {\small Karatzas, I and Shreve, S.E. (1984). Trivariate density of Brownian motion, its local and occupation times, with application to stochastic control. \emph{Annals of Applied Probability} {\bf12}, 819-828.}


\bibitem{Keilson1978} {\small Keilson, J. and Wellner, J.A. (1978). Oscillating Brownian motion. \emph{Journal of Applied Probability}. {\bf 15(2)}, 300-310. }




\bibitem{LeGall84}
\small{
Le Gall, J.F. (1984). One-dimensional stochastic differential equations involving the local times of the unknown process. In: Truman, A., Williams, D. (eds) Stochastic Analysis and Applications. Lecture Notes in Mathematics, vol 1095. Springer, Berlin, Heidelberg. https://doi.org/10.1007/BFb0099122.}


\bibitem{McNamara}
\small{McNamara, J. M., (1984). Control of a diffusion by switching between two drift-diffusion coefficient pairs. \emph{SIAM J Control Optim}, {\bf 22(1)}, 87-94.}

\bibitem{Paulsen}
\small{Paulsen, J., (2008). Optimal dividend payments and reinvestments of diffusion processes with both fixed and proportional costs. \emph{SIAM J Control Optim}, {\bf 47(5)}, 2201-2226}. 


\bibitem{Pelle} {\small Pelletier, D. (2006). Regime switching for dynamic correlations. \emph{Journal of Econometrics}, {\bf131}, 445-473.}

\bibitem{PL}
{\small Plotter, P.E. (2003). Stochastic integration and differential equation. Stochastic Modeling and Applied Probability, 21.}

\bibitem{So} {\small So, E.C.P., Lam, K. and Li, W.K. (1998). A stochastic volatility model with Markov switching. \emph{Journal of Business $\&$ Economic Statistics}, {\bf16}, 244-253.}


\bibitem{WY21}
\small{Wang, W., Yu, X., and Zhou, X., (2024). On optimality of barrier dividend control under endogenous regime switching with application to Chapter 11 bankruptcy. \emph{Applied Mathematics and Optimization}, 89, 13.}

\bibitem{WeiWY} {\small Wei, J., Wang, R. and Yang, H. (2016). On the optimal dividend strategy in a regime-switching diffusion model. \emph{Advances in Applied Probability}, {\bf44}, 886-906.}


\end{thebibliography}
\end{document}